\documentclass[preprint,12pt]{elsarticle}

\usepackage{amssymb}
\usepackage{amsthm}

\usepackage{amssymb, amsmath, mathrsfs}
\usepackage{synttree, bussproofs,scrextend, stmaryrd, rotating, xcolor, algorithm2e, upgreek}
\usepackage{esvect,pgf, tikz, color}
\usetikzlibrary{arrows, automata}
\usepackage[all]{xy}
\usepackage{changepage,enumerate,proof,upgreek,relsize,trfsigns,
comment}
\usepackage{multicol}

  \newtheorem{innercustomthm}{Theorem}
\newenvironment{customthm}[1]
  {\renewcommand\theinnercustomthm{#1}\innercustomthm}
  {\endinnercustomthm}

\DeclareSymbolFont{extraup}{U}{zavm}{m}{n}
\DeclareMathSymbol{\vardiamond}{\mathalpha}{extraup}{87}

\providecommand{\funding}[1]{\textbf{Funding.} #1}

\newtheorem{theorem}{Theorem}
\newtheorem{lemma}[theorem]{Lemma}
\newtheorem{corollary}[theorem]{Corollary}
\newtheorem{remark}[theorem]{Remark}
\newtheorem{definition}[theorem]{Definition}

\newtheorem{example}[theorem]{Example}

\newcommand{\ntr}{\mathfrak{N}}
\newcommand{\ltr}{\mathfrak{L}}
\newcommand{\seqcomp}{\otimes}

\newcommand{\ddn}[1]{#1^{N}}

\definecolor{tim}{RGB}{0, 0, 250}

\newcommand{\ifandonlyif}{\textit{iff} }
\newcommand{\iffi}{\textit{iff} }
\newcommand{\etc}{$\ldots$ }
\newcommand{\dfn}{Definition}
\newcommand{\fig}{Figure}
\newcommand{\lem}{Lemma}
\newcommand{\thm}{Theorem}

\newcommand{\cor}{Corollary}

\newcommand{\sect}{Section}
\newcommand{\cptr}{Ch.}

\newcommand{\D}{\mathrm{D}_{\charx}}
\newcommand{\T}{\mathrm{T}_{\charx}}
\newcommand{\four}{\mathrm{4}_{\charx}}
\newcommand{\B}{\mathrm{B}_{\charx}}
\newcommand{\five}{\mathrm{5}_{\charx}}
\newcommand{\ipa}{\mathrm{IPA}}

\newcommand{\bl}{[}
\newcommand{\br}{]}
\newcommand{\sbl}{\{}
\newcommand{\sbr}{\}}

\newcommand{\nika}{\mathsf{NIK}(\axs)}
\newcommand{\nikam}{\mathsf{NIK}(\axsii)}

\newcommand{\calc}{\mathbf{L}_{\albet}(\axs)}
\newcommand{\rcalc}{\mathbf{L}_{\albet}^{*}(\axs)}%{\ikm(\albet,\axs)\mathbf{L}}
\newcommand{\ncalc}{\mathbf{N}_{\albet}^{*}(\axs)}

\newcommand{\rccalc}{\mathbf{L}^{\mathrm{C}}_{\albet}(\axs)}

\newcommand{\inp}{\bullet}
\newcommand{\outp}{\circ}

\newcommand{\ns}{\Sigma}
\newcommand{\empseq}{\emptyset}

\newcommand{\botin}{(\bot^{\inp})}
\newcommand{\botout}{(\bot^{\outp})}
\newcommand{\conin}{(\land^{\inp})}
\newcommand{\conout}{(\land^{\outp})}
\newcommand{\disin}{(\lor^{\inp})}
\newcommand{\disout}{(\lor^{\outp})}
\newcommand{\iimpin}{(\iimp^{\inp})}
\newcommand{\iimpout}{(\iimp^{\outp})}

\newcommand{\xboxout}{(\xbox^{\outp})}
\newcommand{\yboxout}{(\ybox^{\outp})}

\newcommand{\xdiain}{(\xdia^{\inp})}
\newcommand{\ydiain}{(\ydia^{\inp})}

\newcommand{\sar}{\vdash}

\newcommand{\lseq}{\Lambda}

\newcommand{\der}{\vdash_{\axs}^{\albet}}
\newcommand{\ent}{\Vdash_{\axs}^{\albet}}
\newcommand{\sat}{\Vdash^{\albet}}

\newcommand{\h}{\mathsf{H}}
\newcommand{\axd}{\text{D}}

\newcommand{\axhsl}{\text{HSL}}

\newcommand{\prf}{\mathcal{D}}
            
\newcommand{\km}{\mathsf{K_{m}}}

\newcommand{\ikm}{\mathsf{IK_{m}}}
\newcommand{\ik}{\mathsf{IK}}
\newcommand{\ikt}{\mathsf{IKt}}
\newcommand{\ikma}{\mathsf{IK_{m}}(\albet,\mathcal{A})}
\newcommand{\ika}{\mathsf{IK}(\mathcal{A})}
\newcommand{\ikam}{\mathsf{IK}(\mathcal{B})}

\newcommand{\cate}{}
\newcommand{\concat}{\cdot}

\newcommand{\g}[1]{g(#1)}

\newcommand{\glang}{L_{\g{\axs}}}

\newcommand{\pto}{\longrightarrow}

\newcommand{\cfg}{G}

\newcommand{\prgr}[1]{PG(#1)}     
\newcommand{\prgrdom}{V}
\newcommand{\prgredges}{E}    
\newcommand{\empstr}{\varepsilon} 
\newcommand{\emppath}{\varepsilon}%{\lambda} 
\newcommand{\ppath}{\pi}

\newcommand{\charx}{x}
\newcommand{\chary}{y}
\newcommand{\charz}{z}
\newcommand{\chara}{a}
\newcommand{\charb}{b}
\newcommand{\charc}{c}
\newcommand{\stra}{s}
\newcommand{\strb}{t}
\newcommand{\strc}{r}
\newcommand{\conv}[1]{\overline{#1}}
\newcommand{\albetstr}{\albet^{\ast}}

\newcommand{\strabox}{[\stra]}
\newcommand{\stradia}{\langle \stra \rangle}

\newcommand{\osdr}{\pto_{\g{\axs}}}
\newcommand{\dr}{\pto_{\g{\axs}}^{*}}

\newcommand{\iimp}{\supset}

\newcommand{\ieq}{\equiv}
\newcommand{\inot}{\neg}

\newcommand{\xbox}{[\charx]}
\newcommand{\ybox}{[\chary]}
\newcommand{\xdia}{\langle \charx \rangle}
\newcommand{\ydia}{\langle \chary \rangle}
\newcommand{\xboxc}{[\conv{\charx}]}
\newcommand{\xdiac}{\langle \conv{\charx} \rangle}

\newcommand{\abox}{[\chara]}
\newcommand{\adia}{\langle \chara \rangle}

\newcommand{\dia}{\Diamond}
\newcommand{\diap}[1]{\langle #1 \rangle}
\newcommand{\boxp}[1]{[ #1 ]}

\newcommand{\albet}{\Upsigma}
\newcommand{\albetf}{\albet^{+}}
\newcommand{\albetb}{\albet^{-}}
\newcommand{\lang}[1]{\mathcal{L}(#1)}
\newcommand{\langc}[1]{\mathcal{L}^{\mathrm{C}}(#1)}

\newcommand{\prop}{\Upphi}

\newcommand{\sufo}[1]{\mathrm{S}(#1)}

\newcommand{\axs}{\mathcal{A}}
\newcommand{\axsii}{\mathcal{B}}

\newcommand{\lab}{\mathrm{Lab}}
\newcommand{\fseti}{\mathscr{A}}

\newcommand{\id}{(id)}

\newcommand{\botl}{(\bot_{l})}
\newcommand{\botr}{(\bot_{r})}

\newcommand{\disr}{(\lor_{r})}
\newcommand{\conr}{(\land_{r})}

\newcommand{\ddr}{(d_{\charx})}

\newcommand{\ipar}{(i^{\stra}_{\charx})}
\newcommand{\convr}{(c_{\charx})}

\newcommand{\xboxr}{(\xbox_{r})}
\newcommand{\xboxl}{(\xbox_{l})}
\newcommand{\xdiar}{(\xdia_{r})}
\newcommand{\xdial}{(\xdia_{l})}

\newcommand{\iimpl}{(\iimp_{l})}
\newcommand{\iimpr}{(\iimp_{r})}
\newcommand{\disl}{(\lor_{l})}
\newcommand{\conl}{(\land_{l})}

\newcommand{\prdiai}{(p_{\xdia}^{1})}
\newcommand{\prdiaii}{(p_{\xdia}^{2})}
\newcommand{\prboxi}{(p_{\xbox}^{1})}
\newcommand{\prboxii}{(p_{\xbox}^{2})}

\newcommand{\wk}{(w)}
\newcommand{\wkl}{(w_{l})}
\newcommand{\wkr}{(w_{r})}
\newcommand{\nec}{(n_{\charx})}
\newcommand{\med}{(m_{\charx})}
\newcommand{\ctr}{(c)}
\newcommand{\ctrl}{(ctr_{l})}
\newcommand{\ctrli}{(ctr_{l1})}
\newcommand{\ctrlii}{(ctr_{l2})}

\newcommand{\refi}{(r_{\xdia})}
\newcommand{\refii}{(r_{\xbox})}

\newcommand{\cut}{(cut)}
\newcommand{\ccut}{(cut_{\mathrm{C}})}

\newcommand{\sub}{(s)}

\newcommand{\disru}{(\lor)}
\newcommand{\conru}{(\land)}
\newcommand{\idru}{(id)}
\newcommand{\xboxru}{(\xbox)}
\newcommand{\xdiaru}{(\xdia)}

\newcommand{\prbox}{(p_{\xbox})}
\newcommand{\prdia}{(p_{\xdia})}

\newcommand{\rel}{\mathcal{R}}

\newcommand{\md}[1]{\mathrm{m_{d}}(#1)}

\journal{Logic Journal of the IGPL} %Theoretical Computer Science}

\begin{document}

\begin{frontmatter}



\title{Nested Sequents for Intuitionistic Grammar Logics via Structural Refinement}


\author{Tim S. Lyon} %\orcidID{0000-0003-3214-0828}
\ead{timothy_stephen.lyon@tu-dresden.de}
\ead[url]{https://sites.google.com/view/timlyon}

%Technische Universität Dresden
\affiliation{organization={Institute of Artificial Intelligence, Technische Universität Dresden},%Department and Organization
            addressline={Nöthnitzerstra{\ss}e 46}, 
            city={Dresden},
            postcode={01187}, 
            state={Saxony},
            country={Germany}}

\begin{abstract}
Intuitionistic grammar logics fuse constructive and multi-modal reasoning while permitting the use of converse modalities, serving as a generalization of standard intuitionistic modal logics. In this paper, we provide definitions of these logics as well as establish a suitable proof theory thereof. In particular, we show how to apply the structural refinement methodology to extract cut-free nested sequent calculi for intuitionistic grammar logics from their semantics. This method proceeds by first transforming the semantics of these logics into sound and complete labeled sequent systems, which we prove have favorable proof-theoretic properties such as syntactic cut-elimination. We then transform these labeled systems into nested sequent systems via the introduction of propagation rules and the elimination of structural rules. Our derived proof systems are then put to use, whereby we prove the conservativity of intuitionistic grammar logics over their modal counterparts, establish the general undecidability of these logics, and recognize a decidable subclass, referred to as \emph{simple} intuitionistic grammar logics.
\end{abstract}




\end{frontmatter}


\section{Introduction}\label{sec:introduction}
%1.Discuss Narrow vs. General Structural Refinement
%2. Provide first cut elimination theorem in system for Simpson's systems
%3. Defines intuitionistic grammar logics
%4. First (cut-free) labelled and nested systems for intuitionistic grammar logics
%5. First undecidability results for int grammar logics?
%6. Proves cut-elim, hp-admiss, and hp-invert results for calculi
%7. Can constrast my method with Alwen's using display
%8. All results, lemmas, etc. generalized to grammar case from mono-modal case

%%%Consistency Material
%1. Labelled -> Labeled
%2. qed symbols at end of proofs

 \emph{Grammar logics} form a prominent class of normal, multi-modal logics~\cite{CerPen88} extending classical propositional logic with a set of modalities indexed by characters from a given alphabet. These logics obtain their name on the basis of their relationship to context-free grammars. In particular, grammar logics incorporate axioms which may be viewed as production rules in a context-free grammar, and which generate sequences of edges indexed with characters (and thus may be viewed as words) in corresponding relational models. Due to the generality of this class of logics, it has been found that this class includes many well-known and useful logics such as description logics~\cite{HorSat04}, epistemic logics~\cite{FagMosHalVar95}, information logics~\cite{Vak86}, temporal logics~\cite{CerHer95}, and standard modal logics (e.g. $\mathsf{K}$, $\mathsf{S4}$, and $\mathsf{S5}$)~\cite{DemNiv05}. Despite the %is wide array of logics capturing various modes of reasoning, 
 various modes of reasoning offered within the class of grammar logics, such logics are nevertheless classical at their core, being defined atop classical propositional logic. Thus, it is interesting to place such logics on an intuitionistic, rather than a classical, footing.
 
 \emph{Intuitionistic logic} is one of the most eminent formulations of constructive reasoning, that is, reasoning where the claimed existence of an object implies its constructibility~\cite{Bro75}. Resting on the philosophical work of L.E.J. Brouwer, propositional intuitionistic logic was axiomatized in the early 20\textsuperscript{th} century by Kolmogorov~\cite{Kol25}, Orlov~\cite{Orl28}, and Glivenko~\cite{Gli29}, with a first-order axiomatization given by Heyting~\cite{Hey30}. Rather naturally, as the paradigm of intuitionistic reasoning evolved, it was eventually integrated with the paradigm of modal reasoning, begetting so-called \emph{intuitionistic modal logics}. 
 
 A plethora of intuitionistic (and constructive) modal logics have been proposed in the literature~\cite{BiePai00,BovDov84,Dov85,Fit48,PloSti86,Ser84,Sim94}, though the logics introduced by Fischer-Servi~\cite{Ser84}, and Plotkin and Stirling~\cite{PloSti86}, have become (most notably through the work of Simpson~\cite{Sim94}) one of the most popular formulations. %In the same year that Plotkin and Stirling~\cite{PloSti86} introduced their intuitionistic modal logics,
 Around the same time, Ewald introduced \emph{intuitionistic tense logic}~\cite{Ewa86}, which not only includes modalities that make reference to future states in a relational model, but also includes modalities that make reference to past states. As with (multi-)modal and intuitionistic logics, intuitionistic modal logics have proven useful in computer science; e.g. such logics have been used to design verification techniques~\cite{FaiMen95}, in reasoning about functional programs~\cite{Pit91}, and in the definition of programming languages~\cite{DavPfe01}. Continuing in the same vein, the first contribution of this paper is to generalize both the intuitionistic mono-modal logics of Fischer-Servi~\cite{Ser84} as well as Plotkin and Stirling~\cite{PloSti86}, and Ewald's intuitionistic tense logic, giving rise to the class of \emph{intuitionistic grammar logics}. We provide a semantics for this class of logics, along with axiomatizations, and confirm soundness and completeness thereof. 
 
 Beyond defining the class of intuitionistic grammar logics, we also provide a suitable proof theory for this class of logics. As with any logic, proof calculi are indispensable for facilitating reasoning and prove valuable in establishing non-trivial properties of logics. A prominent formalism that has arisen within the last 30 or so years for formalizing reasoning for modal and intuitionistic logics is the \emph{nested sequent} formalism. Initiated by Bull~\cite{Bul92} and Kashima~\cite{Kas94}, nested sequent systems perform reasoning over trees of (pairs of) multisets of formulae, proving worthwhile in developing automated reasoning techniques for logics. Such systems have been used to write decision algorithms for logics supporting automated counter-model extraction~\cite{GalSal15,TiuIanGor12}, have been employed in constructive proofs of interpolation~\cite{FitKuz15,LyoTiuGorClo20}, and have even been applied in knowledge integration scenarios~\cite{LyoGom22}. 
 
 The value of nested sequent systems comes from the fact that such systems normally exhibit a collection of useful proof-theoretic properties simultaneously. As a first example, the well-known cut rule (cf.~\cite{Gen35a,Gen35b}), which serves as a generalization of \emph{modus ponens}, and is frequently employed in proofs of completeness, tends to be admissible (i.e. redundant) in nested sequent calculi. %; e.g.~\cite{Kas94,Bru09,Pog09,TiuIanGor12,Str13,Fit14,Lyo21a,LyoGom22}. 
 As the cut rule deletes formulae from the premises to the conclusion when applied, the rule is unsuitable for bottom-up proof-search as the deleted formula must be guessed during proof-search, which can obstruct proofs of termination. Since cut is not required to appear in most nested sequent calculi, such calculi tend to be \emph{analytic}, i.e. one can observe that any formula occurring in a proof occurs as a subformula of the conclusion, which bounds the space of possible proofs of a given formula, proving beneficial in establishing decidability~\cite{Bru09,LyoBer19}. Still, nested sequent systems possess favorable properties beyond analyticity. For instance, such systems reason within the economical structure of trees, easing proofs of termination of associated proof-search algorithms, such systems permit the admissibility of useful structural rules (e.g. contractions and weakenings), and such systems tend to have invertible rules, which is helpful in extracting counter-models from failed proof-search.
 
 Recently, the \emph{structural refinement} methodology was developed as a means of generating nested sequent systems for diverse classes of logics~\cite{Lyo21thesis}. The methodology exploits the formalism of \emph{labeled sequents} whereby calculi are obtained by transforming the semantics of a logic into inference rules~\cite{Sim94,Vig00}. As such, labeled sequent systems perform reasoning within the semantics of a logic, and reason over structures closely resembling a logic's models, with sequents encoding arbitrary graphs of (pairs of) multisets of formulae (thus generalizing the data structure used in nested sequents). A nice feature of labeled sequent systems is that general results exist for their construction~\cite{CiaMafSpe13,Sim94} and such systems tend to exhibit desirable properties such as cut-elimination, admissibility of certain structural rules, and invertibility of rules~\cite{Neg05}. However, labeled systems have a variety of drawbacks as such systems typically involve superfluous structures in sequents, yielding larger proofs than necessary, making proof-search algorithms less efficient, and obfuscating termination proofs of associated proof-search algorithms (cf.~\cite{Lyo21thesis}).
 
 To circumvent the drawbacks of labeled sequent systems, the structural refinement methodology (in a broad sense) leverages the general construction techniques in the labeled setting to extract labeled systems from a logic's semantics, and then systematically transforms these systems into nested systems, which are more economical and better suited to applied scenarios. In a narrow sense, structural refinement consists of transforming a labeled sequent system into a nested system through the introduction of \emph{propagation rules} (cf.~\cite{CasCerGasHer97,Fit72}) or \emph{reachability rules}~\cite{Lyo21thesis,Lyo22} and the elimination of structural rules, followed by a notational translation. The propagation rules operate by viewing labeled sequents (which encode binary labeled graphs) as automata, allowing for formulae to be propagated along a path in the underlying graph of a labeled sequent, so long as the path is encoded by a string derivable in a certain formal grammar. The refinement methodology grew out of works relating labeled systems to `more refined' or nested systems~\cite{CiaLyoRam18,GorRam12,LyoBer19}. The propagation rules we use are largely based upon the work of~\cite{GorPosTiu11,TiuIanGor12}, where such rules were used in the setting of display and nested calculi. These rules were then transported to the labeled setting to prove the decidability of agency (STIT) logics~\cite{LyoBer19}, to establish translations between calculi within various proof-theoretic formalisms~\cite{CiaLyoRamTiu21}, and to provide a basis for the structural refinement methodology~\cite{Lyo21}. In this paper, we apply this methodology in the setting of intuitionistic grammar logics, obtaining analytic nested systems for these logics, which are then put to use to establish conservativity and (un)decidability results.
 
 This paper accomplishes the following: In \sect~\ref{sec:log-prelims}, we define intuitionistic grammar logics, providing a semantics, axiomatizations, and confirming soundness and completeness results. We also introduce the grammar-theoretic foundations necessary to define propagation rules. In \sect~\ref{sec:labeled-systems}, we define labeled sequent calculi %that generalize Simpson's labeled sequent calculi for intuitionistic modal logics~\cite{Sim94}, 
 for intuitionistic grammar logics (which generalize Simpson's labeled systems for intuitionistic modal logics characterized by Horn properties~\cite{Sim94}) and provide admissibility, invertibility, and cut-elimination results, which also establishes syntactic cut-elimination for the labeled systems of Simpson mentioned above. In Sections~\ref{sec:refinement} and~\ref{sec:nested-calculi}, we demonstrate how to apply the structural refinement method to extract `refined' labeled and analytic nested sequent systems for intuitionistic grammar logics. In \sect~\ref{sec:(un)decid}, we leverage our nested and refined labeled systems to show three results: (1) Conservativity: we prove that intuitionistic grammar logics are conservative over intuitionistic modal logics, (2) Undecidability: we prove the undecidability of determining if a formula is a theorem of an arbitrary intuitionistic grammar logic by giving a proof-theoretic reduction of the problem from classical context-free grammar logics, and (3) Decidability: we adapt a method due to Simpson~\cite{Sim94} to our setting, recognizing that a subclass of intuitionistic grammar logics, referred to \emph{simple}, is decidable. Last, \sect~\ref{sec:conclusion} concludes and discusses future work.

This paper serves as a journal version extending the conference papers~\cite{Lyo21b} and~\cite{Lyo21a}. The first conference paper~\cite{Lyo21b} introduces intuitionistic grammar logics and provides sound and complete axiomatizations. The second conference paper~\cite{Lyo21a} shows how to derive nested sequent systems for intuitionistic mono-modal logics with seriality and Horn-Scott-Lemmon axioms by means of structural refinement, solving an open problem in~\cite{MarStr14}. We note that the work in Sections~\ref{sec:refinement} and~\ref{sec:nested-calculi} significantly generalizes the work in~\cite{Lyo21a} and that all the work in \sect~\ref{sec:(un)decid} is entirely new.

%%%For Journal Version
%- Can show hp-admissibility of structural rules, hp-invertibility of rules
%- Can show syntactic cut-elimination
%- Can give decidability for full class of logics (with counter-model extraction)

%Introduce Propositional Logics, Semantics, Axiomatizations
\section{Preliminaries}\label{sec:log-prelims}

    \subsection{Intuitionistic Grammar Logics}

    %\resizebox{\columnwidth}{!}{
%%%NOTES
%1. Should mention that monotonic and nomial horn properties not only capture a notion of true formulae persisting into the future, but of false formula persisting into the past
%2. Should explain in a paragraph before properties are introduced what constitutes a model so that the first-order writing of the condition makes sense
%3. Should explain somewhere that \bot and \top are essentially propoisitonal variables unless we add falsum and verum. Also, will want to explain the conditions somewhere, and explain the intuition of the interpretation of each connective
%4. Include a table showing common properties that are among the ones we consider
%5. Include a table with the logics that are obtained by enforcing certain properties (sub-int logics of Restall and Corsi, Predicate logics of Ishigaki, Intuitionistic logic, Bi-intuitionistic logic, grammar logics, T, S4, S5, Kt, 'look-up-names-of-nominal logics', classical logic, 'free logics?', Dosens logic?, Vissers logic?)
%&. Explain that \chara and its converse are reserved for the bi-intuitionistic connectives

 The language of each intuitionistic grammar logic is defined relative to an \emph{alphabet} $\albet$, which is a non-empty countable set of \emph{characters}, used to index modalities. As in~\cite{DemNiv05}, we stipulate that each alphabet $\albet$ can be partitioned into a \emph{forward part} $\albetf := \{\chara, \charb, \charc, \ldots\}$ and a \emph{backward part} $\albetb := \{\conv{\chara}, \conv{\charb}, \conv{\charc}, \ldots\}$ where the following is satisfied:
$$
\albet := \albetf \cup \albetb \text{ where } \albetf \cap \albetb = \emptyset \text{ and } \chara \in \albetf \text{ \ifandonlyif } \conv{\chara} \in \albetb.
$$
 $\albetf$ contains \emph{forward characters}, which we denote by $\chara$, $\charb$, $\charc$, \etc (possibly annotated), and $\albetb$ contains \emph{backward characters}, which we denote by $\conv{\chara}$, $\conv{\charb}$, $\conv{\charc}$, \etc (possibly annotated). A \emph{character} is defined to be either a forward or backward character, and we use $\charx$, $\chary$, $\charz$, \etc (possibly annotated) to denote them. In what follows, modalities indexed with forward characters will be interpreted as making reference to future states along the accessibility relation within a relational model, and modalities indexed with backward characters will make reference to past states. We define the \emph{converse operation} to be a function $\conv{\cdot}$ mapping each forward character $\chara \in \albetf$ to its \emph{converse} $\conv{\chara} \in \albetb$ and vice versa; hence, the converse operation is its own inverse, i.e. for any $\charx \in \albet$, $\charx = \conv{\conv{\charx}}$.

 We let $\prop := \{p, q, r, \ldots\}$ be a denumerable set of \emph{propositional atoms} and define the language $\lang{\albet}$ relative to a given alphabet $\albet$ via the following grammar in BNF:
$$
A ::= p \ | \ \bot \ | \ A \lor A \ | \ A \land A \ | \ A \iimp A \ | \ \xdia A \ | \ \xbox A
$$
where $p$ ranges over the set $\prop$ of propositional atoms and $\charx$ ranges over the characters in the alphabet $\albet$. We use $A$, $B$, $C$, \etc (possibly annotated) to denote formulae in $\lang{\albet}$ and define $\inot A := A \iimp \bot$. %, and define $A \ieq B := (A \iimp B)\land(B \iimp A)$. 
 Formulae are interpreted over \emph{bi-relational $\albet$-models}~\cite{Lyo21b}, which are inspired by the models for intuitionistic modal and tense logics presented in~\cite{BovDov84,Dov85,Ewa86,PloSti86}: 

\begin{definition}[Bi-relational $\albet$-Model~\cite{Lyo21b}]\label{def:bi-relational-model} We define a \emph{bi-relational $\albet$-model} to be a tuple $M = (W, \leq, \{R_{\charx} \ | \ \charx \in \albet\}, V)$ such that:
\begin{itemize}

\item $W$ is a non-empty set of \emph{worlds} $\{w, u, v, \ldots\}$;

\item The \emph{intuitionistic relation} $\leq \ \subseteq W \times W$ is a preorder, i.e. it is reflexive and transitive;

\item The \emph{accessibility relation} $R_{\charx} \subseteq W \times W$ satisfies:

\begin{itemize}

\item[(F1)] For all $w, v, v' \in W$, if $w R_{\charx} v$ and $v \leq v'$, then there exists a $w' \in W$ such that $w \leq w'$ and $w' R_{\charx} v'$;

\item[(F2)] For all $w, w', v \in W$, if $w \leq w'$ and $w R_{\charx} v$, then there exists a $v' \in W$ such that $w' R_{\charx} v'$ and $v \leq v'$;

\item[(F3)] $w R_{\charx} u$ \ifandonlyif $u R_{\conv{\charx}} w$;

\end{itemize}

%\item For the set $\fworlds \subseteq W$ of \emph{fallible worlds}, if $w \in \fworlds$, and $wR_{\charx}u$ or $w \leq u$, then $u \in \fworlds$.

\item $V : W \to 2^{\prop}$ is a \emph{valuation function} satisfying the \emph{monotonicity condition}:  for each $w, u \in W$, if $w \leq u$, then $V(w) \subseteq V(u)$.

\end{itemize}

\end{definition}

\begin{remark} The (F1) condition is implied by the (F2) and (F3) conditions, and the (F2) condition is implied by the (F1) and (F3) conditions. In other words, we need only impose $\{(F1),(F3)\}$ or $\{(F2),(F3)\}$ on %the accessibility relations in 
our bi-relational $\albet$-models to characterize our logics. We mention both conditions (F1) and (F2) however since we make explicit use of both conditions later on.
\end{remark}

\begin{figure}[t]\label{fig:f1-f2}

\begin{center}
\begin{tabular}{c @{\hskip 3em} c}
\xymatrix@=2em{
w'\ar@{.>}[rr]|-{R_{\charx}}  & & v' \\
 & (F1) & \\
w\ar@{.>}[uu]|-{\leq}\ar[rr]|-{R_{\charx}} & & v\ar[uu]|-{\leq}
}

&

\xymatrix@=2em{
w'\ar@{.>}[rr]|-{R_{\charx}}  & & v' \\
 & (F2) & \\
w\ar[uu]|-{\leq}\ar[rr]|-{R_{\charx}} & & v\ar@{.>}[uu]|-{\leq}
}

\end{tabular}
\end{center}

\caption{Depictions of the (F1) and (F2) conditions imposed on bi-relational $\albet$-models. Dotted arrows indicate the relations implied by the presence of the solid arrows.} %Each condition is a consequence of the other in the presence of condition (F3).}
\end{figure}

The (F1) and (F2) conditions are depicted in \fig~\ref{fig:f1-f2} and ensure the monotonicity of complex formulae (see \lem~\ref{lem:persistence}) in our models, which is a property characteristic of intuitionistic logics.\footnote{For a discussion of these conditions and related literature, see~\cite[\cptr~3]{Sim94}.} If an accessibility relation $R_{\chara}$ is indexed with a forward character, then we interpret it as a relation to \emph{future} worlds, and if an accessibility relation $R_{\conv{\chara}}$ is indexed with a backward character, then we interpret it as a relation to \emph{past} worlds. Hence, our formulae and related models have a tense character, showing that our logics generalize the intuitionistic tense logics of Ewald~\cite{Ewa86}.

We interpret formulae from $\lang{\albet}$ over bi-relational models via the following clauses.

%\cite{GabSheSkv09}
\begin{definition}[Semantic Clauses~\cite{Lyo21b}]
\label{def:semantic-clauses} Let $M$ be a bi-relational $\albet$-model with $w \in W$. The \emph{satisfaction relation} $M,w \sat A$ between $w \in W$ of $M$ and a formula $A \in \lang{\albet}$ is inductively defined as follows:

\begin{itemize}

\item $M,w \sat p$ \ifandonlyif $p \in V(w)$, for $p \in \prop$;

\item $M,w \not\sat \bot$;

\item $M,w \sat A \lor B$ \ifandonlyif $M,w \sat A$ or $M,w \sat B$;

\item $M,w \sat A \land B$ \ifandonlyif $M,w \sat A$ and $M,w \sat B$;

%\item $M, w \sat \inot A$ \ifandonlyif for all $w' \in W$, if $w \leq w'$, then $M,w' \not\sat A$;

\item $M,w \sat A \iimp B$ \ifandonlyif for all $w' \in W$, if $w \leq w'$ and $M,w' \sat A$, then $M,w' \sat B$;

\item $M,w \sat \xdia A$ \ifandonlyif there exists a $v \in W$ such that $w R_{\charx} v$ and $M,v \sat A$;

\item $M,w \sat \xbox A$ \ifandonlyif for all $w', v' \in W$, if $w \leq w'$ and $w' R_{\charx} v'$, then $M,v' \sat A$.

\end{itemize}
%A formula $A$ is defined to be \emph{globally true on $M$}, written $M \sat A$, \ifandonlyif $M,u \sat A$ for all worlds $u \in W$ of $M$. A formula $A$ is defined to be \emph{valid}, written $\sat A$, \ifandonlyif $A$ is globally true on every bi-relational $\albet$-model. Last, we say that a set $\fseti$ of formulae \emph{semantically implies} a formula $A$, written $\fseti \sat A$, \ifandonlyif for all bi-relational $\albet$-models $M$ and each $w \in W$ of $M$, if $M, w \sat B$ for each $B \in \fseti$, then $M,w \sat A$.
\end{definition}

\begin{lemma}[Persistence]\label{lem:persistence}
Let $M$ be a bi-relational $\albet$-model with $w,u \in W$ of $M$. If $w \leq u$ and $M, w \sat A$, then $M, u \sat A$.
\end{lemma}

\begin{proof} By induction on the complexity of $A$.
\end{proof}

 Given an alphabet $\albet$, the set of formulae valid with respect to the class of bi-relational $\albet$-models is axiomatizable~\cite{Lyo21b}. We refer to the axiomatization as $\h\ikm(\albet)$ (with $\h$ denoting the fact that the axiomatization is a \emph{Hilbert calculus}), and call the corresponding logic that it generates $\ikm(\albet)$.

\begin{definition}[Axiomatization]\label{def:axiomatization} %Let $\albet$ be an alphabet. 
 We define our axiomatization $\h\ikm(\albet)$ below, where we have an axiom and inference rule for each $\charx \in \albet$.

\begin{multicols}{2}
\begin{itemize}

\item[A0] Any axiomatization for intuitionistic propositional logic

\item[A1] $\xbox (A \iimp B) \iimp (\xbox A \iimp \xbox B)$

\item[A2] $\xbox (A \land B) \ieq (\xbox A \land \xbox B)$

\item[A3] $\xdia (A \lor B) \ieq (\xdia A \lor \xdia B)$

\item[A4] $\xbox (A \iimp B) \iimp (\xdia A \iimp \xdia B)$

\item[A5] $(\xbox A \land \xdia B) \iimp \xdia (A \land B)$

\item[A6] $\inot \xdia \bot$ %$\xbox \inot A \iimp \inot \xdia A$

\item[A7] $(A \iimp \xbox \xdiac A) \land (\xdia \xboxc A \iimp A)$

\item[A8] $(\xdia A \iimp \xbox B) \iimp \xbox (A \iimp B)$

\item[A9] $\xdia (A \iimp B) \iimp (\xbox A \iimp \xdia B)$

\item[R0] \AxiomC{$A$}\AxiomC{$A \iimp B$}\RightLabel{$(mp)$}\BinaryInfC{$B$}\DisplayProof

\item[R1] \AxiomC{$A$}\RightLabel{$\nec$}\UnaryInfC{$\xbox A$}\DisplayProof

\end{itemize}
\end{multicols}
We define the logic $\ikm(\albet)$ to be the smallest set of formulae from $\lang{\albet}$ closed under substitutions of the axioms and applications of the inference rules. A formula $A$ is defined to be a \emph{theorem} of $\ikm(\albet)$  \ifandonlyif $A  \in \ikm(\albet)$.
\end{definition}

 We note that if we let $\albet := \{a,\conv{a}\}$, then the resulting logic is a notational variant of Ewald's intuitionistic tense logic $\ikt$~\cite{Ewa86}, which is a conservative extension of the mono-modal intuitionistic modal logic $\ik$~\cite{PloSti86}. In our setting, we let $\ikm(\albet)$ be the base intuitionistic grammar logic relative to $\albet$, and consider extensions of $\ikm(\albet)$ with sets $\axs$ of the following axioms. %, where $\charx \in \albet$:
$$
\D: \xbox A \iimp \xdia A \quad \ipa: (\diap{\charx_{1}} \cdots \diap{\charx_{n}} A \iimp \diap{\charx} A) \land (\boxp{\charx} A \iimp \boxp{\charx_{1}} \cdots \boxp{\charx_{n}} A)
$$
We refer to axioms of the form shown above left as \emph{seriality axioms}, and axioms of the form shown above right as \emph{intuitionistic path axioms} ($\ipa$s). We use $\axs$ to denote any arbitrary collection of axioms of the above forms. Moreover, we note that the collection of $\ipa$s subsumes the class of Horn-Scott-Lemmon axioms~\cite{Lyo21a} and includes multi-modal and intuitionistic variants of standard axioms such as $\T$, $\B$, $\four$, and $\five$.
$$
\T: (A \iimp \xdia A) \land (\xbox A \iimp A) \quad
\four: (\xdia \xdia A \iimp \xdia A) \land (\xbox A \iimp \xbox \xbox A)
$$
$$
\B: (\xdiac A \iimp \xdia A) \land (\xbox A \iimp \xboxc A)
\quad
\five: (\xdiac \xdia A \iimp \xdia A) \land (\xbox A \iimp \xboxc \xbox A)
$$
 It was proven that any extension of $\h\ikm(\albet)$ with a set $\axs$ of axioms is sound and complete relative to a sub-class of the bi-relational $\albet$-models satisfying frame conditions related to each axiom~\cite{Lyo21b}. %For each axiom in $\ax$ extending $\h\ikm(\albet)$, we impose a frame condition on our class of bi-relational $\albet$-models. 
 Axioms and related frame conditions are displayed in \fig~\ref{fig:axioms-related-conditions}, and extensions of $\h\ikm(\albet)$ with seriality and $\ipa$ axioms, along with their corresponding models, are defined below.

\begin{figure}[t]
\begin{center}
\bgroup
\def\arraystretch{1.5}
\begin{tabular}{| c || c |}
%\hline
%$\D$ & $\ipa$ \\
\hline
$\xbox A \iimp \xdia A$  & $(\diap{\charx_{1}} \cdots \diap{\charx_{n}} A \iimp \diap{\charx} A) \land (\boxp{\charx} A \iimp \boxp{\charx_{1}} \cdots \boxp{\charx_{n}} A)$ \\
\hline
$\forall w \exists u (w R_{\charx} u)$ & $\forall w_{0}, \ldots, w_{n} (w_{0} R_{\charx_{1}} w_{1} \land \cdots \land w_{n-1} R_{\charx_{n}} w_{n} \iimp w_{0} R_{\charx} w_{n})$\\
\hline
\end{tabular}
\egroup
\end{center}

\caption{Axioms are displayed in the first row and their related frame conditions are displayed directly underneath them. We note that when $n=0$ in the $\ipa$, the related frame condition is taken to be $\forall w (w R_{\charx} w)$.}
\label{fig:axioms-related-conditions}
\end{figure}

\begin{definition}[Syntactic Notions for Extensions] We define the axiomatization $\h\ikm(\albet,\axs)$ to be $\h\ikm(\albet) \cup \axs$, and define the logic $\ikm(\albet,\axs)$ to be the smallest set of formulae from $\lang{\albet}$ closed under substitutions of the axioms from $\h\ikm(\albet)$ and $\axs$ and applications of the inference rules. A formula $A$ is defined to be an \emph{$\ikm(\albet,\axs)$-theorem}, written $\der A$, \ifandonlyif $A \in \ikm(\albet,\axs)$, and a formula $A$ is said to be \emph{derivable} from a set of formulae $\fseti \subseteq \lang{\albet}$, written $\fseti \der A$, \ifandonlyif for some $B_{1}, \ldots, B_{n} \in \fseti$, $\der B_{1} \land \cdots \land B_{n} \iimp A$.
\end{definition}

\begin{definition}[Semantic Notions for Extensions] We define a \emph{bi-relational $(\albet,\axs)$-model} to be a bi-relational $\albet$-model satisfying each frame condition related to an axiom $A \in \axs$. A formula $A$ is defined to be \emph{globally true} on a bi-relational $(\albet,\axs)$-model $M$, written $M \ent A$, \ifandonlyif $M,u \sat A$ for all worlds $u \in W$ of $M$. A formula $A$ is defined to be \emph{$(\albet,\axs)$-valid}, written $\ent A$, \ifandonlyif $A$ is globally true on every bi-relational $(\albet,\axs)$-model. Last, we say that a set $\fseti$ of formulae \emph{semantically implies} a formula $A$, written $\fseti \ent A$, \ifandonlyif for all bi-relational $(\albet,\axs)$-models $M$ and each $w \in W$ of $M$, if $M, w \sat B$ for each $B \in \fseti$, then $M,w \sat A$.
\end{definition}

\begin{remark} Note that the axiomatization $\h\ikm(\albet) = \h\ikm(\albet,\emptyset)$ and that a bi-relational $(\albet,\emptyset)$-model is a bi-relational $\albet$-model.
\end{remark}

 One can prove strong soundness for each intuitionistic grammar logic (i.e. for any $\fseti \subseteq \lang{\albet}$ and $A \in \lang{\albet}$, if $\fseti \der A$, then $\fseti \ent A$) by showing that each axiom in $\h\ikm(\albet,\axs)$ is $(\albet,\axs)$-valid and that each inference rule preserves $(\albet,\axs)$-validity. The converse of this result, strong completeness, is shown by means of a typical canonical model construction. Proofs of both facts can be found in Lyon~\cite{Lyo21b}.

\begin{theorem}[Soundness and Completeness~\cite{Lyo21b}]\label{thm:sound-complete-logic}
$\fseti \der A$ \iffi $\fseti \ent A$.
\end{theorem}

%\subsection{Forward Intuitionistic Modal Logics}

%\tim{TODO: explain that we use int. path axioms that only use $\fd$ and $\bd$ and are regular. Prove decidability of these which solved decid problem for transitive case.}

    \subsection{Formal Grammars and Languages}

%As mentioned previously (and as is implied by the name), grammar logics connect concepts concerning formal grammars to logical concepts.  Therefore, due to the intimate connection between formal grammars and grammar logics, we will introduce additional formal language theoretic concepts below. Such concepts will allow us to establish a correspondence between formal grammars---in particular, specific types of \emph{Semi-Thue Systems} (\dfn~\ref{def:CFCST-kms})---with properties imposed on $\albet$-frames (\dfn~\ref{fig:frame-conditions-production-rules}), and to define new classes of grammar logics (\dfn~\ref{def:axiomatization-km}). Most definitions are taken from~\cite{DemNiv05}.

A central component to the structural refinement methodology---i.e. the extraction of nested calculi from labeled---is the use of inference rules (viz. reachability rules) whose applicability depends on strings generated by formal grammars~\cite{Lyo21thesis}. We therefore introduce grammar-theoretic notions that are essential to the functionality of such rules, and formally define such rules in \sect~\ref{sec:refinement}.

%~\cite{DemNiv05}
\begin{definition}[$\albetstr$] We let $\concat$ be the \emph{concatenation operation} with $\varepsilon$ the \emph{empty string}. We define the set $\albet^{*}$ of \emph{strings over $\albet$} to be the smallest set such that:
\begin{itemize}

\item $\albet \cup \{\varepsilon\} \subseteq \albet^{*}$

\item $\text{If } \stra \in \albet^{*} \text{ and } \charx \in \albet \text{, then } \stra \concat \charx \in \albet^{*}$

\end{itemize}
\end{definition}

For a set $\albetstr$ of strings, we use $\stra$, $\strb$, $\strc$, \etc (potentially annotated) to represent strings in $\albetstr$. As usual, the empty string $\empstr$ is taken to be the identity element for the concatenation operation, i.e. $\stra \concat \empstr = \empstr \concat \stra = \stra$ for $\stra \in \albet^{*}$. Furthermore, we will not usually mention the concatenation operation in practice and will let $\stra \cate \strb := \stra \concat \strb$, that is, we denote concatenation by simply gluing two strings together. Beyond concatenation, another useful operation to define on strings is the \emph{converse operation}.

\begin{definition}[String Converse~\cite{DemNiv05}] We extend the converse operation to strings as follows:
\begin{itemize}

\item $\conv{\varepsilon} := \varepsilon$;

\item $\text{If } \stra = \charx_{1} \cdots \charx_{n} \text{, then } \conv{\stra} := \conv{\charx}_{n} \cdots \conv{\charx}_{1}$.

\end{itemize}
\end{definition}

%\begin{definition}[String Length] The \emph{length} of a string $\stra \in \albet^{*}$ is defined recursively: 
%\begin{itemize}

%\item $\lenstr{\stra} := 0$, if $\stra = \empstr$;

%\item $\lenstr{\stra} = \lenstr{\strb \cate \ques} := \lenstr{\strb} + 1$, if $\stra = \strb \cate \ques$ with $\strb \in \albet^{*}$ and $\ques \in \albet$.

%\end{itemize}
%\end{definition}

% Since accessibility relations in a bi-relational $\albet$-model are indexed with characters from $\albet$, paths of accessibility relations correspond to strings from $\albetstr$. We therefore define generalized relations $R_{\stra}$ relative to strings.

%~\cite{DemNiv05}
%\begin{definition}[Generalized Relations]\label{def:generalized-relations} Let $M = (W,\leq,\{R_{\charx} \ | \ \charx \in \albet\},V)$ be a bi-relational $\albet$-model, $\charx \in \albet$, and $\stra \in \albetstr$. We inductively define the \emph{generalized accessibility relation} $R_{\stra}$ relative to the string $\stra$ as follows:
%\begin{itemize}

%\item $R_{\empstr} := \{(w,w) \ | \ w \in W\}$;

%\item $R_{\stra \cate \charx} := \{(w,u) \ | \ \exists v \in W, (w,v) \in R_{\stra} \text{ and } (v,u) \in R_{\charx}\}$.

%\end{itemize}
%\end{definition}

%\begin{remark}\label{rmk:R-Converses}
% By \dfn~\ref{def:generalized-relations} above, it follows that $(w,u) \in R_{\stra}$ \iffi $(u,w) \in R_{\conv{\stra}}$.
%\end{remark}

% Similarly, 
 We define strings of modalities accordingly: if $\stra = \chara_{0} \cate \chara_{1} \cdots \chara_{n}$, then $[ \stra ] =  [ \chara_{0} ] [ \chara_{1} ] \cdots [ \chara_{n} ]$ and $\langle \stra \rangle =  \langle \chara_{0} \rangle \langle \chara_{1} \rangle \cdots \langle \chara_{n} \rangle$, with $\strabox \phi = \stradia \phi = \phi$ when $\stra = \empstr$. Hence, every $\ipa$ may be written in the form $(\langle \stra \rangle A \iimp \langle \charx \rangle A) \land ([ \charx ] A \iimp [\stra] A)$, where $\langle \stra \rangle = \langle \charx_{1} \rangle \cdots \langle \charx_{n} \rangle$ and $[ \stra ] = [ \charx_{1} ] \cdots [ \charx_{n} ]$. We make use of this notation to compactly define \emph{$\axs$-grammars}, which are types of \emph{Semi-Thue systems}~\cite{Pos47}, encoding information contained in a set $\axs$ of axioms, and employed later on in the definition of our reachability rules. % (see \sect~\ref{sec:refinement}).

\begin{definition}[$\axs$-grammar~\cite{Lyo21a}]\label{def:grammar} An \emph{$\axs$-grammar} is a set $\g{\axs}$ such that:
\begin{center}
$(\charx \pto \stra), (\conv{\charx} \pto \conv{\stra}) \in \g{\axs}$ \iffi $(\langle \stra \rangle A \iimp \langle \charx \rangle A) \land ([ \charx ] A \iimp [\stra] A) \in \axs$.
\end{center}
We call rules of the form $\charx \pto \stra$ \emph{production rules}, where $\charx \in \albet$ and $\stra \in \albetstr$.
\end{definition}

%\begin{figure}[t]
%\begin{center}
%\bgroup
%\def\arraystretch{1.1}
%\begin{tabular}{| l | l | l |}
%\hline
%Name & Frame Property & Production Rules\\
%\hline
%Reflexivity & $\forall w wRw$ & $\{\ques \pto \empstr \ | \ \ques \in \albet\}$ \\
%Symmetry & $\forall w, u (wRu \iimp uRw)$ & $\{\ques \pto \conv{\ques} \ | \ \ques \in \albet\}$ \\
%Transitivity & $\forall w, v, u (wRv \land vRu \iimp wRu)$ & $\{\ques \pto \ques \cate \ques \ | \ \ques \in \albet\}$ \\
%Euclideanity & $\forall w, v, u (wRv \land wRu \iimp vRu)$ & $\{\ques \pto \conv{\ques} \cate \ques \ | \ \ques \in \albet\}$ \\
%\hline
%\end{tabular}
%\egroup
%\end{center}
%\caption{**double check these}
%\label{fig:frame-conditions-production-rules}
%\end{figure}

An $\axs$-grammar $\g{\axs}$ is a type of string re-writing system. For example, if $\charx \pto \stra \in \g{\axs}$, we may derive the string $\strb \cate \stra \cate \strc$ from $\strb \cate \charx \cate \strc$ in one-step by applying the mentioned production rule. Through repeated applications of production rules to a given string $\stra \in \albetstr$, one derives new strings, the collection of which, determines a language. Let us make such notions precise by means of the following definition:

%~\cite{DemNiv05}
\begin{definition}[Derivation, Language~\cite{Lyo21a}]\label{def:derivation-language} Let $\g{\axs}$ be an $\axs$-grammar. The \emph{one-step derivation relation} $\osdr$ holds between two strings $\stra$ and $\strb$ in $\albetstr$, written $\stra \osdr \strb$, \iffi there exist $\stra', \strb' \in \albetstr$ and $\charx \pto \strc \in \g{\axs}$ such that $\stra = \stra' \cate \charx \cate \strb'$ and $\strb = \stra' \cate \strc \cate \strb'$.  The \emph{derivation relation} $\dr$ is defined to be the reflexive and transitive closure of $\osdr$. For two strings $\stra, \strb \in \albetstr$, we refer to $\stra \dr \strb$ as a \emph{\cfg-derivation of $\strb$ from $\stra$}, and define %the \emph{length} of a \cfg-derivation to be the number of one-step derivations applied to derive $\strb$ from $\stra$ in $\g{\axs}$. Last, for a string $\stra \in \albet^{*}$, the \emph{language of $\stra$ relative to $\g{\axs}$} is defined to be the set $\glang(\stra) := \{\strb \ | \ \stra \dr \strb \}$.
 its \emph{length} to be equal to the minimal number of one-step derivations needed to derive $\strb$ from $\stra$ in $\g{\axs}$. Last, for a string $\stra \in \albet^{*}$, the \emph{language of $\stra$ relative to $\g{\axs}$} is defined to be the set $\glang(\stra) := \{\strb \ | \ \stra \dr \strb \}$.
\end{definition}

%~\cite{DemNiv05}
%\begin{definition}[Production Rule Satisfaction]\label{def:production-rule-sat-kms} We say that a $\albet$-model \emph{satisfies a production rule $a \pto \stra$} \ifandonlyif $R_{\stra} \subseteq R_{a}$, and say that a $\albet$-model \emph{satisfies a $\axsinp,\axsoutp)$-system $\thuesys$} \ifandonlyif it satisfies all production rules in $\thuesys$.
%\end{definition}

%~\cite{DemNiv05}
%\begin{definition}[$\thuesys$-validity]\label{def:S-validity-kms} Let $\thuesys$ be a $(\axsinp,\axsoutp)$-system. A formula $\phi \in \lang{\albet}$ is \emph{$\thuesys$-valid} \ifandonlyif for every $\albet$-model $M$, if $M$ satisfies $\thuesys$, then $M \Vdash \phi$.
%\end{definition}

%Labelled Sequents a la Simpson
\section{Labeled Sequent Systems}\label{sec:labeled-systems}

%We use the name $\lika$ to denote a labeled system as opposed to Simpson's name $\mathbf{L}_{\Box\Diamond}(\mathcal{T})$~\cite{Sim94}.

We generalize Simpson's labeled sequent systems for intuitionistic modal logics~\cite{Sim94} to the multi-modal case with converse modalities. Our systems make use of \emph{labels} (which we occasionally annotate) from a denumerable set $\lab := \{w,u,v,\ldots\}$, as well as two distinct types of formulae: (i) \emph{labeled formulae}, which are of the form $w : A$ with $w \in \lab$ and $A \in \lang{\albet}$, and (ii) \emph{relational atoms}, which are of the form $wR_{\charx}u$ for $w,u \in \lab$ and $\charx \in \albet$. We define a \emph{labeled sequent} to be a formula of the form $\rel, \Gamma \sar w : A$, where $\rel$ is a (potentially empty) multiset of relational atoms, and $\Gamma$ is a (potentially empty) multiset of labeled formulae. We will occasionally use $\lseq$ and annotated versions thereof to denote labeled sequents, and we let $\lab(\rel)$, $\lab(\Gamma)$, and $\lab(\lseq)$ be the set of labels occurring in a multiset of relational atoms $\rel$, a multiset of labeled formulae $\Gamma$, and a labeled sequent $\lseq$, respectively. For a string $\stra = \charx_{1} \cdots \charx_{n}$, we let $wR_{\stra}u := wR_{\charx_{1}}w_{1}, w_{1}R_{\charx_{2}}w_{2}, \ldots, w_{n-1}R_{\charx_{n}}u$, %for some labels $w_{1}, \ldots, w_{n-1}$, 
 and note that $wR_{\empstr}u := (w = u)$. %For a multiset $\Gamma$ of labeled formulae, we let $\Gamma \restriction w$ be the multiset $\{A \ | \ w : A \in \Gamma\}$, which is equal to the empty multiset $\emptyset$ when $w \not \in \lab(\Gamma)$
For a multiset $\Gamma$ of labeled formulae, we let $w : \Gamma$ be the multiset $\{u : A \in \Gamma \ | \ w = u\}$, which is equal to the empty multiset $\emptyset$ when $w \not \in \lab(\Gamma)$.
 
 For any alphabet $\albet$ and set $\axs$ of axioms, we obtain a calculus $\calc$, which is displayed in \fig~\ref{fig:labeled-calculi}. Such systems can be seen as intuitionistic variants of the labeled systems $\mathsf{G3KM}(S)$ for classical grammar logics~\cite{Lyo21thesis}, obtained by fixing a single labeled formula on the right-hand-side of the sequent arrow.

%Pg. 126 Simpson has lab. seq. calculi
\begin{figure}[!t]
\noindent\hrule

\begin{center}
\begin{tabular}{c c} % @{\hskip 1em} c}
\AxiomC{}
\RightLabel{$\id$}
\UnaryInfC{$\rel, \Gamma, w : p \sar w : p$}
\DisplayProof

&

\AxiomC{}
\RightLabel{$\botl$}
\UnaryInfC{$\rel, w : \bot,\Gamma \sar u : A$}
\DisplayProof
\end{tabular}
\end{center}

\begin{center}
%\resizebox{\columnwidth}{!}{
\begin{tabular}{c @{\hskip .7em} c}
\AxiomC{$\rel, \Gamma, w : A \sar u : C$}
\AxiomC{$\rel, \Gamma, w : B \sar u : C$}
\RightLabel{$\disl$}
\BinaryInfC{$\rel, \Gamma, w : A \vee B \sar u : C$}
\DisplayProof

&

\AxiomC{$\rel, \Gamma, w : A \sar w : B$}
\RightLabel{$\iimpr$}
\UnaryInfC{$\rel, \Gamma \sar w : A \iimp B$}
\DisplayProof
\end{tabular}
%}
\end{center}

\begin{center}
\begin{tabular}{c c}
\AxiomC{$\rel, wR_{\charx}u, \Gamma \sar v : A$}
\RightLabel{$\ddr^{\dag}$}
\UnaryInfC{$\rel, \Gamma \sar v : A$}
\DisplayProof

&

\AxiomC{$\rel, \Gamma \sar w : A_{i}$}
\RightLabel{$\disr~i \in \{1,2\}$}
\UnaryInfC{$\rel, \Gamma \sar w : A_{1} \vee A_{2}$}
\DisplayProof
\end{tabular}
\end{center}

\begin{center}
\begin{tabular}{c @{\hskip 1em} c}
\AxiomC{$\rel, \Gamma, w :A, w :B \sar u : C$}
\RightLabel{$\conl$}
\UnaryInfC{$\rel, \Gamma, w :A \wedge B \sar u : C$}
\DisplayProof

&

\AxiomC{$\rel, \Gamma \sar w : A$}
\AxiomC{$\rel, \Gamma \sar w : B$}
\RightLabel{$\conr$}
\BinaryInfC{$\rel, \Gamma \sar w : A \wedge B$}
\DisplayProof
\end{tabular}
\end{center}

\begin{center}
\begin{tabular}{c}
\AxiomC{$\rel, \Gamma, w :A \iimp B \sar w :A$}
\AxiomC{$\rel, \Gamma, w : B \sar u : C$}
\RightLabel{$\iimpl$}
\BinaryInfC{$\rel, \Gamma, w :A \iimp B \sar u : C$}
\DisplayProof 
\end{tabular}
\end{center}

\begin{center}
\begin{tabular}{c c}
\AxiomC{$\rel, w R_{\charx} u, \Gamma, u : A \sar v : B$}
\RightLabel{$\xdial^{\dag}$}
\UnaryInfC{$\rel, \Gamma, w : \xdia A \sar v : B$}
\DisplayProof

&

\AxiomC{$\rel, w R_{\charx} u, \Gamma \sar u : A$}
\RightLabel{$\xdiar$}
\UnaryInfC{$\rel, w R_{\charx} u, \Gamma \sar w : \xdia A$}
\DisplayProof
\end{tabular}
\end{center}

\begin{center}
\begin{tabular}{c c}
\AxiomC{$\rel, w R_{\charx} u, \Gamma \sar u : A$}
\RightLabel{$\xboxr^{\dag}$}
\UnaryInfC{$\rel, \Gamma \sar w : \xbox A$}
\DisplayProof

&

\AxiomC{$\rel, w R_{\charx} u, \Gamma, w : \xbox A, u : A \sar v : C$}
\RightLabel{$\xboxl$}
\UnaryInfC{$\rel, w R_{\charx} u, \Gamma, w : \xbox A \sar v : C$}
\DisplayProof
\end{tabular}
\end{center}

\begin{center}
\begin{tabular}{c c}
\AxiomC{$\rel, w R_{\stra} u, w R_{\charx} u,  \Gamma \sar v : A$}
\RightLabel{$\ipar$}
\UnaryInfC{$\rel, w R_{\stra} u, \Gamma \sar v : A$}
\DisplayProof

&

\AxiomC{$\rel, w R_{\charx} u, u R_{\conv{\charx}} w,  \Gamma \sar v : A$}
\RightLabel{$\convr$}
\UnaryInfC{$\rel, w R_{\charx} u, \Gamma \sar v : A$}
\DisplayProof
\end{tabular}
\end{center}

\hrule
\caption{The labeled calculi $\calc$. We have $\ddr$ as a rule in the calculus, if $\D \in \axs$, and an $\ipar$ rule in the calculus, for each $\ipa$ of the form $(\langle \stra \rangle A \iimp \langle \charx \rangle A) \land ([ \charx ] A \iimp [\stra] A) \in \axs$. Furthermore, we have a $\xdial$, $\xdiar$, $\xboxl$, $\xboxr$, and $\convr$ rule for each $\charx \in \albet$. The side condition $\dag$ states that the rule is applicable only if $u$ is fresh.} %\emph{fresh}, i.e. $u$ cannot occur in the conclusion.}
\label{fig:labeled-calculi}
\end{figure}%Pg. 126 Simpson has sequent calculus

We classify the $\id$ and $\botl$ rules as \emph{initial rules}, the $\ddr$, $\ipar$, and $\convr$ rules as \emph{structural rules}, and the remaining rules in \fig~\ref{fig:labeled-calculi} as \emph{logical rules}. Initial rules serve as the axioms of our labeled systems and initiate proofs, structural rules operate on relational atoms, and logical rules construct complex logical formulae. %Our use of the term \emph{structural rules} in reference to $\ddr$, $\ipar$, and $\convr$ is consistent with the use of the term in the literature on proof systems for modal and related logics~\cite{Bru09,CiaLyoRamTiu21,GorPosTiu08,GorPosTiu11} and is based on the fact that such rules manipulate the underlying data structure of sequents as opposed to introducing more complex logical formulae. 
 The $\ddr$, $\xdial$, and $\xboxr$ rules possess a side condition, which stipulates that the rule is applicable only if the label $u$ is \emph{fresh}, i.e. if the rule is applied, then the label $u$ will not occur in the conclusion. We define the \emph{auxiliary}, \emph{principal}, and \emph{active} formulae of a rule in the usual way (cf.~\cite[Chapter 1]{Tak13}), i.e., a formula is \emph{auxiliary} in a rule if it is explicitly mentioned in the premise(s) and is used to derive the \emph{principal} formula, which is the formula explicitly mentioned in the conclusion; a formula is \emph{active} if it is auxiliary or principal. For example, in the $\xboxr$ rule, $w R_{\charx} u$ and $u : A$ are auxiliary, $w : \xbox A$ is principal, and all such formulae are active. A \emph{proof} in $\calc$ is constructed in the usual fashion by successively applying logical or structural rules to initial rules, and the \emph{height} of a proof is defined to be the longest sequence of sequents from the conclusion of the proof to an initial rule.

 We point out that the $\ddr$ and $\ipar$ structural rules form a proper subclass of Simpson's \emph{geometric structural rules}~(see~\cite[p.~126]{Sim94}) used to generate labeled sequent systems for $\ik$ extended with any number of \emph{geometric axioms}. When $\stra = \empstr$ in $\ipar$, i.e. when $(A \iimp \langle \charx \rangle A) \land ([ \charx ] A \iimp A) \in \axs$, the structural rule $(i_{\charx}^{\empstr})$ is defined accordingly: 
%We note that although the term \emph{structural rule} has historically been used to categorize rules of a different form than $\D$ and $\sa$ (e.g. weakening and contraction~\cite{Tak13}), the usage in our context is consistent with the more contemporary usage in the literature~\cite{Bru09,CiaLyoRamTiu21,GorPosTiu08} and is based on the fact that such rules manipulate the underlying data structure of sequents as opposed to introducing more complex logical formulae. Also, we point out that the $\sa$ rules form a proper subclass of Simpson's $(S_{\chi})$ \emph{geometric structural rules}~(see~\cite[p.~126]{Sim94}) used to generate labeled sequent systems for $\ik$ extended with any number of \emph{geometric axioms}. When $n=0$ or $k=0$ in an $\axhsl$, i.e. when $\phi(0,k) \in \axs$, $\phi(n,0) \in \axs$, or $\phi(0,0) \in \axs$, the structural rules $(S_{0,k})$, $(S_{n,0})$, and $(S_{0,0})$ are defined accordingly:
\begin{center}
\AxiomC{$\rel, w R_{\charx} w,  \Gamma \sar u : A$}
\RightLabel{$(i_{\charx}^{\empstr})$}
\UnaryInfC{$\rel, \Gamma \sar u : A$}
\DisplayProof
\end{center}

 We now prove that each calculus $\calc$ is complete, that is, every $(\albet,\axs)$-valid formula is provable in $\calc$. (NB. Soundness will be shown in the next section.) To accomplish this, we first prove a sequence of height-preserving admissibility and invertibility results, which will be helpful in proving cut-elimination (see \thm~\ref{thm:cut-elim}), and ultimately, in establishing completeness. A rule is \emph{(hp-)admissible} \iffi if each premise $\lseq_{1}, \ldots, \lseq_{n}$ of the rule has a proof (of heights $h_{1}, \ldots, h_{n}$, respectively), then the conclusion has a proof (of height $h \leq \max\{h_{1}, \ldots, h_{n}\}$).\footnote{We use the prefix `hp-' to stand for `height-preserving.'} The (hp-)admissible rules we consider are shown in \fig~\ref{fig:admiss-rules}. A rule is \emph{hp-invertible} \iffi if the conclusion has a proof of height $h$, then each premise has a proof of height $h$ or less. Let us now establish a variety of proof-theoretic properties, which hold for each calculus $\calc$. % with $\albet$ an alphabet and $\axs$ a set of axioms.

\begin{figure}[t]\label{fig:admiss-rules}
\noindent\hrule

\begin{center}
\begin{tabular}{c c}
\AxiomC{$\rel, \Gamma \sar w : \bot$}
\RightLabel{$\botr$}
\UnaryInfC{$\rel, \Gamma \sar u : A$}
\DisplayProof

&

\AxiomC{$\rel, \Gamma \sar w : A$}
\RightLabel{$\sub$}
\UnaryInfC{$\rel(u/v), \Gamma(u/v) \sar w : A(u/v)$}
\DisplayProof
\end{tabular}
\end{center}
\begin{center}
\begin{tabular}{c c}
%\AxiomC{$\rel, \Gamma \sar $}
%\RightLabel{$\wkr$}
%\UnaryInfC{$\rel, \Gamma \sar w : A$}
%\DisplayProof
\AxiomC{$\rel, wR_{\charx}u, wR_{\charx}u, \Gamma \sar v : A$}
\RightLabel{$\ctrli$}
\UnaryInfC{$\rel, wR_{\charx}u, \Gamma \sar v : A$}
\DisplayProof

&

\AxiomC{$\rel, \Gamma, w : A, w : A \sar u : B$}
\RightLabel{$\ctrlii$}
\UnaryInfC{$\rel, \Gamma, w : A \sar u : B$}
\DisplayProof
\end{tabular}
\end{center}
\begin{center}
\begin{tabular}{c c}
\AxiomC{$\rel, \Gamma \sar w : A$}
\RightLabel{$\wkl$}
\UnaryInfC{$\rel', \rel, \Gamma, \Gamma' \sar w : A$}
\DisplayProof

&

\AxiomC{$\rel, \Gamma \sar w : A$}
\AxiomC{$\rel', \Gamma', w : A \sar u : B$}
\RightLabel{$\cut$}
\BinaryInfC{$\rel, \rel', \Gamma, \Gamma' \sar u : B$}
\DisplayProof
\end{tabular}
\end{center}

\hrule
\caption{Admissible rules.}
\end{figure}

\begin{lemma}\label{lem:generalized-id}
If $A \in \lang{\albet}$, then $\rel, \Gamma, w : A \sar w : A$ is provable in $\calc$.
\end{lemma}

\begin{proof} By induction on the complexity of $A$.
\end{proof}

\begin{lemma}\label{lem:botr-admiss}
The $\botr$ rule is hp-admissible in $\calc$.
\end{lemma}

\begin{proof} We prove the result by induction on the height of the given proof.

\textit{Base case.} Observe that the conclusion of $\botr$ in the proof shown below left is an instance of $\botl$, and hence, we can prove the same conclusion with a single application of $\botl$ as shown below right.
\begin{center}
\begin{tabular}{c @{\hskip 1em} c @{\hskip 1em} c}
\AxiomC{}
\RightLabel{$\botl$}
\UnaryInfC{$\rel, \Gamma, w : \bot \sar w : \bot$}
\RightLabel{$\botr$}
\UnaryInfC{$\rel, \Gamma, w : \bot \sar u : A$}
\DisplayProof

&

$\leadsto$

&

\AxiomC{}
\RightLabel{$\botl$}
\UnaryInfC{$\rel, \Gamma, w : \bot \sar u : A$}
\DisplayProof
\end{tabular}
\end{center}

\textit{Inductive step.} We note that $\botr$ cannot be applied after an application of $\disr$, $\conr$, $\iimpr$, $\xdiar$, or $\xboxr$ as the labeled formula on the right of the sequent arrow cannot be of the form $w : \bot$. Hence, these cases need not be considered. For all remaining cases, the $\botr$ rule freely permutes above any other rule instance, thus establishing the inductive step. For instance, if $\iimpl$ is followed by an application of $\botr$, then $\botr$ may be permuted above the right premise of $\iimpl$ as shown below: 
\begin{flushleft}
\begin{tabular}{c @{\hskip 1em} c}
\AxiomC{$\rel, \Gamma, w :A \iimp B \sar w :A$}
\AxiomC{$\rel, \Gamma, w : B \sar u : \bot$}
\RightLabel{$\iimpl$}
\BinaryInfC{$\rel, \Gamma, w :A \iimp B \sar u : \bot$}
\RightLabel{$\botr$}
\UnaryInfC{$\rel, \Gamma, w :A \iimp B \sar v : C$}
\DisplayProof 

&

$\leadsto$
\end{tabular}
\end{flushleft}
\begin{flushright}
\AxiomC{$\rel, \Gamma, w :A \iimp B \sar w :A$}
\AxiomC{$\rel, \Gamma, w : B \sar u : \bot$}
\RightLabel{$\botr$}
\UnaryInfC{$\rel, \Gamma, w : B \sar v : C$}
\RightLabel{$\iimpl$}
\BinaryInfC{$\rel, \Gamma, w :A \iimp B \sar v : C$}
\DisplayProof 
\end{flushright}
\end{proof}

\begin{lemma}\label{lem:wk-sub-admiss}
 The $\sub$ and $\wkl$ rules are hp-admissible in $\calc$.
\end{lemma}

\begin{proof} The hp-admissibility of each rule is shown by induction on the height of the given proof. The base cases are simple as applying $\sub$ or $\wkl$ to an instance of $\id$ or $\botl$ yields another instance of the rule. In the inductive step, we make a case distinction based on the last rule applied above $\sub$ or $\wkl$. With the exception of $\ddr$, $\xdial$, and $\xboxr$, all cases are handled by permuting $\sub$ or $\wkl$ above the rule. In the $\ddr$, $\xdial$, and $\xboxr$ cases, we must ensure that if $\sub$ or $\wkl$ is applied to the premise of the rule that the freshness condition still holds. This can be ensured in the $\sub$ case by invoking IH twice and in the $\wkl$ case by invoking the hp-admissibility of $\sub$ before invoking IH. 

For instance, in the proof shown below left, the fresh label $u$ is substituted for $v$ after $\ddr$ is applied, and therefore, we must invoke IH once to replace $u$ by a fresh label $z$, and then invoke IH a second time to substitute $u$ for $v$, as shown in the proof below right. Applying $\ddr$ after both substitutions gives the desired conclusion.
\begin{center}
\begin{tabular}{c @{\hskip 1em} c @{\hskip 1em} c}
\AxiomC{$\rel, wR_{\charx}u, \Gamma \sar v : A$}
\RightLabel{$\ddr$}
\UnaryInfC{$\rel, \Gamma \sar v : A$}
\RightLabel{$\sub$}
\UnaryInfC{$\rel(u/v), \Gamma(u/v) \sar u : A$}
\DisplayProof

&

$\leadsto$

&

\AxiomC{$\rel, wR_{\charx}u, \Gamma \sar v : A$}
\RightLabel{IH}
\UnaryInfC{$\rel, wR_{\charx}z, \Gamma \sar v : A$}
\RightLabel{IH}
\UnaryInfC{$\rel(u/v), wR_{\charx}z, \Gamma(u/v) \sar u : A$}
\RightLabel{$\ddr$}
\UnaryInfC{$\rel(u/v), \Gamma(u/v) \sar u : A$}
\DisplayProof
\end{tabular}
\end{center}
\end{proof}

%All proofs are similar to those given for classical grammar logics~\cite{Lyo21thesis}.

\begin{lemma}\label{lem:invert}
 The following properties hold for $\calc$:
\begin{enumerate}

\item The $\disl$, $\conl$, $\xdial$, and $\xboxl$ rules are hp-invertible;

\item The $\iimpl$ rule is hp-invertible in the right premise;

\item All structural rules are hp-invertible.

\end{enumerate}
\end{lemma}

\begin{proof} The proofs of claims 1 and 2 are standard and are shown by induction on the height of the given proof. Claim 3 follows immediately from the hp-admissibility of $\wkl$. %We show claim 2 as the cases in claim 1 are similar.
\end{proof}

\begin{lemma}\label{lem:ctr-admiss}
 The $\ctrli$ and $\ctrlii$ rules are hp-admissible in $\calc$.
\end{lemma}

\begin{proof} The hp-admissibility of $\ctrli$ and $\ctrlii$ is shown by induction on the height of the given proof. The base cases for both rules are trivial since applying either rule to any instance of $\id$ or $\botl$ gives another instance of the rule. The inductive step for $\ctrli$ is also trivial as the rule freely permutes above any other rule of $\calc$. 

 The inductive step for $\ctrlii$ requires slightly more work: we assume that a rule $(r)$ was applied, yielding a labeled sequent $\rel, \Gamma, w : A, w : A \sar u : B$, followed by an application of $\ctrlii$. If neither contraction formula $w : A$ is principal in $(r)$, then we may resolve the case by invoking IH, and then applying $(r)$. If, however, a contraction formula $w : A$ is principal in $(r)$, then we need to use \lem~\ref{lem:invert}. We show how to resolve the case where $(r)$ is $\xdial$ and $w : A = w : \xdia C$. The remaining cases are similar.
\begin{flushleft}
\begin{tabular}{c @{\hskip 1em} c}
\AxiomC{$\rel, wR_{\charx}u, \Gamma, w : \xdia C, u : C \sar z : B$}
\RightLabel{$\xdial$}
\UnaryInfC{$\rel, \Gamma, w : \xdia C, w : \xdia C \sar z : B$}
\RightLabel{$\ctrlii$}
\UnaryInfC{$\rel, \Gamma, w : \xdia C \sar z : B$}
\DisplayProof

&

$\leadsto$
\end{tabular}
\end{flushleft}
\begin{flushright}
\AxiomC{$\rel, wR_{\charx}u, \Gamma, w : \xdia C, u : C \sar z : B$}
\RightLabel{\lem~\ref{lem:invert}}
\UnaryInfC{$\rel, wR_{\charx}v, wR_{\charx}u, \Gamma, v : C, u : C \sar z : B$}
\RightLabel{$\sub$}
\UnaryInfC{$\rel, wR_{\charx}u, wR_{\charx}u, \Gamma, u : C, u : C \sar z : B$}
\RightLabel{$\ctrli$}
\UnaryInfC{$\rel, wR_{\charx}u, \Gamma, u : C, u : C \sar z : B$}
\RightLabel{IH}
\UnaryInfC{$\rel, wR_{\charx}u, \Gamma, u : C \sar z : B$}
\RightLabel{$\xdial$}
\UnaryInfC{$\rel, \Gamma, w : \xdia C \sar z : B$}
\DisplayProof
\end{flushright}
 In the output proof, shown above right, the $\sub$ rule substitutes the label $u$ for the fresh label $v$ and we apply the hp-admissibility of $\ctrli$. As all operations are height-preserving, we can apply IH after all such operations.
\end{proof}

\begin{theorem}[Cut-elimination]\label{thm:cut-elim}
 The $\cut$ rule is admissible in $\calc$.
\end{theorem}

\begin{proof} We prove the result by induction on the lexicographic ordering of pairs $(|A|,h_{1} + h_{2})$, where $|A|$ is the complexity of the cut formula $A$ and $h_{1} + h_{2}$ is the sum of the heights of the proofs deriving the left and right premises of $\cut$, respectively.

Let us suppose first that one of the premises of $\cut$ is an initial rule. If the left premise of $\cut$ is an instance of $\botl$, then the conclusion will be an instance of $\botl$, and thus, the conclusion is provable by $\botl$ only, without the use of $\cut$. If the left premise of $\cut$ is an instance of $\id$, then our $\cut$ is of the following shape:
\begin{center}
\AxiomC{}
\RightLabel{$\id$}
\UnaryInfC{$\rel, \Gamma, w : p \sar w : p$}
\AxiomC{$\rel', \Gamma', w : p \sar u : A$}
\RightLabel{$\cut$}
\BinaryInfC{$\rel, \rel', \Gamma, w : p, \Gamma' \sar u : A$}
\DisplayProof
\end{center}
 Observe that the conclusion may be proven without $\cut$ by applying the hp-admissibility of $\wkl$ as shown below:
\begin{center}
\AxiomC{$\rel', \Gamma', w : p \sar u : A$}
\RightLabel{$\wkl$}
\UnaryInfC{$\rel, \rel', \Gamma, w : p, \Gamma' \sar u : A$}
\DisplayProof
\end{center}
 If the right premise of $\cut$ is an instance of $\botl$ with $w : \bot$ principal, then there are two cases to consider. Either, $w : \bot$ is not the cut formula, in which case, the conclusion is an instance of $\botl$, or $w : \bot$ is the cut formula, in which case our $\cut$ is of the following shape:
\begin{center}
\AxiomC{$\rel, \Gamma \sar w : \bot$}
\AxiomC{}
\RightLabel{$\botl$}
\UnaryInfC{$\rel', \Gamma', w : \bot \sar u : A$}
\RightLabel{$\cut$}
\BinaryInfC{$\rel, \rel', \Gamma, \Gamma' \sar u : A$}
\DisplayProof
\end{center}
 Then, we may prove the conclusion without the use of $\cut$ by applying the hp-admissibility of $\botr$ and $\wkl$.
\begin{center}
\AxiomC{$\rel, \Gamma \sar w : \bot$}
\RightLabel{$\botr$}
\UnaryInfC{$\rel, \Gamma \sar u : A$}
\RightLabel{$\wkl$}
\UnaryInfC{$\rel, \rel', \Gamma, \Gamma' \sar u : A$}
\DisplayProof
\end{center}
 If the right premise of $\cut$ is an instance of $\id$ with $w : p$ principal, then we have two cases to consider. Either, $w : p$ is not the cut formula, in which case, the conclusion is an instance of $\id$, or $w : p$ is the cut formula, meaning that our $\cut$ is of the following shape:
\begin{center}
\AxiomC{$\rel, \Gamma \sar w : p$}
\AxiomC{}
\RightLabel{$\id$}
\UnaryInfC{$\rel', \Gamma', w : p \sar w : p$}
\RightLabel{$\cut$}
\BinaryInfC{$\rel, \rel', \Gamma, \Gamma' \sar w : p$}
\DisplayProof
\end{center}
 Observe that the conclusion can be proven from the left premise of $\cut$ by applying the hp-admissibility of $\wkl$. Hence, $\cut$ is eliminable whenever one premise of $\cut$ is an initial rule.
 
 Let us now suppose that no premise of $\cut$ is an initial rule. We have two cases to consider: either, the cut formula is not principal in at least one premise of $\cut$ or the cut formula is principal in both premises. We consider each case in turn.
 
 We assume that neither cut formula is principal and consider the case when the left premise of $\cut$ is proven via a unary rule and the right premise of $\cut$ is proven via a binary rule. All remaining cases are shown similarly, and thus, we omit them.

\begin{center}
\AxiomC{$\rel_{1}, \Gamma_{1} \sar w : A$}
\RightLabel{$(r_{1})$}
\UnaryInfC{$\rel_{1}', \Gamma_{1}' \sar w : A$}
\AxiomC{$\rel_{2}, \Gamma_{2}, w : A  \sar \Delta_{2}$}
\AxiomC{$\rel_{2}', \Gamma_{2}', w : A  \sar \Delta_{2}'$}
\RightLabel{$(r_{2})$}
\BinaryInfC{$\rel_{2}'', \Gamma_{2}'', w : A \sar \Delta_{2}''$}
\RightLabel{$\cut$}
\BinaryInfC{$\rel_{1}', \rel_{2}'', \Gamma_{1}', \Gamma_{2}'' \sar \Delta_{2}''$}
\DisplayProof
\end{center}
 The case is resolved as shown below by lifting the $\cut$ to apply to the premise of $(r_{1})$. We note that if $(r_{1})$ must satisfy a freshness condition (e.g. as with the $\xdial$ rule), then it may be necessary to apply the hp-admissibility of $\sub$ to the premise of $(r_{1})$ prior to lifting the $\cut$ upward to ensure the condition will be met after applying $\cut$. However, this causes no issues as $h_{1} + h_{2}$ will still decrease.
\begin{center}
\AxiomC{$\rel_{1}, \Gamma_{1} \sar w : A$}
\AxiomC{$\rel_{2}, \Gamma_{2}, w : A  \sar \Delta_{2}$}
\AxiomC{$\rel_{2}', \Gamma_{2}', w : A  \sar \Delta_{2}'$}
\RightLabel{$(r_{2})$}
\BinaryInfC{$\rel_{2}'', \Gamma_{2}'', w : A \sar \Delta_{2}''$}
\RightLabel{$\cut$}
\BinaryInfC{$\rel_{1}, \rel_{2}'', \Gamma_{1}, \Gamma_{2}'' \sar \Delta_{2}''$}
\RightLabel{$(r_{1})$}
\UnaryInfC{$\rel_{1}', \rel_{2}'', \Gamma_{1}', \Gamma_{2}'' \sar \Delta_{2}''$}
\DisplayProof
\end{center}

 The case where the cut formula is principal in one premise of $\cut$ but not the other is routine, %and relies on the hp-invertibility 
 %We now consider 
 so we focus on the interesting case where the cut formula is principal in both premises of $\cut$. We consider the cases where the cut formula is of the form $A \iimp B$ or $\xdia A$ and omit the remaining cases as they are similar. In the first case, our $\cut$ is of the following shape: 
\begin{center}
\begin{tabular}{c @{\hskip 1em} c @{\hskip 1em} c}
$\prf$

&

$=$

&

\AxiomC{$\rel', \Gamma', w :A \iimp B \sar w :A$}
\AxiomC{$\rel', \Gamma', w : B \sar u : C$}
\RightLabel{$\iimpl$}
\BinaryInfC{$\rel', \Gamma', w : A \iimp B \sar u : C$}
\DisplayProof
\end{tabular}
\end{center}
\begin{center}
\AxiomC{$\rel, \Gamma, w : A \sar w : B$}
\RightLabel{$\iimpr$}
\UnaryInfC{$\rel, \Gamma \sar w : A \iimp B$}
\AxiomC{$\prf$}
\RightLabel{$\cut$}
\BinaryInfC{$\rel, \rel', \Gamma, \Gamma', \sar u : C$}
\DisplayProof 
\end{center}
 We first apply $\cut$ to the conclusion of $\iimpr$ and to the left premise of $\iimpl$, thus reducing the height of $\cut$, and then apply $\cut$ two additional times to the formulae $A$ and $B$ of less complexity.
\begin{center}
\begin{tabular}{c @{\hskip 1em} c @{\hskip 1em} c}
$\prf'$

&

$=$

&

\AxiomC{$\rel, \Gamma, w : A \sar w : B$}
\RightLabel{$\iimpr$}
\UnaryInfC{$\rel, \Gamma \sar w : A \iimp B$}
\AxiomC{$\rel', \Gamma', w :A \iimp B \sar w :A$}
\RightLabel{$\cut$}
\BinaryInfC{$\rel, \rel', \Gamma, \Gamma' \sar w :A$}
\DisplayProof
\end{tabular}
\end{center}
\begin{center}
\AxiomC{$\prf'$}

\AxiomC{$\rel, \Gamma, w : A \sar w : B$}
\AxiomC{$\rel', \Gamma', w : B \sar u : C$}
\RightLabel{$\cut$}
\BinaryInfC{$\rel, \rel', \Gamma, \Gamma', w : A \sar u : C$}

\RightLabel{$\cut$}

\BinaryInfC{$\rel, \rel, \rel', \rel', \Gamma, \Gamma, \Gamma', \Gamma' \sar u : C$}
\RightLabel{$\ctrl$}
\UnaryInfC{$\rel, \rel', \Gamma, \Gamma', \sar u : C$}
\DisplayProof 
\end{center}
 
 For the second case, when our cut formula is of the form $\xdia A$, our $\cut$ is of the following shape:
\begin{center}
\AxiomC{$\rel, w R u, \Gamma \sar u : A$}
\RightLabel{$\xdiar$}
\UnaryInfC{$\rel, w R u, \Gamma \sar w : \xdia A$}
\AxiomC{$\rel', w R v, \Gamma', v : A \sar z : B$}
\RightLabel{$\xdial$}
\UnaryInfC{$\rel', \Gamma', w : \xdia A \sar z : B$}
\RightLabel{$\cut$}
\BinaryInfC{$\rel, w R u, \rel', \Gamma, \Gamma' \sar z : B$}
\DisplayProof
\end{center}
 The case is resolved as shown below, where the hp-admissibility of $\sub$ is applied to replace the fresh label $v$ by $u$.
\begin{center}
\AxiomC{$\rel, w R u, \Gamma \sar u : A$}

\AxiomC{$\rel', w R v, \Gamma', v : A \sar z : B$}
\RightLabel{$\sub$}
\UnaryInfC{$\rel', w R u, \Gamma', u : A \sar z : B$}
\RightLabel{$\cut$}

\BinaryInfC{$\rel, \rel', w R u, w R u, \Gamma, \Gamma' \sar z : B$}
\RightLabel{$\ctrli$}
\UnaryInfC{$\rel, \rel', w R u, \Gamma, \Gamma' \sar z : B$}
\DisplayProof
\end{center}
%\begin{center}
%\AxiomC{$\rel, w R u, \Gamma \sar u : A$}
%\RightLabel{$\xboxr$}
%\UnaryInfC{$\rel, \Gamma \sar w : \xbox A$}
%\AxiomC{$\rel', w R v, \Gamma', w : \xbox A, v : A \sar z : C$}
%\RightLabel{$\xboxl$}
%\UnaryInfC{$\rel', w R v, \Gamma', w : \xbox A \sar z : C$}
%\RightLabel{$\cut$}
%\BinaryInfC{$\rel, \rel', w R v, \Gamma, \Gamma' \sar z : C$}
%\DisplayProof
%\end{center}
\end{proof}

Finally, we are in a position to confirm the completeness of each calculus $\calc$. Completeness relies in part on \thm~\ref{thm:sound-complete-logic}, that is, we assume that a formula $A$ is $(\albet,\axs)$-valid, implying that $A$ is provable in $\h\ikma$. We then establish completeness by showing that all axioms of $\h\ikma$ are provable in $\calc$ and that all inference rules are admissible, which implies that $A$ is provable in $\calc$. It is a straightforward exercise to show that all axioms of $\h\ikma$ are provable in $\calc$. The admissibility of $\nec$ is shown by invoking the hp-admissibility of $\wkl$ and $\sub$, and %since standard axiomatizations for propositional intuitionistic logic (axiom $A0$ in \dfn~\ref{def:axiomatization}) include \emph{modus ponens} as an inference rule~\cite{TroDal88}, 
 we use $\cut$ (which is admissible by \thm~\ref{thm:cut-elim}) to simulate $(mp)$.

\begin{theorem}[Completeness]\label{thm:calc-complete}
If $\Vdash^{\albet}_{\axs} A$, then $\vdash w : A$ is provable in $\calc$.
\end{theorem}

%Refinement
\section{Structural Refinement}\label{sec:refinement}

We now show how to \emph{structurally refine} the labeled systems introduced in the previous section, that is, we implement a methodology introduced and applied in~\cite{CiaLyoRamTiu21,Lyo21,Lyo21thesis,LyoBer19,Lyo21a} (referred to as \emph{structural refinement}, or \emph{refinement} more simply) for simplifying labeled systems and/or permitting the extraction of nested systems. The methodology consists of eliminating structural rules (viz. the $\ipar$ and $\convr$ rules in our setting) through the addition of \emph{propagation rules} (cf.~\cite{CasCerGasHer97,Fit72,Sim94}) or \emph{reachability rules}~\cite{Lyo21thesis,Lyo22} to labeled calculi. This begets systems (which we refer to as \emph{refined labeled calculi}) that are translatable into nested systems as only labeled sequents of a treelike shape are required in proofs.

It is interesting to point out that structural refinement takes on a more complex form in the current setting than in previous approaches. Typically, refined labeled calculi inherit both completeness and soundness from their `parent' labeled systems; however, in the current setting only the completeness of refined labeled calculi is obtained through this method. Soundness must be shown directly within each refined labeled calculus as the proof relies on labeled sequents having a treelike shape.\footnote{This is reminiscent of what happens in the setting of intuitionistic modal logics; see Simpson~\cite[Section 8.1.3]{Sim94}.} We then leverage the soundness of our refined labeled calculi to prove a soundness result for their parent labeled calculi.

\subsection{Completeness}

In the first part of this section, we show how to extract refined labeled calculi from the systems introduced in the previous section, which doubles as a completeness result for the extracted systems. As mentioned above, we demonstrate that structural rules can be eliminated in favor of propagation rules. The propagation rules we introduce are motivated by those of~\cite{Lyo21thesis,Lyo21a,TiuIanGor12}, and operate by viewing a labeled sequent as an automaton, allowing for the propagation of a formula (when applied bottom-up) from a label $w$ to a label $u$ given that a certain path of relational atoms exists between $w$ and $u$ (corresponding to a string generated by an $\axs$-grammar).

The definition of our propagation rules is built atop the notions introduced in the following two definitions: % (\dfn~\ref{def:propagation-graph} and~\ref{def:propagation-path}) and are based on similar definitions

%Propagation rules date back to the work of Fitting in~\cite{Fit72}, and obtain their name (coined in~\cite{CasCerGasHer97}) on the basis of their behavior, namely, when applied bottom-up such rules propagate formulae to different labels in a labeled sequent. Propagation rules hav 

\begin{definition}[Propagation Graph]\label{def:propagation-graph} The \emph{propagation graph} $\prgr{\rel} = (V,E)$ of a multiset of relational atoms $\rel$ is defined recursively on the structure of $\rel$:
%\begin{eqnarray*}
\begin{itemize}

\item $\prgr{\empseq} := (\emptyset, \emptyset)$;

\item $\prgr{wR_{\charx}u} := (\{w,u\}, \{(w,\charx,u),(u,\conv{\charx},w)\})$;

\item $\prgr{\rel_{1},\rel_{2}} := (V_{1}\cup V_{2}, E_{1} \cup E_{2}) \text{ where } PG(\rel_{i}) = (V_{i},E_{i})$.

\end{itemize}
%\end{eqnarray*}
We will often write $w \in \prgr{\rel}$ to mean $w \in \prgrdom$, and $(w,\charx,u) \in \prgr{\rel}$ to mean $(w,\charx,u) \in \prgredges$.
\end{definition}

\begin{definition}[Propagation Path]\label{def:propagation-path} %Let $\rel$ be a multiset of relational atoms. 
 We define a \emph{propagation path from $w_{1}$ to $w_{n}$ in $\prgr{\rel} := (V,E)$} to be a sequence of the following form:
$$
\ppath(w_{1},w_{n}) := w_{1}, \charx_{1}, w_{2}, \charx_{2}, \ldots , \charx_{n-1}, w_{n}
$$
such that $(w_{1}, \charx_{1}, w_{2}) , (w_{2}, \charx_{2}, w_{3}), \ldots, (w_{n-1}, \charx_{n-1}, w_{n}) \in \prgredges$. Given a propagation path of the above form, we define its \emph{converse} as shown below top and its \emph{string} as shown below bottom:
$$
\conv{\ppath}(w_{n},w_{1}) := w_{n}, \conv{\charx}_{n-1}, w_{n-1}, \conv{\charx}_{n-2}, \ldots, \conv{\charx}_{1}, w_{1}
$$
$$
\stra_{\ppath}(w_{1},w_{n}) := \charx_{1} \cate \charx_{2} \cate \cdots \cate \charx_{n-1}
$$
Last, we let $\emppath(w,w) := w$ represent an \emph{empty path} with the string of the empty path defined as $\stra_{\emppath}(w,w) := \empstr$.
\end{definition}

We are now in a position to define the operation of our propagation rules $\prdia$ and $\prbox$, which are displayed in \fig~\ref{fig:propagation-rules}. Each propagation rule $\prdia$ and $\prbox$ is applicable only if \emph{there exists a propagation path $\ppath(w,u)$ from $w$ to $u$ in the propagation graph $\prgr{\rel}$ such that the string $\stra_{\ppath}(w,u)$ is in the language $\glang(\charx)$.} We express this statement compactly by making use of its equivalent first-order representation:
$$
\exists \ppath(w,u) \in \prgr{\rel} (\stra_{\ppath}(w,u) \in \glang(\charx))
$$
 We consider the relational atoms that give rise to a propagation path used in an application of $\prdia$ or $\prbox$ to be \emph{active}. We provide further intuition %concerning the operation of 
 regarding such rules by means of an example:

\begin{example}\label{ex:propagation-graph-path} Let $\rel := vR_{\charx}u, uR_{\conv{\chary}}w$. We give a graphical depiction of $\prgr{\rel}$:
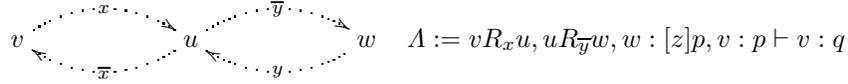
\begin{center}
\begin{minipage}[t]{.5\textwidth}
\xymatrix{
  v \ar@/^1pc/@{.>}[rr]|-{\charx} & &  u\ar@/^1pc/@{.>}[rr]|-{\conv{\chary}} \ar@/^1pc/@{.>}[ll]|-{\conv{\charx}} & &  w \ar@/^1pc/@{.>}[ll]|-{\chary}
}
\end{minipage}
%\begin{minipage}[t]{.05\textwidth}
%\
%\end{minipage}
\begin{minipage}[t]{.5\textwidth}
%\vspace{2.5em}
\begin{tabular}{c}
\AxiomC{ }
\noLine
\UnaryInfC{$\Lambda := vR_{\charx}u, uR_{\conv{\chary}}w, w : [z] p, v : p \sar v : q$}
\DisplayProof
\end{tabular}
\end{minipage}
\end{center}

Let $\axs := \{(\langle \chary \rangle \langle \conv{\charx} \rangle A \iimp \langle z \rangle A) \land ([z] A \iimp [\chary] [\conv{\charx}] A)\}$, so that the corresponding $\axs$-grammar is $\g{\axs} = \{z \pto \chary \conv{\charx}, \conv{z} \pto \conv{\chary} \charx \}$. Then, the path $\ppath(w,v) := w, \chary, u, \conv{\charx}, v$ exists between $w$ and $v$. The first production rule of $\g{\axs}$ implies that $\stra_{\ppath}(w,v) = \chary \conv{\charx} \in \glang(z)$. Therefore, we are permitted to (top-down) apply the propagation rule $(p_{[z]})$ to $\Lambda$ to delete the labeled formula $v : p$, letting us derive $vR_{\charx}u, uR_{\conv{\chary}}w, w : [z] p \sar v : q$. The relational atoms $vR_{\charx}u, uR_{\conv{\chary}}w$ are active in the inference.
\end{example}

\begin{figure}[t]
\noindent\hrule

\begin{center}
\begin{tabular}{c c}
\AxiomC{$\rel, \Gamma \sar u : A$}
\RightLabel{$\prdia^{\dag}$} %~\textit{ only if }~\exists \ppath(w,u) \in \prgr{\rel} (\stra_{\ppath}(w,u) \in \glang(\charx))$}
\UnaryInfC{$\rel, \Gamma \sar w : \xdia A$}
\DisplayProof
%\end{center}

&

%\begin{center}
%\resizebox{\columnwidth}{!}{
\AxiomC{$\rel, \Gamma, w : \xbox A, u : A \sar v : B$}
\RightLabel{$\prbox^{\dag}$} %~\textit{ only if }~\exists \ppath(w,u) \in \prgr{\rel} (\stra_{\ppath}(w,u) \in \glang(\charx))$}
\UnaryInfC{$\rel, \Gamma, w : \xbox A \sar v : B$}
\DisplayProof
\end{tabular}
\end{center}
$$
\dag:= \exists \ppath(w,u) \in \prgr{\rel} (\stra_{\ppath}(w,u) \in \glang(\charx))
$$

\hrule
\caption{Propagation rules. Each rule is applicable only if the side condition $\dag$ holds.}
\label{fig:propagation-rules}
\end{figure}

\begin{remark}\label{rem:diar-boxl-propagation-rule-instances}
The $\xdiar$ and $\xboxl$ rules are instances of $\prdia$ and $\prbox$, respectively.
\end{remark}

\begin{definition}[Refined Labeled Calculus] We define the \emph{refined labeled calculus} $\rcalc := (\calc \setminus \mathbf{R}) \cup \{\prdia,\prbox\}$, where
$$
\mathbf{R} := \{\ipar \ | \ (\langle \stra \rangle A \iimp \langle \charx \rangle A) \land ([ \charx ] A \iimp [\stra] A) \in \axs\} \cup \{\xdiar, \xboxl, \convr \ | \ \charx \in \albet\}.
$$
\end{definition}

 We show that each calculus $\rcalc$ is complete by means of a proof transformation procedure. That is, we show that through the elimination of structural rules we can transform a proof in $\calc$ into a proof in $\rcalc$. We first prove three crucial lemmata, and then argue the elimination result.

\begin{lemma}\label{lem:permutation-1} Let $\rel_{1} := \rel, w R_{\stra} u, w R_{\chary} u$ and $\rel_{2} := \rel, w R_{\stra} u$. Suppose we are given a proof in $\calc \cup \rcalc$ %\{\prdia, \prbox \ | \ \charx \in \albet\}$ 
 ending with:
\begin{center}
%\resizebox{\columnwidth}{!}{
\begin{tabular}{c}
\AxiomC{$\rel, w R_{\stra} u, w R_{\chary} u, \Gamma \sar v : A$}
\RightLabel{$\prdia$}
\UnaryInfC{$\rel, w R_{\stra} u, w R_{\chary} u,  \Gamma \sar z : \xdia A$}
\RightLabel{$(i^{\stra}_{\chary})$}
\UnaryInfC{$\rel, w R_{\stra} u, \Gamma \sar z : \xdia A$}
\DisplayProof
\end{tabular}
\end{center}
where the side condition $\exists \ppath(z,v) \in \prgr{\rel_{1}} (\stra_{\ppath}(z,v) \in \glang(\charx))$ holds due to $\prdia$. Then, $\exists \ppath'(z,v) \in \prgr{\rel_{2}} (\stra_{\ppath'}(z,v) \in \glang(\charx))$, that is to say, the $(i^{\stra}_{\chary})$ rule is permutable with the $\prdia$ rule. %%%, as shown below:
\end{lemma}

\begin{proof} Let $w R_{\stra} u = w R_{\charx_{1}} u_{1}, \ldots, u_{n} R_{\charx_{n}} u$. We have two cases: either (i) the relational atom $w R_{\chary} u$ is not active in the $\prdia$ inference, or (ii) it is. As case (i) is easily resolved, we show (ii).

 We suppose that the relational atom $w R_{\chary} u$ is active in $\prdia$, i.e. $w R_{\chary} u$ occurs along the propagation path $\ppath(z,v)$. To prove the claim, we need to show that $\exists \ppath'(z,v) \in \prgr{\rel_{2}} (\stra_{\ppath'}(z,v) \in \glang(\charx))$. We construct such a propagation path by performing the following operation on $\ppath(z,v)$: replace each occurrence of $w, \chary, u$ in $\prgr{\rel_{1}}$ with the propagation path $\ppath''(w,u)$ shown below left and replace each occurrence of $u, \conv{\chary}, w$ in $\prgr{\rel_{1}}$ with the propagation path $\conv{\ppath}''(u,w)$ shown below right:
$$
\ppath''(w,u) = w, \charx_{1}, u_{1}, \ldots, u_{n}, \charx_{n}, u
\qquad
\conv{\ppath}''(u,w) = u, \conv{\charx}_{n}, u_{n}, \ldots, u_{1}, \conv{\charx}_{1}, w
$$
 We let $\ppath'(z,v)$ denote the path obtained by performing the above operations on $\ppath(z,v)$. $\ppath''(w,u)$ corresponds to the edges $(w,\charx_{1},u_{1}), \ldots, (u_{n},\charx_{n},u) \in \prgr{\rel_{1}}$ and $\conv{\ppath}''(u,w)$ corresponds to the edges $(u,\conv{\charx}_{n},u_{n}), \ldots, (u_{1},\conv{\charx}_{1},w) \in \prgr{\rel_{1}}$, obtained from the relational atoms $w R_{\stra} u \in \rel_{1}$ (by \dfn~\ref{def:propagation-graph}). Since the sole difference between $\prgr{\rel_{1}}$ and $\prgr{\rel_{2}}$ is that the former is guaranteed to contain the edges $(w,\chary,u)$ and $(u,\conv{\chary},w)$ obtained from $w R_{\chary} u$, while the latter is not, and since $\ppath'(z,v)$ omits the use of such edges (i.e. $w, \chary, u$ and $u,\conv{\chary}, w$ do not occur in $\ppath'(z,v)$), we have that $\ppath'(z,v)$ is a propagation path in $\prgr{\rel_{2}}$.

 To complete the proof, we need to additionally show that $\stra_{\ppath'}(z,v) \in \glang(\charx)$. By assumption, $\stra_{\ppath}(z,v) \in \glang(\charx)$, which implies that $\charx \dr \stra_{\ppath}(z,v)$ by \dfn~\ref{def:derivation-language}. Since $(i^{\stra}_{\chary})$ is a rule in $\calc$, it follows that $(\chary \pto \stra), (\conv{\chary} \pto \conv{\stra}) \in \g{\axs}$ by \dfn~\ref{def:grammar}. If we apply $\chary \pto \stra$ to each occurrence of $\chary$ in $\stra_{\ppath}(z,v)$ corresponding to the edge $(w,\chary,u)$ (and relational atom $w R_{\chary} u$), and apply $\conv{\chary} \pto \conv{\stra}$ to each occurrence of $\conv{\chary}$ in $\stra_{\ppath}(z,v)$ corresponding to the edge $(u,\conv{\chary},w)$ (and relational atom $w R_{\chary} u$), we obtain the string $\stra_{\ppath'}(z,v)$ and show that $\charx \dr \stra_{\ppath'}(z,v)$, i.e. $\stra_{\ppath'}(z,v) \in \glang(\charx)$.
\end{proof}

\begin{lemma}\label{lem:permutation-2} Let $\rel_{1} := \rel, w R_{\stra} u, w R_{\chary} u$ and $\rel_{2} := \rel, w R_{\stra} u$. Suppose we are given a proof in $\calc \cup \rcalc$ %\{\prdia, \prbox \ | \ \charx \in \albet\}$ 
 ending with:
\begin{center}
%\resizebox{\columnwidth}{!}{
\begin{tabular}{c}
\AxiomC{$\rel, w R_{\stra} u, w R_{\chary} u, \Gamma, z : \xbox A, v : A \sar t : B$}
\RightLabel{$\prbox$}
\UnaryInfC{$\rel, w R_{\stra} u, w R_{\chary} u,  \Gamma, z : \xbox A, v : A \sar t : B$}
\RightLabel{$(i^{\stra}_{\chary})$}
\UnaryInfC{$\rel, w R_{\stra} u, \Gamma, z : \xbox A \sar t : B$}
\DisplayProof
\end{tabular}
\end{center}
where the side condition $\exists \ppath(z,v) \in \prgr{\rel_{1}} (\stra_{\ppath}(z,v) \in \glang(\charx))$ holds due to $\prdia$. Then, $\exists \ppath'(z,v) \in \prgr{\rel_{2}} (\stra_{\ppath'}(z,v) \in \glang(\charx))$, that is to say, the $(i^{\stra}_{\chary})$ rule is permutable with the $\prbox$ rule.
\end{lemma}

\begin{proof}
Similar to the proof of \lem~\ref{lem:permutation-1} above. 
\end{proof}

\begin{lemma}\label{lem:permutation-3} For each $\charx \in \albet$, the $\convr$ rule is hp-admissible in $\rcalc$.
\end{lemma}

\begin{proof} Let $\charx \in \albet$. We prove the result by induction on the height of the given proof. The base cases are trivial as any application of $\convr$ to $\id$ or $\botl$ yields another instance of the rule. Furthermore, the inductive step follows from the fact that $\convr$ permutes above every rule of $\rcalc$. %with the exception of $\prdia$ and $\prbox$. 
 We show the $\prdia$ case as the remaining cases are simple or similar. %$\prbox$ case is similar.

 Let us suppose that our proof ends with an application of $\prdia$ followed by an application of $(c_{\chary})$ as shown below left. We let $\rel_{1} = \rel, w R_{\chary} u, u R_{\conv{\chary}} w$ and $\rel_{2} = \rel, w R_{\chary} u$. We know that $\exists \ppath(z,v) \in \prgr{\rel_{1}} (\stra_{\ppath}(z,v) \in \glang(\charx))$ by the side condition on $\prdia$. Moreover, we may assume that $u R_{\conv{\chary}} w$ is active in the application of $\prdia$ because the two rules freely permute in the alternative case. If we apply $(c_{\chary})$ first, as shown in the proof below right, observe that since $w R_{\chary} u$ still occurs in the conclusion, $\prgr{\rel_{1}} = \prgr{\rel_{2}}$. Hence, the side condition of $\prdia$ still holds, showing that the rule may be applied after $(c_{\chary})$.
\begin{center}
\begin{tabular}{c @{\hskip 1em} c @{\hskip 1em} c}
\AxiomC{$\rel, w R_{\chary} u, u R_{\conv{\chary}} w, \Gamma \sar v : A$}
\RightLabel{$\prdia$}
\UnaryInfC{$\rel, w R_{\chary} u, u R_{\conv{\chary}} w, \Gamma \sar z : \xdia A$}
\RightLabel{$(c_{\chary})$}
\UnaryInfC{$\rel, w R_{\chary} u, \Gamma \sar z : \xdia A$}
\DisplayProof

&

$\leadsto$

&

\AxiomC{$\rel, w R_{\chary} u, u R_{\conv{\chary}} w, \Gamma \sar v : A$}
\RightLabel{$(c_{\chary})$}
\UnaryInfC{$\rel, w R_{\chary} u, \Gamma \sar v :A$}
\RightLabel{$\prdia$}
\UnaryInfC{$\rel, w R_{\chary} u, \Gamma \sar z : \xdia A$}
\DisplayProof
\end{tabular}
\end{center}
\end{proof}

 To improve the comprehensibility of the above lemmata, we provide an example of permuting instances of structural rules above an instance of a propagation rule.

\begin{example}\label{ex:permutation} Let $\axs := \{(\langle \conv{\chary} \rangle \langle \charz \rangle A \iimp \langle \charx \rangle A) \land ([\charx] A \iimp [\conv{\chary}] [\charz] A)\}$ so that the $\axs$-grammar $\g{\axs} =  \{\charx \pto \conv{\chary} \charz, \conv{\charx} \pto \conv{\charz} \chary\}$. In the top proof below, we let $\rel_{1} = wR_{\charz}v, wR_{\chary}u, uR_{\conv{\chary}}w, uR_{\charx}v$ and assume that $\prdia$ is applied due to the existence of the propagation path $\ppath(u,v) = u, \charx, v$ in $\prgr{\rel_{1}}$, where $\stra_{\ppath}(u,v)  = \charx \in \glang(\charx)$ by \dfn~\ref{def:derivation-language}. The propagation graph $\prgr{\rel_{1}}$ corresponding to the top sequent of the proof shown below left is shown below right:

\begin{center}
\begin{minipage}{.33\textwidth}
\begin{center}
\begin{tabular}{c}
\AxiomC{}
\RightLabel{$\id$}
\UnaryInfC{$wR_{\charz}v, wR_{\chary}u, uR_{\conv{\chary}}w, uR_{\charx}v, v : p \sar v : p$}
\RightLabel{$\prdia$}
\UnaryInfC{$wR_{\charz}v, wR_{\chary}u, uR_{\conv{\chary}}w, uR_{\charx}v, v : p \sar u : \xdia p$}
\RightLabel{$(i^{\conv{\chary}\charz}_{\charx})$}
\UnaryInfC{$wR_{\charz}v, wR_{\chary}u, uR_{\conv{\chary}}w, v : p \sar u : \xdia p$}
\RightLabel{$(c_{\chary})$}
\UnaryInfC{$wR_{\charz}v, wR_{\chary}u, v : p \sar u : \xdia p$}
\DisplayProof
\end{tabular}
\end{center}
\end{minipage}
\begin{minipage}{.3\textwidth}
\ \\
\end{minipage}
\begin{minipage}{.33\textwidth}
\begin{center}
\begin{tabular}{c}
\xymatrix{
w \ar@/^-1pc/@{.>}[dd]|-{\chary} \ar@/^1pc/@{.>}[drr]|-{\charz} &	& \\
 & &  v \ar@/^1pc/@{.>}[ull]|-{\conv{\charz}}\ar@/^1pc/@{.>}[dll]|-{\conv{\charx}} \\
u \ar@/^-1pc/@{.>}[uu]|-{\conv{\chary}}\ar@/^1pc/@{.>}[urr]|-{\charx} & &  
}
\end{tabular}
\end{center}
\end{minipage}
\end{center}

Let $\rel_{2} = wR_{\charz}v, wR_{\chary}u, uR_{\conv{\chary}}w$. If we apply $\charx \pto \conv{\chary} \charz \in \g{\axs}$ to $\stra_{\ppath}(u,v) = \charx$, then we obtain the string $\conv{\chary} \charz$. Hence, $\charx \dr \conv{\chary} \charz$, i.e. $\conv{\chary} \charz \in \glang(\charx)$, meaning that a propagation path $\ppath'(u,v)$ ($= u, \conv{\chary}, w, \charz, v$) exists in $\prgr{\rel_{2}}$ such that $\stra_{\ppath'}(u,v) = \conv{\chary} \charz \in \glang(\charx)$. We may therefore apply $(i^{\conv{\chary}\charz}_{\charx})$ and then $\prdia$ as shown below left; the propagation graph $\prgr{\rel_{2}}$ is shown below right:
\begin{center}
\begin{minipage}{.33\textwidth}
\begin{center}
\begin{tabular}{c} %@{\hskip -.05em} c}
%\vspace*{1 em}
%\ \\
\AxiomC{}
\RightLabel{$\id$}
\UnaryInfC{$wR_{\charz}v, wR_{\chary}u, uR_{\conv{\chary}}w, uR_{\charx}v, v : p \sar v : p$}
\RightLabel{$(i^{\conv{\chary}\charz}_{\charx})$}
\UnaryInfC{$wR_{\charz}v, wR_{\chary}u, uR_{\conv{\chary}}w, v : p \sar v : p$}
\RightLabel{$\prdia$}
\UnaryInfC{$wR_{\charz}v, wR_{\chary}u, uR_{\conv{\chary}}w, v : p \sar u : \xdia p$}
\RightLabel{$(c_{\chary})$}
\UnaryInfC{$wR_{\charz}v, wR_{\chary}u, v : p \sar u : \xdia p$}
\DisplayProof
\end{tabular}
\end{center}
\end{minipage}
\begin{minipage}{.30\textwidth}
\ \\
\end{minipage}
\begin{minipage}{.33\textwidth}
\begin{center}
\begin{tabular}{c}
\xymatrix{
w \ar@/^-1pc/@{.>}[dd]|-{\chary} \ar@/^1pc/@{.>}[drr]|-{\charz} &	& \\
 & &  v \ar@/^1pc/@{.>}[ull]|-{\conv{\charz}} \\
u \ar@/^-1pc/@{.>}[uu]|-{\conv{\chary}} & &  
}
\end{tabular}
\end{center}
\end{minipage}
\end{center}

 Furthermore, if we let $\rel_{3} = wR_{\charz}v, wR_{\chary}u$, then observe that $\prgr{\rel_{2}} = \prgr{\rel_{3}}$, and thus, the propagation path $\ppath'(u,v)$ exists in $\prgr{\rel_{3}}$ as well, showing that $(c_{\chary})$ can be permuted above $\prdia$, as shown below:
\begin{center}
\AxiomC{}
\RightLabel{$\id$}
\UnaryInfC{$wR_{\charz}u, wR_{\chary}v, vR_{\conv{\chary}}w, uR_{\charx}v, v : p \sar v : p$}
\RightLabel{$(i^{\conv{\chary}\charz}_{\charx})$}
\UnaryInfC{$wR_{\charz}v, wR_{\chary}u, uR_{\conv{\chary}}w, v : p \sar v : p$}
\RightLabel{$(c_{\chary})$}
\UnaryInfC{$wR_{\charz}v, wR_{\chary}u, v : p \sar v : p$}
\RightLabel{$\prdia$}
\UnaryInfC{$wR_{\charz}v, wR_{\chary}u, u : p \sar u : \xdia p$}
\DisplayProof
\end{center}
 Since the conclusion of $(c_{\chary})$ is an instance of $\id$, we can delete the applications of $(i^{\conv{\chary}\charz}_{\charx})$ and $(c_{\chary})$, thus giving the proof shown below, which is free of structural rules and exists in $\rcalc$. Thus, we have shown how a proof in $\calc$ can be transformed into a proof in $\rcalc$.
\begin{center}
\AxiomC{}
\RightLabel{$\id$}
\UnaryInfC{$wR_{\charz}v, wR_{\chary}u, v : p \sar v : p$}
\RightLabel{$\prdia$}
\UnaryInfC{$wR_{\charz}v, wR_{\chary}u, u : p \sar u : \xdia p$}
\DisplayProof
\end{center}
\end{example}

\begin{theorem}\label{thm:calc-to-rcalc}
Every proof in $\calc$ can be transformed into a proof in $\rcalc$. % in polynomial-time.
\end{theorem}

\begin{proof} We consider a proof in $\calc$, which is a proof in $\calc \cup \rcalc$ (meaning \lem~\ref{lem:permutation-1} and~\ref{lem:permutation-2} are applicable). By Remark~\ref{rem:diar-boxl-propagation-rule-instances}, each instance of $\xdiar$ and $\xboxl$ can be replaced by a $\prdia$ and $\prbox$ instance, respectively, meaning we may assume our proof is free of $\xdiar$ and $\xboxl$ instances. We show that the proof can be transformed into a proof in $\rcalc$ by induction on its height, that is, we consider a topmost occurrence of a structural rule $\ipar$ or $\convr$ and show that it can be eliminated. Through successively eliminating topmost instances of structural rules, we obtain a proof in $\rcalc$.

\textit{Base case.} The base case is trivial as any application of $\ipar$ or $\convr$ to $\id$ or $\botl$ yields another instance of the rule.

\textit{Inductive step.} It is straightforward to verify that $\ipar$ freely permutes above all rules in $\calc \cup \rcalc$ with the exception of $\ipar$, $\convr$, $\prdia$, and $\prbox$ (this follows from the fact that all other rules do not have active relational atoms in their conclusion). Since we are considering a topmost application of $\ipar$, we need not consider the permutation of $\ipar$ above $\ipar$ or $\convr$. The last two cases of permuting $\ipar$ above $\prdia$ and $\prbox$ follow from \lem~\ref{lem:permutation-1} and~\ref{lem:permutation-2}, respectively. Last, by \lem~\ref{lem:permutation-3}, we know that $\convr$ can be eliminated as we are considering a topmost occurrence of a structural rule, and thus, the proof above $\convr$ is a proof in $\rcalc$. 
 %It is straightforward to verify that this proof transformation (consisting of upward rule permutations) is polynomial-time computable in the size of the input proof.
\end{proof}

The following is a consequence of the above theorem and Theorem~\ref{thm:calc-complete}.

\begin{theorem}[$\rcalc$ Completeness]\label{thm:ikal-complete} If $\Vdash_{\axs}^{\albet} A$, then $\vdash w : A$ is provable in $\rcalc$.
\end{theorem}

\subsection{Soundness}

We now establish a soundness result for each refined labeled calculus $\rcalc$, which relies on a semantics for our labeled sequents, given below. We then leverage this soundness result, along with the proof transformation detailed in \thm~\ref{thm:calc-to-rcalc}, to transfer the soundness of the refined labeled calculi to their parent labeled calculi.

%Interpretation, Satisfiability, Validity
\begin{definition}[Sequent Semantics]\label{def:sequent-semantics-kms} Let $M := (W, \leq, \{R_{\charx} \ | \ \charx \in \albet \}, V)$ be a bi-relational $(\albet,\axs)$-model with $I : \ \lab \mapsto W$ an \emph{interpretation function} mapping labels to worlds.  We define the \emph{satisfaction} of relational atoms and labeled formulae as follows:
\begin{itemize}

\item $M, I \models_{\axs}^{\albet} \rel$ \ifandonlyif for all $wR_{\charx}u \in \rel$, $I(w) R_{\charx} I(u)$;

\item $M, I \models_{\axs}^{\albet} \Gamma$ \ifandonlyif for all $w : A \in \Gamma$, $M, I(w) \Vdash A$.

\end{itemize}

A labeled sequent $\Lambda := \rel, \Gamma \sar w : B$ is \emph{satisfied} in $M$ with $I$, written $M,I \models_{\axs}^{\albet} \Lambda$, \iffi if $M, I \models_{\axs}^{\albet} \rel$ and $M, I \models_{\axs}^{\albet} \Gamma$, then $M, I \models_{\axs}^{\albet} w : B$. A labeled sequent $\Lambda$ is \emph{falsified} in $M$ with $I$ \iffi $M, I \not\models_{\axs}^{\albet} \Lambda$, that is, $\Lambda$ is not satisfied by $M$ with $I$.

Last, a labeled sequent $\Lambda$ is \emph{$(\albet,\axs)$-valid}, written $\models_{\axs}^{\albet} \Lambda$, \iffi it is satisfiable in every bi-relational $(\albet,\axs)$-model $M$ with every interpretation function $I$. We say that a labeled sequent $\Lambda$ is \emph{$(\albet,\axs)$-invalid} \iffi $\not\models_{\axs}^{\albet} \Lambda$, i.e. $\Lambda$ is not $(\albet,\axs)$-valid.
\end{definition}

Before moving on to prove our soundness result, we discuss the structure of our proof as it differs from usual soundness proofs for labeled calculi. The proof proceeds in two stages: first, we show that every proof in $\rcalc$ of a \emph{labeled tree sequent} (cf.~\cite{GorRam12,Lyo21}), which is a labeled sequent that has the structure of a tree, contains only labeled tree sequents, i.e. the proof is a \emph{labeled tree proof} (see \dfn~\ref{def:lab-tree-seq} below). Second, we show that the conclusion of any labeled tree proof in $\rcalc$ is $(\albet,\axs)$-valid. As with intuitionistic mono-modal logics, the restriction to labeled tree proofs in the soundness proof is required due to the $\iimpr$ rule; cf. Simpson~\cite[Section 8.1.3]{Sim94}. An application of $\iimpr$ semantically corresponds to a `jump' along the intuitionistic relation $\leq$ in a bi-relational $(\albet,\axs)$-model (see \thm~\ref{thm:rcalc-sound} below). When both the premise $\rel, \Gamma, w : A \sar w : B$ and conclusion $\rel, \Gamma \sar w : A \iimp B$ of $\iimpr$ are labeled tree sequents, then it is possible to show via the (F1) and (F2) conditions that the satisfaction of $\rel$ can be `lifted' upward along the intuitionistic relation, which suffices for the soundness proof to go through.

%In order to prove this fact, we make use of the following definitions, which are based on the work of~\cite{GorRam12,Lyo21}. 

\begin{definition}[Labeled Tree Sequent/Proof]\label{def:lab-tree-seq}
We define a \emph{labeled tree sequent} to be a labeled sequent $\Lambda := \rel, \Gamma \sar \Delta$ such that $\rel$ forms a tree and all labels in $\Gamma, \Delta$ occur in $\rel$ (unless $\rel$ is empty, in which case every labeled formula in $\Gamma, \Delta$ must share the same label). 

 We define a \emph{labeled tree proof} to be a proof containing only labeled tree sequents. We say that a labeled tree proof has the \emph{fixed root property} \iffi every labeled sequent in the proof has the same root.
\end{definition}

\begin{lemma}\label{lem:labeled-tree-derivations}
Every proof in $\rcalc$ of a labeled tree sequent is a labeled tree proof with the fixed root property.
\end{lemma}

\begin{proof} Let us consider a proof in $\rcalc$ of a labeled tree sequent in a bottom-up fashion. With the exception of $\xdial$ and $\xboxr$, every rule of $\rcalc$ bottom-up preserves the relational atoms present in the conclusion of the rule (only changing the formulae associated with labels), showing that if the conclusion is a labeled tree sequent, then the premises will be labeled tree sequents. In the case of $\xdial$ and $\xboxr$, a new relational atom is added to the premise, thus creating a new branching edge due to the freshness condition of the rule, but demonstrating that the premise is a labeled tree sequent nonetheless. Hence, the property of being a labeled tree sequent will be inherited by all labeled sequents occurring in the proof, and since no rule of $\rcalc$ changes the `root' of the conclusion as only forward edges are bottom-up added to labeled tree sequents, the proof will have the fixed-root property.
\end{proof}

%Graph morphism (used for soundness) is defined on pg. 76 Simpson
%Pg. 127 Thm. 7.2.1 Simpsoin Completeness/Soundness of Sequent Calc wrt to ND
%Pg. 156 Th. 8.1.4 Simpson Completeness/Soundness of ND wrt Axiomatization
\begin{theorem}[$\rcalc$ Soundness]\label{thm:rcalc-sound}
Let $\rel, \Gamma \vdash w : A$ be a labeled tree sequent. If $\rel, \Gamma \vdash w : A$ is provable in $\rcalc$, then $\models^{\albet}_{\axs} \rel, \Gamma \vdash w : A$.
\end{theorem}

\begin{proof} We prove the result by induction on the height of the given proof, which is a labeled tree proof by \lem~\ref{lem:labeled-tree-derivations}. As the base cases are trivial, we omit them, and with the exception of $\iimpr$, all cases of the inductive step are straightforward and are similar to those given for classical grammar logics~\cite[\thm~5]{Lyo21thesis}. Hence, we prove the $\iimpr$ case of the inductive step.

 We argue by contraposition, that is, we show that if the conclusion is $(\albet,\axs)$-invalid, then the premise is $(\albet,\axs)$-invalid. Let $M$ be a $(\albet,\axs)$-model and $I$ an interpretation function such that $M, I \not\models_{\axs}^{\albet} \rel, \Gamma \sar w : A \iimp B$. Then, $M, I \not\models_{\axs}^{\albet} w : A \iimp B$, which implies that there exists a world $u \in W$ of $M$ such that $I(w) \leq u$, $M, u \Vdash_{\axs}^{\albet} A$, and $M, u \not\Vdash_{\axs}^{\albet} B$. We now define a new interpretation $I'$ such that $M, I' \not\models_{\axs}^{\albet} \rel, \Gamma, w : A \sar w : B$. 

First, we set $I'(w) = u$. Second, we let $wR_{\charx_{1}}v_{1}, \ldots, wR_{\charx_{n}}v_{n}$ be all relational atoms in $\rel$ of the form $wR_{\charx}v$ (containing all successors of $w$), and $zR_{\chary}w$ be the relational atom in $\rel$ with $z$ the predecessor of $w$ (if $w$ is the root of $\rel$, then $w$ does not have any predecessors, so skip this step; recall that $\rel$ forms a tree). %Let us now pick an arbitrary $wR_{\charx_{i}}v_{i}$ from the first collection. 
 Since $I(w) R_{\charx_{i}} I(v_{i})$ and $I(z) R_{\chary} I(w)$ hold in $M$ and $I(w) \leq u$, we know there exists a $v_{i}'$ and $z'$ such that $I(v_{i}) \leq v_{i}'$ and $uR_{\charx_{i}}v_{i}'$ by condition (F2), and $I(z) \leq z'$ and $z'R_{\chary}u$ by condition (F1) (see \dfn~\ref{def:bi-relational-model} for these conditions). Thus, if we set $I'(v_{i}) = v_{i}'$ and $I'(z) = z'$, then it follows that $I'(w) R_{\charx_{i}} I'(v_{i})$ and $I'(z') R_{\chary} I'(w)$. Moreover, since $M, I \models_{\axs}^{\albet} v_{i} : \Gamma$ and $M, I \models_{\axs}^{\albet} z : \Gamma$, it follows that $M, I' \models_{\axs}^{\albet} v_{i} : \Gamma$ and $M, I' \models_{\axs}^{\albet} z : \Gamma$ by \lem~\ref{lem:persistence} as $I(v_{i}) \leq I'(v_{i})$ and $I(z) \leq I'(z)$. We %continue in this fashion, defining $I'$ for all $R_{\charx}$-successors and predecessors of $w$, for all $\charx \in \albet$, and then 
 successively repeat this process %for each successor and predecessor, 
 until all descendants and ancestors are of $w$ in $\rel$ are defined for $I'$. %For all remaining labels $z \in \lab(\rel, \Gamma \sar w : A \iimp B)$, we let $I'(w) = I(z)$.

One can confirm that $M, I' \models_{\axs}^{\albet} \rel, \Gamma$, and since $I'(w) = u$, we also know that $M, I \models_{\axs}^{\albet} w : A$, but $M, I \not\models_{\axs}^{\albet} w : B$. Hence, the premise of $\iimpr$ is $(\albet,\axs)$-invalid.
\end{proof}

\begin{theorem}[$\calc$ Soundness]\label{thm:calc-sound}
Let $\rel, \Gamma \vdash w : A$ be a labeled tree sequent. If $\rel, \Gamma \vdash w : A$ is provable in $\calc$, then $\models^{\albet}_{\axs} \rel, \Gamma \vdash w : A$.
\end{theorem}

\begin{proof} Follows from Theorems~\ref{thm:calc-to-rcalc} and~\ref{thm:rcalc-sound}.
\end{proof}

\subsection{Additional Properties}

In the last part of this section, we show how to transforms proofs from $\rcalc$ into proofs in $\calc$. This backward transformation lets us transfer the admissibility and invertibility properties of the latter calculi to the former calculi as witnessed in \cor~\ref{cor:properties-rcalc}.

\begin{lemma}\label{lem:deleting-rel-atoms}
If $\rel, uR_{\charx}v, \Gamma \vdash w : A$ is provable in $\calc$ and $\exists \ppath(u,v) \in \prgr{\rel} (\stra_{\ppath}(u,v) \in \glang(\charx))$, then $\rel, \Gamma \vdash w : A$ is provable in $\calc$.
\end{lemma}

\begin{proof} By our assumption that $\exists \ppath(u,v) \in \prgr{\rel} (\stra_{\ppath}(u,v) \in \glang(\charx))$, we know that $\charx \dr \stra_{\ppath}(u,v)$. We show the result by induction on the length of the \cfg-derivation $\charx \dr \stra_{\ppath}(u,v)$.\footnote{The length of a \cfg-derivation in an $\axs$-grammar is defined in \dfn~\ref{def:derivation-language}.}

\textit{Base case.} If the \cfg-derivation $\charx \dr \stra_{\ppath}(u,v)$ has length zero, then it follows that $\charx \osdr \charx$, implying that $\rel, uR_{\charx}v, uR_{\charx}v, \Gamma \vdash w : A$ is provable in $\calc$. A single application of the hp-admissible $\ctrli$ rule (see \lem~\ref{lem:ctr-admiss}) gives the desired result. If the \cfg-derivation $\charx \dr \stra_{\ppath}(u,v)$ has a length of one, then $\charx \pto \stra_{\ppath}(u,v)$ is a production rule in $\g{\axs}$. Hence, $\rel, uR_{\charx}v, \Gamma \vdash w : A$ is of the form $\rel',uR_{\stra'}v, uR_{\charx}v, \Gamma \vdash w : A$, where the relational atoms $uR_{\stra'}v$ give rise to the propagation path $\ppath(u,v)$. We note that $\rel',uR_{\stra}v, uR_{\charx}v, \Gamma \vdash w : A$ can be derived from the above labeled sequent by applying the hp-admissible $\wkl$ rule (see \lem~\ref{lem:wk-sub-admiss}) and the $\convr$ rules a sufficient number of times, from which $\rel',uR_{\stra}v \Gamma \vdash w : A$ is provable by a single application of $\ipar$. Again, by applying $\wkl$ and the $\convr$ rules a sufficient number of times, we derive $\rel, \Gamma \sar w : A$.

\textit{Inductive step.} Suppose the \cfg-derivation $\charx \dr \stra_{\ppath}(u,v)$ has a length of $n+1$. This implies the existence of a string $\strb$ such that $\charx \dr \strb$ and $\strb \osdr \stra_{\ppath}(u,v)$. The former is a \cfg-derivation with length $n$ and the latter is a one-step \cfg-derivation. Therefore, there exists a production rule $\chary \pto \strc \in \g{\axs}$ such that $\strb = \strb_{0}\chary\strb_{1}$ and $\stra_{\ppath}(u,v) = \strb_{1}\strc\strb_{2}$, which implies the existence of a propagation path $\ppath_{\strb_{1}}(u,z), \ppath_{\strc}(z,z'), \ppath_{\strb_{1}}(z',v)$ in $\rel$. By applying the hp-admissible $\wkl$ rule to the provable labeled sequent $\rel, uR_{\charx}v, \Gamma \vdash w : A$, we obtain a proof of $\rel, zR_{\chary}z', uR_{\charx}v, \Gamma \vdash w : A$. Since the \cfg-derivation $\charx \dr t = \strb_{0}\chary\strb_{1}$ has a length $n$ and the propagation path $\ppath_{\strb_{1}}(u,z), \ppath_{\chary}(z,z'), \ppath_{\strb_{1}}(z',v)$ occurs in $\rel,zR_{\chary}z'$, we may invoke the inductive hypothesis to derive $\rel, zR_{\chary}z', \Gamma \vdash w : A$. Last, by an argument similar to the base case, we may derive $\rel, \Gamma \sar w : A$ since $\ppath_{\strc}(z,z')$ occurs in $\rel$.
\end{proof}

\begin{theorem}\label{thm:rcalc-to-calc}
Every proof in $\rcalc$ can be transformed into a proof in $\calc$.
\end{theorem}

\begin{proof} The theorem is shown by induction on the height of the given proof. We show the non-trivial $\prdia$ case of the inductive step as the remaining cases are simple or similar.
\begin{center}
\begin{tabular}{c @{\hskip 1em} c @{\hskip 1em} c}
\AxiomC{$\rel, \Gamma \sar u : A$}
\RightLabel{$\prdia$}
\UnaryInfC{$\rel, \Gamma \sar w : \xdia A$}
\DisplayProof

&

$\leadsto$

&

\AxiomC{$\rel, \Gamma \sar u : A$}
\RightLabel{$\wkl$}
\UnaryInfC{$\rel, wR_{\charx}u, \Gamma \sar u : A$}
\RightLabel{$\xdiar$}
\UnaryInfC{$\rel, wR_{\charx}u, \Gamma \sar w : \xdia A$}
\RightLabel{\lem~\ref{lem:deleting-rel-atoms}}
\UnaryInfC{$\rel, \Gamma \sar w : \xdia A$}
\DisplayProof
\end{tabular}
\end{center}
 By the side condition on $\prdia$, we know that \lem~\ref{lem:deleting-rel-atoms} is applicable in the proof shown above right.
\end{proof}

\begin{corollary}\label{cor:properties-rcalc} The following properties hold for $\rcalc$:
\begin{enumerate}

\item The $\botr$, $\sub$, $\wkl$, $\ctrli$, $\ctrlii$, and $\cut$ rules are admissible;

\item The $\disl$, $\conl$, $\xdial$, and $\prbox$ rules are invertible;

\item The $\iimpl$ rule is invertible in the right premise.

\end{enumerate}
\end{corollary}

\begin{proof} Follows from Theorems~\ref{thm:calc-to-rcalc} and~\ref{thm:rcalc-to-calc}. For the admissible rules, suppose the premise(s) is (are) provable in $\rcalc$. Then, by \thm~\ref{thm:rcalc-to-calc}, these labeled sequents are provable in $\calc$. Since all such rules are admissible in $\calc$, the rule may be applied and the desired conclusion is provable in $\rcalc$ by \thm~\ref{thm:calc-to-rcalc}. The invertibility results are argued in a similar fashion.
\end{proof}

%Translation to Nested Sequents
\section{Nested Sequent Systems}\label{sec:nested-calculi}

%NOTES:
% - Mention that we combine notation from Lutz 2013 and Alwen et al. 2012
%- might want to show that propagation rules subsume behavior of two input diamond rules and two output box rules (which are analogous to the rules with Alwen, Raj, and Linda's paper)
%- Might want to mention how having one formulae corresponds to having a single labeled formulae on the right in Simoson's formalism

 In our setting, nested sequents are taken to be trees of multisets of formulae containing a unique formula that occupies a special status. We utilize a syntax which incorporates components from the nested sequents given for classical grammar logics~\cite{TiuIanGor12} and intuitionistic modal logics~\cite{Str13}. Following~\cite{Str13}, we mark a special, unique formula with a white circle $\outp$ indicating that the formula is of \emph{output polarity}, and mark the other formulae with a black circle $\inp$ indicating that the formulae are of \emph{input polarity}. A nested sequent $\ns$ is defined via the following BNF grammars:
%\begin{definition}[Nested Sequent~\cite{Str13}]
$$
\ns ::= \Delta, \Pi \qquad \Delta ::= A_{1}^{\inp}, \ldots, A_{n}^{\inp}, (\charx_{1})\bl \Delta_{1} \br, \ldots, (\charx_{k})\bl \Delta_{k} \br \qquad \Pi ::= A^{\outp} \ | \ (\charx)\bl \ns \br
$$
%\end{definition}
 where $\charx_{1}, \ldots, \charx_{k}, \charx \in \albet$ and $A_{1}, \ldots, A_{n}, A \in \lang{\albet}$.

We assume that the comma operator associates and commutes, implying that such sequents are truly trees of multisets of formulae, and we let the \emph{empty sequent} be the empty multiset $\empseq$. We refer to a sequent in the shape of $\Delta$ (which contains only input formulae) as an \emph{LHS-sequent}, a sequent in the shape of $\Pi$ %%%(which is either an output formula or a sequent $\bl \ns \br$) 
 as an \emph{RHS-sequent}, and a sequent $\ns$ as a \emph{full sequent}. We use both $\ns$ and $\Delta$ to denote LHS- and full sequents with the context differentiating the usage.

As for classical modal logics (e.g.~\cite{Bru09,GorPosTiu11}), we define a \emph{context} $\ns\{ \ \}\cdots\{ \ \}$ to be a nested sequent with some number of holes $\{ \ \}$ in the place of formulae. This gives rise to two types of contexts: \emph{input contexts}, which require holes to be filled with LHS-sequents to obtain a full sequent, and \emph{output contexts}, which require a single hole to be filled with an RHS-sequent and the remaining holes to be filled with LHS-sequents to obtain a full sequent. We also define the \emph{output pruning} of an input context $\ns\{ \ \}\cdots\{ \ \}$ or full sequent $\ns$, denoted $\ns^{\downarrow}\{ \ \}\cdots\{ \ \}$ and $\ns^{\downarrow}$ respectively, to be the same context or sequent with the unique output formula deleted. We note that all of the above terminology is due to~\cite{Str13}.

\begin{example} Let $\ns_{1}\{ \ \} := p^{\inp}, (\chara)\bl \adia q^{\inp}, \{ \ \} \br$ be an output context and $\ns_{2}\{ \ \}:= p^{\inp}, (\conv{\charc})\bl \adia q^{\outp}, \{ \ \} \br$ be an input context. Also, let $\Delta_{1} := \bot^{\inp}, (\charb)\bl q \iimp r^{\outp} \br$ be a full sequent and $\Delta_{2} := \bot^{\inp}, (\chara)\bl q \iimp r^{\inp} \br$ be an LHS-sequent. Observe that neither $\ns_{1}\{\Delta_{2}\}$ (shown below left) nor $\ns_{2}\{\Delta_{1}\}$ (shown below right) are full sequents since the former has no output formula and the latter has two output formulae.
$$
p^{\inp}, (\chara)\bl \adia q^{\inp}, \bot^{\inp}, (\chara)\bl q \iimp r^{\inp} \br \br
\qquad
p^{\inp}, (\conv{\charc})\bl \adia q^{\outp}, \bot^{\inp}, (\charb)\bl q \iimp r^{\outp} \br \br
$$
Conversely, both $\ns_{1}\{\Delta_{1}\}$ (shown below left) and $\ns_{2}\{\Delta_{2}\}$ (shown below right) are full sequents.
$$
p^{\inp}, (\chara)\bl \adia q^{\inp}, \bot^{\inp}, (\charb)\bl q \iimp r^{\outp} \br \br
\qquad
p^{\inp}, (\conv{\charc})\bl \adia q^{\outp}, \bot^{\inp}, (\chara)\bl q \iimp r^{\inp} \br \br
$$
\end{example}

%For more information on nested sequents of the above form, consult~\cite{Sta13}.

Our nested sequent systems are presented in \fig~\ref{fig:nested-calculus} and are generalizations of those given for extensions of the intuitionistic modal logic $\ik$ with seriality and Horn-Scott-Lemmon axioms of the form $(\dia^{n} \Box A \iimp \Box^{k} A) \land (\dia^{k} A \iimp \Box^{n} \dia A)$ with $n,k \in \mathbb{N}$~\cite{Lyo21b}, which includes the logics of the intuitionistic modal cube~\cite{Str13}. %We will discuss later on how intuitionistic mono-modal logics (e.g. $\ik$~\cite{PloSti86}) can be captured with fragments of our calculi by means of the \emph{separation property}~\cite{GorPosTiu11}. 
 These systems can also be seen as intuitionistic variants of the nested sequent systems given for classical grammar logics~\cite{TiuIanGor12}. %For example, a nested sequent system for the intuitionistic modal logic $\ik + \{(\dia^{0} \Box A \iimp \Box^{3} A) \land (\dia^{3} A \iimp \Box^{0} \dia A)\}$ incorporating the 3-to-1 transitivity axiom, which falls outside the intuitionistic modal cube, is obtained by employing the $\axs$-grammar $\g{\axs} = \{\fd \pto \fd \fd \fd, \bd \pto \bd \bd \bd \}$ in the propagation rules $\prdia$ and $\prbox$. 
 Our nested systems incorporate the propagation rules $\prdia$ and $\prbox$, which rely on auxiliary notions (e.g. propagation graphs and paths) that we now define. %as our propagation rules of the previous section. %and which we now define for nested sequents.

\begin{definition}[Propagation Graph/Path]\label{def:propagation-graph} Let $w$ be the label assigned to the root of the nested sequent $\ns$. We define the \emph{propagation graph} $PG(\ns) := PG_{w}(\ns) = (V,E,L)$ of a nested sequent $\ns$ recursively on the structure of the nested sequent.
\begin{itemize}

\item $PG_{u}(\empseq) := (\emptyset, \emptyset, \emptyset)$;

\item $PG_{u}(A) := (\{u\}, \emptyset, \{(u,A)\}) \text{ with } A \in \{A^{\inp}, A^{\outp}\}$;

\item $PG_{u}(\Delta_{1},\Delta_{2}) := (V_{1}\cup V_{2}, E_{1} \cup E_{2}, L_{1} \cup L_{2}) \text{ where } PG_{u}(\Delta_{i}) = (V_{i},E_{i},L_{i})$;

\item $PG_{u}((\charx)\bl \ns \br) := (V' \cup \{u\}, E' \cup \{(u,\charx,v),(v,\conv{\charx},u)\},L') \text{ where $v$ is a fresh label }\\ \text{ and } PG_{v}(\ns) = (V',E',L')$.

\end{itemize}
We will often write $u \in \prgr{\ns}$ to mean $u \in \prgrdom$, and $(u,\charx,v) \in \prgr{\ns}$ to mean $(u,\charx,v) \in \prgredges$. We define \emph{propagation paths}, \emph{converses} of propagation paths, and \emph{strings} of propagation paths as in \dfn~\ref{def:propagation-path}.
\end{definition}

For input or output formulae $A$ and $B$, we use the notation $\ns\{A\}_{w}$ and $\ns\{A\}_{w}\{B\}_{u}$ to mean that $(w,A) \in L$ and $(w,A),(u,B) \in L$ in $\prgr{\ns}$, respectively. For example, if $\ns := p \iimp q^{\outp}, (\charb)\bl p^{\inp}, (\conv{\chara})\bl \abox p^{\inp} \br \br$, $\prgr{\ns} := (V,E,L)$, and $(v,p \iimp q^{\outp}),(u, p^{\inp}), (w,\abox p^{\inp}) \in L$, then both $\ns\{p \iimp q^{\outp}\}_{v}\{\abox p^{\inp}\}_{w}$ and $\ns\{p^{\inp}\}_{u}\{p \iimp q^{\outp}\}_{v}$ are valid representations of $\ns$ in our notation.
%\begin{center}
%\xymatrix{
% & &  \overset{\boxed{p}}{w}\ar[drr]\ar[dll]\ar@/^-1pc/@{.>}[dll]|-{\fd}\ar@/^1pc/@{.>}[drr]|-{\fd} & & \\
%  \overset{\boxed{\bot}}{u} \ar@/^-1pc/@{.>}[urr]|-{\bd} & &  & & \overset{\boxed{\dia q}}{v} \ar@/^1pc/@{.>}[ull]|-{\bd}
%}
%\end{center}

\begin{figure}[t]
\noindent\hrule

\begin{center}
\begin{tabular}{c c c} % @{\hskip 1em} c}
\AxiomC{}
\RightLabel{$\id$}
\UnaryInfC{$\ns \sbl p^{\inp}, p^{\outp} \sbr$}
\DisplayProof

&

\AxiomC{}
\RightLabel{$\botin$}
\UnaryInfC{$\ns \sbl \bot^{\inp} \sbr$}
\DisplayProof

&

\AxiomC{$\ns \sbl A^{\inp}, B^{\inp} \sbr$}
\RightLabel{$\conin$}
\UnaryInfC{$\ns \sbl A \land B^{\inp} \sbr$}
\DisplayProof
\end{tabular}
\end{center}

\begin{center}
\begin{tabular}{c c c}
\AxiomC{$\ns \sbl A^{\outp} \sbr$}
\AxiomC{$\ns \sbl B^{\outp} \sbr$}
\RightLabel{$\conout$}
\BinaryInfC{$\ns \sbl A \land B^{\outp} \sbr$}
\DisplayProof

&

\AxiomC{$\ns \sbl (\charx)\bl \empseq \br \sbr$}
\RightLabel{$\ddr$}
\UnaryInfC{$\ns \sbl \empseq \sbr$}
\DisplayProof

&

\AxiomC{$\ns \sbl A_{i}^{\outp} \sbr$}
\RightLabel{$\disout~i \in \{1,2\}$}
\UnaryInfC{$\ns \sbl A_{1} \lor A_{2}^{\outp} \sbr$}
\DisplayProof
\end{tabular}
\end{center}

\begin{center}
\begin{tabular}{c c c}
%\AxiomC{$\ns \sbl A^{\inp}, \inot A^{\outp} \sbr$}
%\RightLabel{$\negout$}
%\UnaryInfC{$\ns \sbl \inot A^{\outp} \sbr$}
%\DisplayProof
%&
%\AxiomC{$\ns^{\da} \sbl \inot A^{\inp}, A^{\outp} \sbr$}
%\RightLabel{$\negin$}
%\UnaryInfC{$\ns \sbl \inot A^{\inp} \sbr$}
%\DisplayProof
\AxiomC{$\ns \sbl A^{\inp} \sbr$}
\AxiomC{$\ns \sbl B^{\inp} \sbr$}
\RightLabel{$\disin$}
\BinaryInfC{$\ns \sbl A \lor B^{\inp} \sbr$}
\DisplayProof

&

\AxiomC{$\ns \sbl (\charx)\bl A^{\inp} \br \sbr$}
\RightLabel{$\xdiain$}
\UnaryInfC{$\ns \sbl \xdia A^{\inp} \sbr$}
\DisplayProof

&

\AxiomC{$\ns \sbl A^{\inp}, B^{\outp} \sbr$}
\RightLabel{$\iimpout$}
\UnaryInfC{$\ns \sbl A \iimp B^{\outp} \sbr$}
\DisplayProof
\end{tabular}
\end{center}

\begin{center}
\begin{tabular}{c c}
\AxiomC{$\ns^{\downarrow} \sbl A \iimp B^{\inp}, A^{\outp} \sbr$}
\AxiomC{$\ns \sbl B^{\inp} \sbr$}
\RightLabel{$\iimpin$}
\BinaryInfC{$\ns \sbl A \iimp B^{\inp} \sbr$}
\DisplayProof

&

\AxiomC{$\ns \sbl (\charx)\bl A^{\outp} \br \sbr$}
\RightLabel{$\xboxout$}
\UnaryInfC{$\ns \sbl \xbox A^{\outp} \sbr$}
\DisplayProof
\end{tabular}
\end{center}

\begin{center}
\begin{tabular}{c c}
\AxiomC{$\ns \sbl \Delta_{1} \sbr_{w} \sbl A^{\outp}, \Delta_{2} \sbr_{u}$}
\RightLabel{$\prdia^{\dag}$}
\UnaryInfC{$\ns \sbl \xdia A^{\outp}, \Delta_{1} \sbr_{w} \sbl \Delta_{2} \sbr_{u}$}
\DisplayProof

&

%\resizebox{\columnwidth}{!}{
\AxiomC{$\ns \sbl \xbox A^{\inp}, \Delta_{1} \sbr_{w} \sbl A^{\inp}, \Delta_{2} \sbr_{u}$}
\RightLabel{$\prbox^{\dag}$}
\UnaryInfC{$\ns \sbl \xbox A^{\inp}, \Delta_{1} \sbr_{w} \sbl \Delta_{2} \sbr_{u}$}
\DisplayProof
\end{tabular}
\end{center}
$$
\dag := \exists \ppath(w,u) \in \prgr{\ns} (\stra_{\ppath}(w,u) \in \glang(\charx))
$$

\hrule
\caption{The nested sequent calculi $\ncalc$. We have a copy of $\xdiain$, $\xboxout$, $\prdia$, and $\prbox$ for each $\charx \in \albet$, and the $\ddr$ rule occurs in a calculus $\ncalc$ \ifandonlyif $\D \in \axs$. The $\prdia$ and $\prbox$ rule are applicable only if the side condition $\dag$ holds.}
\label{fig:nested-calculus}
\end{figure}

% We now prove that proofs can be translated between our refined labeled and nested systems.

We now define our translation functions which transform a \emph{full} nested sequent into a labeled tree sequent, and vice-versa. Our translations additionally depend on \emph{sequent compositions} and \emph{labeled restrictions}. If $\Lambda_{1} := \rel_{1}, \Gamma_{1} \sar \Gamma_{1}'$ and $\Lambda_{2} := \rel_{2}, \Gamma_{2} \sar \Gamma_{2}'$, then we define its sequent composition $\Lambda_{1} \seqcomp \Lambda_{2} := \rel_{1}, \rel_{2},\Gamma_{1}, \Gamma_{2} \sar \Gamma_{1}', \Gamma_{2}'$. Given that $\Gamma$ is a multiset of labeled formulae, we define the labeled restriction $\Gamma \restriction w := \{A \ | \ w : A \in \Gamma\}$, and if $w$ is not a label in $\Gamma$, then $\Gamma \restriction w := \emptyset$. %Moreover, for $\ast \in \{\inp,\outp\}$, we recursively define $(\emptyset)^{\ast} := \emptyset$ and $(A_{1}, \ldots, A_{n-1}, A_{n})^{\ast} := (A_{1}, \ldots,A_{n-1})^{\ast}, A_{n}^{\ast}$. 
 Moreover, for a multiset $A_{1}, \ldots,A_{n}$ of formulae, we define $(A_{1}, \ldots,A_{n})^{\ast} := A_{1}^{\ast}, \ldots,A_{n}^{\ast}$ and $(\emptyset)^{\ast} := \emptyset$, where $\ast \in \{\inp,\outp\}$.

\begin{definition}[Translation $\ltr$] We define $\ltr_{w}(\ns) := \rel, \Gamma \sar u : A$ as follows:

\begin{center}
\begin{minipage}{.4\textwidth}
\begin{itemize}

\item $\ltr_{v}(\empseq) := \emptyset \sar \emptyset$

\item $\ltr_{v}(A^{\inp}) := v : A \sar \empseq$

\item $\ltr_{v}(A^{\outp}) := \empseq \sar v : A$

\end{itemize}
\end{minipage}
\begin{minipage}{.55\textwidth}
\begin{itemize}

\item $\ltr_{v}(\Delta_{1},\Delta_{2}) := \ltr_{v}(\Delta_{1}) \seqcomp \ltr_{v}(\Delta_{2})$

\item $\ltr_{v}((\charx)\bl \ns \br) := (vR_{\charx}z \sar \empseq) \seqcomp \ltr_{z}(\ns)$\\ $\text{ with $z$ fresh}$

\end{itemize}
\end{minipage}
\end{center}
We note that since $\ns$ is a full sequent, the obtained labeled sequent will contain a single labeled formula in its consequent.
\end{definition}

\begin{example} We let $\ns := p \iimp q^{\outp}, (b)\bl p^{\inp}, (\conv{c})\bl \abox p^{\inp} \br \br$ and show the output labeled sequent under the translation $\ltr$.
$$
\ltr_{w}(\Sigma) = wR_{b}v,vR_{\conv{c}}u, v : p, u : \abox p \sar w : p \iimp q
$$
\end{example}

 Given that $\Lambda := \rel, \Gamma \sar w : A$ is a labeled tree sequent, we define $\Lambda' \subseteq \Lambda$ \iffi there exists a labeled tree sequent $\Lambda''$ such that $\Lambda = \Lambda' \seqcomp \Lambda''$. We then define $\Lambda_{u} := \rel', \Gamma' \sar \Delta'$ to be the unique labeled tree sequent rooted at $u$ such that $\Lambda_{u} \subseteq \Lambda$, $\Gamma' \restriction u = \Gamma \restriction u$, $\Delta' \restriction u = \Delta \restriction u$, and $\rel'$ is the downward closure of $u$ in $\rel$. Using this notation, we define the reverse translation accordingly.

\begin{definition}[Translation $\ntr$] % Let $\Lambda := \rel, \Gamma \sar w : A$ be a labeled tree sequent with root $u$. We define $\Lambda' \subseteq \Lambda$ \iffi there exists a labeled tree sequent $\Lambda''$ such that $\Lambda = \Lambda' \seqcomp \Lambda''$. Let us define $\Lambda_{u} := \rel', \Gamma' \sar \Delta'$ to be the unique labeled tree sequent rooted at $u$ such that $\Lambda_{u} \subseteq \Lambda$, $\Gamma' \restriction u = \Gamma \restriction u$, and $\Delta' \restriction u = \Delta \restriction u$. 
 Let $\Lambda := \rel, \Gamma \sar w : A$. We define the translation $\ntr$ recursively on the tree structure of its input, starting at the root; in particular, $\ntr(\Lambda) := \ntr_{u}(\Lambda) = $
\[
  \begin{cases}
  (\Gamma \restriction u)^{\inp}, (\{w : A\} \restriction u)^{\outp} & \text{if $\rel = \empseq$}; \\
  (\Gamma \restriction u)^{\inp}, (\{w : A\} \restriction u)^{\outp}, (\charx_{1})\bl \ntr_{z_{1}}(\Lambda_{z_{1}}) \br, \ldots, (\charx_{n})\bl \ntr_{z_{n}}(\Lambda_{z_{n}}) \br & \text{otherwise}. 
  \end{cases}
\]
In the second case above, we assume that $vR_{\charx_{1}}z_{1}, \ldots vR_{\charx_{n}}z_{n}$ are all of the relational atoms occurring in the input sequent which have the form $vR_{\chary}x$.
\end{definition}

\begin{example} We let $\Lambda := wR_{b}v,vR_{\conv{c}}u, v : p, u : \abox p \sar w : p \iimp q$ and show the output nested sequent under the translation $\ntr$.
$$
\ntr(\Lambda) = p \iimp q^{\outp}, (b)\bl p^{\inp}, (\conv{c})\bl \abox p^{\inp} \br \br
$$
\end{example}

\begin{theorem}\label{thm:nested-labeled-equiv}
Every proof of a labeled tree sequent in $\rcalc$ is transformable into a proof in $\ncalc$, and vice-versa.
\end{theorem}

\begin{proof} Follows from \lem~\ref{lem:labeled-tree-derivations}, and the fact that the rules of $\rcalc$ and $\ncalc$ are translations of one another under the $\ntr$ and $\ltr$ functions. 
\end{proof}

\begin{theorem}[$\ncalc$ Soundness and Completeness]\label{thm:nested-sound-complete}
 A formula $A$ is provable in $\ncalc$ \iffi $A$ is $(\albet,\axs)$-valid.
\end{theorem}

\begin{proof}
Follows from Theorems~\ref{thm:ikal-complete}, \ref{thm:rcalc-sound}, and~\ref{thm:nested-labeled-equiv}. 
\end{proof}

 We end this section by establishing a collection of proof-theoretic properties satisfied by each nested calculus $\ncalc$. In particular, we argue that certain rules of $\ncalc$ are hp-invertible and that the rules displayed in \fig~\ref{fig:structural-rules} are hp-admissible.

\begin{figure}[t]
\noindent\hrule

\begin{center}
\begin{tabular}{c c c}
\AxiomC{$\ns \sbl \bot^{\outp} \sbr$}
\RightLabel{$\botout$}
\UnaryInfC{$\ns \sbl A^{\outp} \sbr$}
\DisplayProof

&

\AxiomC{$\ns$}
\RightLabel{$\nec$}
\UnaryInfC{$(\charx) \bl \ns \br$}
\DisplayProof

&

\AxiomC{$\ns \sbl \empseq \sbr$}
\RightLabel{$\wk$}
\UnaryInfC{$\ns \sbl \Delta \sbr$}
\DisplayProof
\end{tabular}
\end{center}

\begin{center}
\begin{tabular}{c c}
\AxiomC{$\ns \sbl (\charx)\bl \Delta_{1}\br,  (\charx)\bl \Delta_{2} \br \sbr$}
\RightLabel{$\med$}
\UnaryInfC{$\ns \sbl (\charx)\bl \Delta_{1}, \Delta_{2} \br \sbr$}
\DisplayProof

&

\AxiomC{$\ns \sbl A^{\inp}, A^{\inp} \sbr$}
\RightLabel{$\ctr$}
\UnaryInfC{$\ns \sbl A^{\inp} \sbr$}
\DisplayProof
\end{tabular}
\end{center}

\hrule
\caption{Hp-admissible rules.} %(cf.~\ref{Str13}).}
\label{fig:structural-rules}
\end{figure}

\begin{theorem}\label{thm:nested-properties}
 The following properties hold for $\ncalc$:
\begin{enumerate}

\item The $\botout$, $\nec$, and $\wk$ rules are hp-admissible;

\item The $\disin$, $\conin$, $\xdiain$, and $\prbox$ rules are hp-invertible;

\item The $\iimpin$ rule is hp-invertible in the right premise;

\item The $\med$ and $\ctr$ rules are hp-admissible.

\end{enumerate}
\end{theorem}

\begin{proof} Every claim is shown by induction on the height of the given proof. As the proofs of claims 1 - 3 are standard, we omit them, and only show the proof of claim 4. We first argue the hp-admissibility of $\med$, and then use this to demonstrate the hp-admissibility of $\ctr$.

 We note that the base cases are trivial, and with the exception of the $\prdia$ and $\prbox$ rules, all cases of the inductive step are trivial as $\ctr$ freely permutes above each rule instance. Regarding the $\prdia$ and $\prbox$ cases, as discussed in Gor\'e et al.~\cite[\fig~12]{GorPosTiu11}, the $\med$ rule preserves propagation paths, and therefore, if we permute $\med$ above $\prdia$ or $\prbox$, then the rule may still be applied afterward.
 
 Let us now argue the hp-admissibility of $\ctr$ by induction on the height of the given proof. The base cases are simple as applying $\ctr$ to $\id$ or $\botin$ yields another instance of the rule, thus showing that the conclusion is provable without the use of $\ctr$. For the inductive step, if neither contraction formula $A^{\inp}$ is principal in the last rule applied above $\ctr$, then we may freely permute $\ctr$ above the rule instance. On the other hand, if one of the contraction formulae $A^{\inp}$ is principal in the rule $(r)$ applied above $\ctr$, then we use claim 2, claim 3, or the hp-admissibility of $\med$, along with the inductive hypothesis to resolve the case. For instance, if $(r)$ is the rule $\xdiain$, then our proof is as shown below left. The desired conclusion may be derived by invoking the hp-admissibility of $\xdiain$, applying the hp-admissibility of $\med$, and then applying IH, followed by an application of $\xdiain$.
\begin{center}
\begin{tabular}{c @{\hskip 1em} c @{\hskip 1em} c}
\AxiomC{$\ns \sbl (\charx) \bl A^{\inp} \br, \xdia A^{\inp} \sbr$}
\RightLabel{$\xdiain$}
\UnaryInfC{$\ns \sbl \xdia A^{\inp}, \xdia A^{\inp} \sbr$}
\RightLabel{$\ctr$}
\UnaryInfC{$\ns \sbl \xdia A^{\inp} \sbr$}
\DisplayProof

&

$\leadsto$

&

\AxiomC{$\ns \sbl (\charx) \bl A^{\inp} \br, \xdia A^{\inp} \sbr$}
\RightLabel{\thm~\ref{thm:nested-properties}-(2)}
\UnaryInfC{$\ns \sbl (\charx) \bl A^{\inp} \br, (\charx) \bl A^{\inp} \br \sbr$}
\RightLabel{$\med$}
\UnaryInfC{$\ns \sbl (\charx) \bl A^{\inp}, A^{\inp} \br \sbr$}
\RightLabel{IH}
\UnaryInfC{$\ns \sbl (\charx) \bl A^{\inp} \br \sbr$}
\RightLabel{$\xdiain$}
\UnaryInfC{$\ns \sbl \xdia A^{\inp} \sbr$}
\DisplayProof
\end{tabular}
\end{center}
\end{proof}

%\begin{proposition}
%If $\axs \subseteq \{\Box A \iimp \dia A\} \cup \{(\dia^{n} \Box A \iimp \Box^{k} A) \land (\dia^{k} A \iimp \Box^{n} \dia A) \ | \ n,k \in \mathbb{N}\}$, then the logic $\ikt(\axs)$ is conservative over $\ik(\axs)$.
%\end{proposition}

%Undecidability
\section{Properties of Intuitionistic Grammar Logics}\label{sec:(un)decid}

%Buss p. 67 for Double-Negation Translation for Reduction
%Can define K_m(\albet,\axs)$ as set of all theorems provable from the given labeled calculus
%Use refined labeled calculus because it is easier to explain that we put everything in antecedent, take the ddn, and negate it in lemma below

 We now put our refined labeled and nested systems to use, proving that intuitionistic grammar logics satisfy a certain collection of properties. We first employ our nested systems in establishing the conservativity of intuitionistic grammar logics over their (mono-)modal restrictions (defined below). In the second subsection, we show that it is undecidable to check if a formula is valid in an arbitrary intuitionistic grammar logic by means of a proof-theoretic reduction from the validity problem for classical context-free grammar logics (which is known to be undecidable~\cite{BalGioMar98}). In the third and final section, we recognize that validity can be decided for \emph{simple} intuitionistic grammar logics, which are defined by restricting the $\ipa$'s that may occur as axioms. In the latter two subsections, we make use of our refined labeled systems as the syntax of such systems is better suited for our purposes.

\subsection{Conservativity}

 Each intuitionistic modal logic exists as a mono-modal fragment of an intuitionistic grammar logic. In~\cite{Lyo21a}, (cut-free) nested sequent systems were provided for an extensive class of intuitionistic modal logics, which can be seen as restricted variants of the nested systems presented in the previous section. We leverage this fact to establish the conservativity of intuitionistic grammar logics over their modal counterparts. Toward this end, we first define the class of intuitionistic modal logics from~\cite{Lyo21a}, and subsequently discuss the fundamental concepts required to state our conservativity result, ending the section with a proof thereof.
 
\begin{definition}[Intuitionistic Modal Logics] We define the language $\mathcal{L}$ to be the set of all formulae generated via the following grammar in BNF:
$$
A ::= p \ | \ \bot \ | \ A \lor A \ | \ A \land A \ | \ A \iimp A \ | \ \dia A \ | \ \Box A
$$
 where $p$ ranges over the set $\prop$ of propositional atoms. We define the base intuitionistic modal logic $\ik$ to be the smallest set of formulae closed under substitutions of the following axioms and applications of the following inference rules.
\begin{multicols}{2}
\begin{itemize}

\item[A0] Any set of axioms for propositional intuitionistic logic

\item[A1] $\Box (A \iimp B) \iimp (\Box A \iimp \Box B)$

\item[A2] $\Box (A \iimp B) \iimp (\dia A \iimp \dia B)$

\item[A3] $\inot \dia \bot$

\item[A4] $\dia (A \lor B) \iimp (\dia A \lor \dia B)$

\item[A5] $(\dia A \iimp \Box B) \iimp \Box (A \iimp B)$

\item[R0] \AxiomC{$A$}\AxiomC{$A \iimp B$}\RightLabel{(mp)}\BinaryInfC{$B$}\DisplayProof

\item[R1] \AxiomC{$A$}\RightLabel{(nec)}\UnaryInfC{$\Box A$}\DisplayProof

\end{itemize}
\end{multicols}
\noindent
We also consider extensions of $\ik$ with sets $\axsii$ of the following axioms, where $n,k \in \mathbb{N}$. We refer to \emph{\axd} as the \emph{seriality axiom} and to each \emph{\axhsl} as a \emph{Horn-Scott-Lemmon axiom}.
$$
\emph{\axd} \ : \ \Box A \iimp \dia A \qquad \emph{\axhsl} \ : \ (\dia^{n} \Box A \iimp \Box^{k} A) \land (\dia^{k} A \iimp \Box^{n} \dia A)
$$
We define the intuitionistic modal logic $\ikam$ to be the smallest set of formulae closed under substitutions of the axioms A0--A5 and $\axsii$, and closed under the inference rules R0 and R1.
\end{definition}

 For a given alphabet $\albet$, we can encode the language $\mathcal{L}$ in the language $\lang{\albet}$ by identifying the modalities $\xdia = \dia$ and $\xbox = \Box$ for a fixed character $\charx \in \albet$. We can then view the language $\mathcal{L}$ as the subset of $\lang{\albet}$ containing only \emph{$\charx$-formulae}, i.e. formulae from $\lang{\albet}$ that only use the modalities $\langle \charx \rangle$ and $[\charx ]$. Similarly, we can view each logic $\ikam$ as a subset of $\lang{\albet}$ by identifying each formula $A \in \ikam$ with the formula $B \in \lang{\albet}$ obtained by replacing every $\dia$ and $\Box$ in $A$ by $\xdia$ and $\xbox$, respectively. For the remainder of the section we view $\mathcal{L}$ and $\ikam$ in the manner just described. For each intuitionistic modal logic, we can then define a corresponding nested calculus as follows:
 
\begin{definition}[$\nikam$~\cite{Lyo21a}] We define the nested calculus $\nikam$ for $\ikam$ to be the set of rules $\id$, $\botin$, $\disin$, $\disout$, $\conin$, $\conout$, $\iimpin$, $\iimpout$, $\xdial$, $\xboxr$, $\prdia$, and $\prbox$, which also contains $\ddr$ \iffi $\axd \in \axsii$. We define the grammar $\g{\axsii}$ used in the $\prdia$ and $\prbox$ rules as: $(\xdia \pto \xdiac^{n} \cate \xdia^{k}),$ $(\xdiac \pto \xdiac^{k} \cate \xdia^{n}) \in \g{\axsii}$ \iffi $(\dia^{n} \Box A \iimp \Box^{k} A) \land (\dia^{k} A \iimp \Box^{n} \dia A) \in \axsii$.
\end{definition}

 For each intuitionistic modal logic $\ikam$, the calculus $\nikam$ is isomorphic to and functions precisely as the nested calculus introduced for the same logic in~\cite{Lyo21a}.\footnote{In~\cite{Lyo21a}, each intuitionistic modal logic is denoted $\ika$ and its nested calculus is denoted $\nika$. We have opted to use $\axsii$ as opposed to $\axs$ in this section however to distinguish sets $\axs$ of intuitionistic path axioms and sets $\axsii$ of Horn-Scott-Lemmon axioms.} Hence, the following soundness and completeness result follows from~\cite[\thm~6]{Lyo21a}.

\begin{theorem}[Soundness and Completeness~\cite{Lyo21a}]\label{thm:modal-nested-sound-complete}
Let $\albet$ be an alphabet with $\charx \in \albet$ and $A$ be an $\charx$-formula. Then, $A$ is provable in $\nikam$ \iffi $A \in \ikam$.
\end{theorem}

 Let us now define the notion of conservativity in our setting. Afterward, we put our calculi to use and establish the conservativity relation between specific intuitionistic modal and grammar logics.

\begin{definition}[Conservative Extension] Let $\albet$ be an alphabet with $\charx \in \albet$. %We define an \emph{$\charx$-formula} to be a formula from $\lang{\albet}$ that only uses modalities indexed with $\charx$. 
 We define an intuitionistic grammar logic $\ikma$ to be an \emph{$\charx$-extension} of an intuitionistic modal logic $\ikam$ \iffi (1) $\axd \in \axsii$ \iffi $\text{D}\textsubscript{\charx} \in \axs$, and (2) the Horn-Scott-Lemmon axiom $(\dia^{n} \Box A \iimp \Box^{k} A) \land (\dia^{k} A \iimp \Box^{n} \dia A) \in \axsii$ \iffi the intuitionistic path axiom $(\diap{\conv{\charx}}^{n} \diap{\charx}^{k} A \iimp \diap{\charx} A) \land (\boxp{\charx} A \iimp \boxp{\conv{\charx}}^{n} \boxp{\charx}^{k} A) \in \axs$. Last, we say that an intuitionistic grammar logic $\ikma$ is an \emph{$\charx$-conservative extension} of an intuitionistic modal logic $\ikam$ \iffi for any $\charx$-formula $A$, if $A$ is a theorem of $\ikma$, then $A$ is a theorem of $\ikam$.
\end{definition}

 If one observes the rules employed in our nested calculi, they will find that every rule exhibits the \emph{sub-formula property}, that is, the formulae occurring within the premise of a rule are sub-formule of those occurring in the conclusion. By this observation, along with the observation that no rule changes the character $\charx$ indexing a nesting $(\charx)[\Sigma]$, it follows that our nested calculi possess the \emph{separation property}~\cite{GorPosTiu11}, summarized in the statement of the theorem below. We note that nested calculi also exhibit this property in the setting of classical (context-free) grammar and tense logics~\cite{GorPosTiu11,TiuIanGor12}.

\begin{theorem}[Separation]\label{thm:separation}
Let $\albet$ be an alphabet with $\charx \in \albet$, $A$ be an $\charx$-formula, and $\ikma$ be an $\charx$-extension of $\ikam$. Then, $A$ is provable in $\ncalc$ \iffi $A$ is provable in $\nikam$.
\end{theorem}

\begin{proof} First, we note that if $A$ is provable in $\nikam$, then $A$ is provable in $\ncalc$ as the latter calculus is an extension of the former. Thus, the backward implication is trivial. We therefore argue that if $A$ is provable in $\ncalc$, then $A$ is provable in $\nikam$. 

Let $\prf$ be a proof of $A$ in $\ncalc$. By the subformula property of $\ncalc$, it follows that if any rule $\ydiain$, $\yboxout$, $(p_{\langle \chary \rangle})$, or $(p_{[\chary]})$ is applied in $\prf$ with $\chary \in \albet \setminus \{\charx\}$, then either $\ydia$ or $\ybox$ must occur in $A$. This contradicts the fact that $A$ is an $\charx$-formula. Hence, $\prf$ only consists of rules that exist in $\nikam$. Furthermore, since $\ikma$ is an $\charx$-extension of $\ikam$, we know that $\ddr \in \ncalc$ \iffi $\ddr \in \nikam$. In addition, the side conditions on $\prdia$ and $\prbox$ will be identical in $\ncalc$ and $\nikam$ as each Horn-Scott-Lemmon axiom $(\dia^{n} \Box A \iimp \Box^{k} A) \land (\dia^{k} A \iimp \Box^{n} \dia A)$ and intuitionistic path axiom $(\diap{\conv{\charx}}^{n} \diap{\charx}^{k} A \iimp \diap{\charx} A) \land (\boxp{\charx} A \iimp \boxp{\conv{\charx}}^{n} \boxp{\charx}^{k} A)$ give rise to the same set of production rules. We may conclude that $\prf$ is a proof of $A$ in $\nikam$.
\end{proof}

\begin{corollary}
Let $\albet$ be an alphabet with $\charx \in \albet$ and $\ikma$ be an $\charx$-extension of $\ikam$. Then, $\ikma$ is an $\charx$-conservative extension of $\ikam$.
\end{corollary}

\begin{proof} Let $A$  be a theorem of $\ikma$. Then, $A$ is provable in $\ncalc$ by \thm~\ref{thm:nested-sound-complete}, from which it follows that $A$ is provable in $\nikam$ by \thm~\ref{thm:separation} above. Hence, $A$ is a theorem of $\ikam$ by \thm~\ref{thm:modal-nested-sound-complete}.
\end{proof}

\subsection{General Undecidability}

 We now establish the undecidability of the \emph{general validity problem} over the class of intuitionistic grammar logics. In other words, we show that given an arbitrary intuitionistic grammar logic $\ikm(\albet,\axs)$ and an arbitrary formula $A$, it is undecidable to determine if $\Vdash^{\albet}_{\axs} A$. We prove this result by giving a (proof-theoretic) reduction from \emph{classical context-free grammar logics}, for which it is known that determining the (in)validity of a formula for an arbitrary logic is undecidable~\cite{BalGioMar98}. Hence, we introduce classical context-free grammar logics, which we usually refer to as \emph{classical grammar logics} for simplicity, and also introduce their refined labeled systems~\cite[p.~98]{Lyo21thesis}, which are based on the nested systems of Tiu et al.~\cite{TiuIanGor12}.
 
 Classical grammar logics utilize a language similar to their intuitionistic counterparts, but where formulae are in negation normal form. We define this language, denoted $\langc{\albet}$, via the following grammar in BNF:
$$
A ::= p \ | \ \neg p \ | \ A \lor A \ | \ A \land A \ | \ \xdia A \ | \ \xbox A
$$
 where $p \in \prop$ and $x \in \albet$. Although negation is restricted to propositional atoms in $\langc{\albet}$, we can recursively define the negation $\neg A$ of an arbitrary formula $A \in \langc{\albet}$ as follows:
\begin{multicols}{2}
\begin{itemize}

\item $\neg p := \neg p$ 

\item $\neg (B \lor C) := \neg B \land \neg C$

\item $\neg \xdia B := \xbox \neg B$

\item $\neg \neg p := p$

\item $\neg (B \land C) := \neg B \lor \neg C$

\item $\neg \xbox B := \xdia \neg B$

\end{itemize}
\end{multicols}
\noindent
 We define $\bot := p \land \neg p$ for a fixed $p \in \prop$ and $A \iimp B := \neg A \lor B$. Therefore, we may assume that $\lang{\albet}$ and $\langc{\albet}$ use the same signature since every logical connective in one language occurs or can be defined in the other.
 
\begin{definition}[Classical Grammar Logic] Let $\albet$ be an alphabet. We define the base classical grammar logic $\km(\albet)$ to be the smallest set of formulae from $\langc{\albet}$ closed under substitutions of the following axioms and applications of the following inference rules. We note that we have an axiom and inference rule for each $\charx \in \albet$.
\begin{multicols}{2}
\begin{itemize}

\item[A0] Any set of axioms for propositional classical logic

\item[A1] $\xbox (A \iimp B) \iimp (\xbox A \iimp \xbox B)$

\item[A2] $A \iimp \xbox \xdiac A$

\item[R0] \AxiomC{$A$}\AxiomC{$A \iimp B$}\RightLabel{(mp)}\BinaryInfC{$B$}\DisplayProof

\item[R1] \AxiomC{$A$}\RightLabel{(nec)}\UnaryInfC{$\xbox A$}\DisplayProof

\end{itemize}
\end{multicols}
 We also consider extensions of $\km(\albet)$ with sets $\axs$ of seriality axioms $\D = \xbox A \iimp \xdia A$ and \emph{path axioms} $\langle \charx_{1} \rangle \cdots \langle \charx_{n} \rangle A \iimp \xdia A$. We define the classical grammar logic $\km(\albet,\axs)$ to  be the smallest set of formulae closed under substitutions of the axioms A0--A2 and $\axs$, and closed under the inference rules R0 and R1. Last, for an intuitionistic grammar logic $\ikma$ we define its corresponding classical grammar logic to be $\km(\albet',\axs')$ where (1) $\albet = \albet'$, (2) $\D \in \axs$ \iffi $\D \in \axs'$, and (3) $(\diap{\charx_{1}} \cdots \diap{\charx_{n}} A \iimp \diap{\charx} A) \land (\boxp{\charx} A \iimp \boxp{\charx_{1}} \cdots \boxp{\charx_{n}} A) \in \axs$ \iffi $\langle \charx_{1} \rangle \cdots \langle \charx_{n} \rangle A \iimp \xdia A \in \axs'$. For the remainder of the section, when we refer to $\ikma$ and $\km(\albet,\axs)$, we assume that the latter is the corresponding classical grammar logic of the former.
 \end{definition}
 
 For those interested in the semantics of classical (context-free) grammar logics, consult~\cite{BalGioMar98,CerPen88}.
 
 The refined labeled systems for classical grammar logics are displayed in \fig~\ref{fig:classical-calculus}, and employ labeled sequents of the form $\rel \sar \Gamma$, where $\rel$ is a multiset of relational atoms and $\Gamma$ is a multiset of labeled formulae of the form $w : A$ with $w \in \lab$ and $A \in \langc{\albet}$. We define $\neg \Gamma = w_{1} : \neg A_{1}, \ldots, w_{n} : \neg A_{n}$ for $\Gamma = w_{1} : A_{1}, \ldots, w_{n} : A_{n}$. Furthermore, we remark that the weakening right rule $\wkr$ and classical cut rule $\ccut$ (shown in \fig~\ref{fig:classical-calculus}) are admissible in each calculus $\rccalc$~\cite[\cor~1]{Lyo21thesis}. We also define the grammar $\g{\axs}$ that parameterizes the $\xdiaru$ rule as:
\begin{center}
 $(\xdia \pto \langle \charx_{1} \rangle \cdots \langle \charx_{n} \rangle), (\xdiac \pto \langle \conv{\charx}_{n} \rangle \cdots \langle \conv{\charx}_{1} \rangle) \in \g{\axs}$ \iffi $\langle \charx_{1} \rangle \cdots \langle \charx_{n} \rangle A \iimp \xdia A \in \axs$. 
\end{center}
 Observe that if $\km(\albet,\axs)$ corresponds to $\ikma$, then both logics generate the same grammar $\g{\axs}$ (see \dfn~\ref{def:grammar}).
 
 %We define each classical grammar logic $\km(\albet,\axs)$ syntactically as the set of all formulae $A \in \langc{\albet}$ such that $\sar w : A$ is provable in $\rccalc$. For those interested in the semantics of classical grammar logics, consult del Cerro and Penttonen~\cite{CerPen88}.

\begin{figure}[t]
\noindent\hrule

\begin{center}
\begin{tabular}{c @{\hskip 1em} c @{\hskip 1em} c} % @{\hskip 1em} c}
\AxiomC{}
\RightLabel{$\idru$}
\UnaryInfC{$\rel \sar \Gamma, w : p, w : \neg p$}
\DisplayProof

&

\AxiomC{$\rel \sar \Gamma, w : A, w : B$}
\RightLabel{$\disru$}
\UnaryInfC{$\rel \sar \Gamma, w : A \lor B$}
\DisplayProof

&

\AxiomC{$\rel, w R_{\charx} u \sar \Gamma$}
\RightLabel{$\ddr^{\dag}$}
\UnaryInfC{$\rel \sar \Gamma$}
\DisplayProof
\end{tabular}
\end{center}

\begin{center}
\begin{tabular}{c @{\hskip 1em} c @{\hskip 1em} c}
\AxiomC{$\rel \sar \Gamma, w : A$}
\AxiomC{$\rel \sar \Gamma, w : B$}
\RightLabel{$\conru$}
\BinaryInfC{$\rel \sar \Gamma, w : A \land B$}
\DisplayProof

&

\AxiomC{$\rel, w R_{\charx} u \sar \Gamma, u : A$}
\RightLabel{$\xboxru^{\dag}$}
\UnaryInfC{$\rel \sar \Gamma, w : \xbox A$}
\DisplayProof
\end{tabular}
\end{center}

\begin{center}
\AxiomC{$\rel \sar \Gamma, w : \xdia A, u : A$}
\RightLabel{$\xdiaru\textit{ only if }\exists \ppath(w,u) \in \prgr{\rel} (\stra_{\ppath}(w,u) \in \glang(\charx))$}
\UnaryInfC{$\rel \sar \Gamma, w : \xdia A$}
\DisplayProof
\end{center}

\begin{center}
\begin{tabular}{c @{\hskip 1em} c}
\AxiomC{$\rel \sar \Gamma$}
\RightLabel{$\wkr$}
\UnaryInfC{$\rel \sar \Gamma, w : A$}
\DisplayProof

&

\AxiomC{$\rel \sar \Gamma, w : A$}
\AxiomC{$\rel \sar \Gamma, w : \neg A$}
\RightLabel{$\ccut$}
\BinaryInfC{$\rel \sar \Gamma$}
\DisplayProof
\end{tabular}
\end{center}

\hrule
\caption{The refined labeled calculus $\rccalc$ for the classical grammar logic $\km(\albet,\axs)$~\cite[p.~98]{Lyo21thesis}. We have a copy of $\xdiaru$ and $\xboxru$ for each $\charx \in \albet$, and $\ddr$ occurs in a calculus $\rccalc$ \iffi $\D \in \axs$. The side condition $\dag$ states the rule is applicable only if $u$ is fresh.}
\label{fig:classical-calculus}
\end{figure}

Our reduction relies on a variant of the well-known \emph{double-negation translation}, attributed to G\"odel, Gentzen, and Kolmogorov~\cite{Bus98}. We define a modal version of the translation below, and utilize it in a sequence of subsequent lemmata that are ultimately used to confirm our undecidability result.

\begin{definition}[Double-Negation Translation] We recursively define the \emph{double-negation translation} over the set of formulae in $\lang{\albet} \cup \langc{\albet}$ as follows:
\begin{multicols}{2}
\begin{itemize}

\item $\ddn{p} = \neg \neg p$

\item $\ddn{\neg p} = \neg \neg \neg p$

\item $\ddn{\bot} = \neg \neg \bot$

\item $\ddn{(A \lor B)} = \neg(\neg\ddn{A} \land \neg\ddn{B})$

\item $\ddn{(A \land B)} = \ddn{A} \land \ddn{B}$

\item $\ddn{(A \iimp B)} = \ddn{A} \iimp \ddn{B}$

\item $\ddn{(\xdia A)} = \neg \xbox \neg \ddn{A}$

\item $\ddn{(\xbox A)} = \xbox \ddn{A}$

\end{itemize}
\end{multicols}
\end{definition}

\begin{lemma}\label{lem:double-neg-implies-ddn-trans}
For any $A \in \lang{\albet}$, $w : \neg \neg \ddn{A} \sar w : \ddn{A}$ is provable in $\rcalc$.
\end{lemma}

\begin{proof} The lemma is shown by induction on the complexity of $A$. We show the case where $A$ is of the form $\xbox B$, and omit the remaining cases as they are simple or similar. In the proof below, we invoke the admissibility of $\cut$ in $\rcalc$ (see \cor~\ref{cor:properties-rcalc}) and note that the top sequent in $\prf_{2}$ is provable by \lem~\ref{lem:generalized-id} and \thm~\ref{thm:calc-to-rcalc}.
\begin{flushleft}
\begin{tabular}{c @{\hskip 1em} c @{\hskip 1em} c}
$\prf_{1}$

&

$=$

&

\AxiomC{}
\RightLabel{$\botl$}
\UnaryInfC{$wR_{\charx}u, w : \neg \neg \xbox \ddn{B}, u : \neg \ddn{B}, w : \bot \sar u :\bot$}
\DisplayProof
\end{tabular}
\end{flushleft}

\begin{flushleft}
\begin{tabular}{c @{\hskip 1em} c @{\hskip 1em} c}
$\prf_{2}$

&

$=$

&

\AxiomC{$wR_{\charx}u, w : \neg \neg \xbox \ddn{B}, u : \neg \ddn{B}, w : \xbox \ddn{B}, u : \ddn{B} \sar w : \bot$}
\RightLabel{$\xboxl$}
\UnaryInfC{$wR_{\charx}u, w : \neg \neg \xbox \ddn{B}, u : \neg \ddn{B}, w : \xbox \ddn{B} \sar w : \bot$}
\RightLabel{$\iimpr$}
\UnaryInfC{$wR_{\charx}u, w : \neg \neg \xbox \ddn{B}, u : \neg \ddn{B} \sar w : \neg \xbox \ddn{B}$}
\DisplayProof
\end{tabular}
\end{flushleft}

\begin{center}
\AxiomC{$\prf_{1}$}

\AxiomC{$\prf_{2}$}

\BinaryInfC{$wR_{\charx}u, w : \neg \neg \xbox \ddn{B}, u : \neg \ddn{B} \sar u : \bot$}
\RightLabel{$\iimpr$}
\UnaryInfC{$wR_{\charx}u, w : \neg \neg \xbox \ddn{B} \sar u : \neg \neg \ddn{B}$}

\AxiomC{$u : \neg \neg \ddn{B} \sar u :  \ddn{B}$}
\RightLabel{$\cut$}
\BinaryInfC{$wR_{\charx}u, w : \neg \neg \xbox \ddn{B} \sar u : \ddn{B}$}
\RightLabel{$\xboxr$}
\UnaryInfC{$w : \neg \neg \xbox \ddn{B} \sar w : \xbox \ddn{B}$}
\DisplayProof
\end{center}
\end{proof}

\begin{lemma}\label{lem:dn-classical-to-int}
If $\rel \sar \Gamma$ is provable in $\rccalc$, then $\rel, \neg \ddn{\Gamma} \sar w : \bot$ is provable in $\rcalc$.
\end{lemma}

\begin{proof} We prove the result by induction on the height of the proof of $\rel \sar \Gamma$ in $\rccalc$. Below, we use $\neg^{n} A$ to denote $A$ prefixed with $n \in \mathbb{N}$ negation symbols.

\textit{Base case.} We may transform an instance of $\id$ in $\rccalc$ into a proof of the desired labeled sequent in $\rcalc$. Recall that in the intuitionistic setting $\neg A = A \iimp \bot$, thus explaining the $\iimpl$ inference applied in the output proof, whose left premise is provable by \lem~\ref{lem:generalized-id} and \thm~\ref{thm:calc-to-rcalc}.
\begin{flushleft}
\begin{tabular}{c c}
\AxiomC{}
\RightLabel{$\id$}
\UnaryInfC{$\rel \sar \Gamma, u : p, u : \neg p$}
\DisplayProof

&

$\leadsto$
\end{tabular}
\end{flushleft}
\begin{flushright}
\AxiomC{$\rel, \neg \ddn{\Gamma}, u : \neg^{3} p, u : \neg^{4} p \sar u : \neg^{3} p$}
\AxiomC{}
\RightLabel{$\botl$}
\UnaryInfC{$\rel, \neg \ddn{\Gamma}, u : \neg^{3} p, u : \bot \sar w : \bot$}
\RightLabel{$\iimpl$}
\BinaryInfC{$\rel, \neg \ddn{\Gamma}, u : \neg^{3} p, u : \neg^{4} p \sar w : \bot$}
\DisplayProof
\end{flushright}

\textit{Inductive step.} We show the $\xboxru$ case as the remaining cases are shown similarly. Below, the top sequent in $\prf$ is provable by IH and the right premise of $\cut$ is provable by \lem~\ref{lem:double-neg-implies-ddn-trans}.
\begin{flushleft}
\begin{tabular}{c @{\hskip 1em} c @{\hskip 1em} c}
$\prf$

&

$=$

&

\AxiomC{$\rel, u R_{\charx} v, \neg \ddn{\Gamma}, v : \neg \ddn{A} \sar v : \bot$}
\RightLabel{$\wkl$}
\UnaryInfC{$\rel, u R_{\charx} v, \neg \ddn{\Gamma}, u : \neg \xbox \ddn{A}, v : \neg \ddn{A} \sar v : \bot$}
\RightLabel{$\iimpr$}
\UnaryInfC{$\rel, u R_{\charx} v, \neg \ddn{\Gamma}, u : \neg \xbox \ddn{A} \sar v : \neg \neg \ddn{A}$}
\DisplayProof
\end{tabular}
\end{flushleft}

\begin{center}
\AxiomC{$\prf$}
\AxiomC{$v : \neg \neg \ddn{A} \sar v : \ddn{A}$}
\RightLabel{$\cut$}
\BinaryInfC{$\rel, u R_{\charx} v, \neg \ddn{\Gamma}, u : \neg \xbox \ddn{A} \sar v : \ddn{A}$}
\RightLabel{$\xboxr$}
\UnaryInfC{$\rel, \neg \ddn{\Gamma}, u : \neg \xbox \ddn{A} \sar u : \xbox \ddn{A}$}

\AxiomC{}
\UnaryInfC{$\rel, \neg \ddn{\Gamma},  u : \bot \sar w : \bot$}

\RightLabel{$\iimpl$}
\BinaryInfC{$\rel, \neg \ddn{\Gamma}, u : \neg \xbox \ddn{A} \sar w : \bot$}
\DisplayProof
\end{center}
\end{proof}

\begin{lemma}\label{lem:dn-implies-formula}
For any $A \in \langc{\albet}$, $\sar w : \neg \ddn{A}, w : A$ is provable in $\rccalc$.
\end{lemma}

\begin{proof} The result can be shown by a straightforward induction on the complexity of $A$.
\end{proof}

\begin{lemma}\label{lem:dn-int-to-classical}
If $\rel, \Gamma \sar w : A$ is provable in $\rcalc$, then $\rel \sar \neg \Gamma, w : A$ is provable in $\rccalc$.
\end{lemma}

\begin{proof} We prove the lemma by induction on the height of the given proof.

\textit{Base case.} We show the $\botl$ case as the $\id$ case is trivial. Observe that in the output proof we use the definition of negation and $\bot$ as $p \land \neg p$ in $\langc{\albet}$ to show how the proof is translated from $\rcalc$ to $\rccalc$.
\begin{center}
\begin{tabular}{c @{\hskip 1em} c @{\hskip 1em} c}
\AxiomC{}
\RightLabel{$\botl$}
\UnaryInfC{$\rel, \Gamma, w : \bot \sar u : A$}
\DisplayProof

&

$\leadsto$

&

\AxiomC{}
\RightLabel{$\id$}
\UnaryInfC{$\rel \sar \neg \Gamma, w : \neg p, w : p, w : A$}
\RightLabel{$\disru$}
\UnaryInfC{$\rel \sar \neg \Gamma, w : \neg p \lor p, w : A$}
\RightLabel{$=$}
\UnaryInfC{$\rel \sar \neg \Gamma, w : \neg \bot, w : A$}
\DisplayProof
\end{tabular}
\end{center}

\textit{Inductive step.} We show the $\prdia$ and $\xdial$ cases as the remaining are similar. For the $\prdia$ case, we invoke the admissibility of $\wkr$ in $\rccalc$ and note that the side condition of $\prdia$ holds in the output proof as $\rel$ is unaffected in the translation.
\begin{center}
\begin{tabular}{c @{\hskip 1em} c @{\hskip 1em} c}
\AxiomC{$\rel, \Gamma \sar u : A$}
\RightLabel{$\prdia$}
\UnaryInfC{$\rel, \Gamma \sar w : \xdia A$}
\DisplayProof

&

$\leadsto$

&

\AxiomC{$\rel \sar \neg \Gamma, u : A$}
\RightLabel{$\wkr$}
\UnaryInfC{$\rel \sar \neg \Gamma, w : \xdia A, u : A$}
\RightLabel{$\xdiaru$}
\UnaryInfC{$\rel \sar \neg \Gamma, w : \xdia A$}
\DisplayProof
\end{tabular}
\end{center}

 In the $\xdial$ case, we apply the definition of negation in $\langc{\albet}$ to obtain the desired result.
\begin{center}
\begin{tabular}{c @{\hskip 1em} c @{\hskip 1em} c}
\AxiomC{$\rel, w R_{\charx} u, u : A, \Gamma \sar v : B$}
\RightLabel{$\xdial$}
\UnaryInfC{$\rel, \Gamma, w : \xdia A \sar v : B$}
\DisplayProof

&

$\leadsto$

&

\AxiomC{$\rel, w R_{\charx} u \sar \neg \Gamma, u : \neg A, v : B$}
\RightLabel{$\xboxru$}
\UnaryInfC{$\rel \sar \neg \Gamma, u : \xbox \neg A, v : B$}
\RightLabel{$=$}
\UnaryInfC{$\rel \sar \neg \Gamma, u : \neg \xdia A, v : B$}
\DisplayProof
\end{tabular}
\end{center}
\end{proof}

\begin{theorem}\label{thm:classical-intuitionistic-reduction}
$A \in \mathsf{K_{m}}(\albet,\axs)$ \iffi $\ddn{A} \in \ikm(\albet,\axs)$.
\end{theorem}

\begin{proof} For the forward direction, we assume that $A \in \mathsf{K_{m}}(\albet,\axs)$, which implies that $\sar w : A$ is provable in $\rccalc$ by completeness. By \lem~\ref{lem:dn-classical-to-int}, it follows that $w : \neg \ddn{A} \sar u : \bot$ is provable in $\rcalc$, and so, $w : \neg \ddn{A} \sar w : \bot$ is provable in $\rcalc$ as $u : \bot \sar w : \bot$ is an instance of $\botl$ and $\cut$ is admissible. Therefore, $\sar w : \neg \neg \ddn{A}$ is provable in $\rcalc$, meaning $\sar w : \ddn{A}$ is provable in $\rcalc$ by \lem ~\ref{lem:double-neg-implies-ddn-trans} and the admissibility of $\cut$.

For the backward direction, we assume that $\ddn{A} \in \ikm(\albet,\axs)$, which implies that $\sar w : \ddn{A}$ is provable in $\rcalc$ by completeness. By \lem~\ref{lem:dn-int-to-classical}, $\sar w : \ddn{A}$ is provable in $\rccalc$, showing that $\sar w : A$ is provable in $\rccalc$ since $\sar w : \neg \ddn{A}, w : A$ is provable in $\rccalc$ by \lem~\ref{lem:dn-implies-formula} and $\ccut$ is admissible.
\end{proof}

\begin{corollary}
It is undecidable whether a formula of an arbitrarily given intuitionistic grammar logic is a theorem.
\end{corollary}

\begin{proof} Follows from \thm~\ref{thm:classical-intuitionistic-reduction} and the undecidability of the validity problem for classical context-free grammar logics~\cite{BalGioMar98}.
\end{proof}

\subsection{Decidability}

 Despite the undecidability of the general validity problem, certain subclasses of intuitionistic grammar logics remain decidable. In this section, we identify such a subclass by relating it to the class of intuitionistic modal logics proven decidable by Simpson~\cite[\sect~7.3]{Sim94}. % , arguing the decidability of the validity problem by leveraging a technique due to Simpson~\cite[\sect~7.3]{Sim94}. 
 Simpson established the decidability of $\ik$ extended with combinations of the seriality, reflexivity, and symmetry axioms by demonstrating that every validity of such a logic has a finite number of proofs within a certain form. Thus, decidability of a formula is obtained by searching through this finite set, and if a proof is found, then the formula is known to be valid, and if a proof is not found, then the formula is known to be invalid. 
 
 With basic modifications and extensions, Simpson's decidability method may be straightforwardly adapted to our setting of intuitionistic grammar logics. For an arbitrary alphabet $\albet$, we obtain the decidability of all intuitionistic grammar logics $\ikma$ such that 
$$
\axs \subseteq\{(\ydia^{n} A \iimp \xdia A) \land (\xbox A \iimp \ybox^{n} A) \ \vert \ \charx, \chary \in \albet, 0 \leq n \leq 1\} \cup \{\D \ \vert \ \charx \in \albet\}.
$$
 %For the remainder of the section, we assume that $\axs$ is defined as above and 
 We refer to all such intuitionistic grammar logics as \emph{simple}. Moreover, we explicitly present the propagation rules in \fig~\ref{fig:propagation-rules-simple} that appear in the refined labeled systems for simple intuitionistic grammar logics. For each axiom of the form $(A \iimp \xdia A) \land (\xbox A \iimp A)$, $\rcalc$ includes the rules $\refi$ and $\refii$, and for each axiom of the form $(\ydia A \iimp \xdia A) \land (\xbox A \iimp \ybox A)$, $\rcalc$ includes the $\prdiai$, $\prdiaii$, $\prboxi$, and $\prboxii$ rules. For any simple intuitionistic grammar logic $\ikma$, all propagation rules of $\rcalc$ take the form of a (pair of) rule(s) in \fig~\ref{fig:propagation-rules-simple}.
 
 \begin{figure}[t]
\noindent\hrule

\begin{center}
\begin{tabular}{c c}
\AxiomC{$\rel, \Gamma \sar w : A$}
\RightLabel{$\refi$}
\UnaryInfC{$\rel, \Gamma \sar w : \xdia A$}
\DisplayProof

&

\AxiomC{$\rel, \Gamma, w : \xbox A, w : A \sar v : B$}
\RightLabel{$\refii$}
\UnaryInfC{$\rel, \Gamma, w : \xbox A \sar v : B$}
\DisplayProof
\end{tabular}
\end{center}

\begin{center}
\begin{tabular}{c c}
\AxiomC{$\rel, w R_{\chary} u, \Gamma \sar u : A$}
\RightLabel{$\prdiai$}
\UnaryInfC{$\rel, w R_{\chary} u, \Gamma \sar w : \xdia A$}
\DisplayProof

&

\AxiomC{$\rel, w R_{\chary} u, \Gamma, w : \xbox A, u : A \sar v : B$}
\RightLabel{$\prboxi$}
\UnaryInfC{$\rel, w R_{\chary} u, \Gamma, w : \xbox A \sar v : B$}
\DisplayProof
\end{tabular}
\end{center}
\begin{center}
%\resizebox{\columnwidth}{!}{
\begin{tabular}{c c}
\AxiomC{$\rel, u R_{\conv{\chary}} w, \Gamma \sar u : A$}
\RightLabel{$\prdiaii$}
\UnaryInfC{$\rel, u R_{\conv{\chary}} w, \Gamma \sar w : \xdia A$}
\DisplayProof

&

\AxiomC{$\rel, u R_{\conv{\chary}} w, \Gamma, w : \xbox A, u : A \sar v : B$}
\RightLabel{$\prboxii$}
\UnaryInfC{$\rel, u R_{\conv{\chary}} w, \Gamma, w : \xbox A \sar v : B$}
\DisplayProof
\end{tabular}
\end{center}

\hrule
\caption{Propagation rules for simple intuitionistic grammar logics.}
\label{fig:propagation-rules-simple}
\end{figure}
 
 We note that the proof of decidability for this class of logics is \emph{almost identical} to Simpson's proof with the exception that basic modifications are required to handle the use of multiple modalities and converse modalities. By reading through Simpson's proof while taking our simple intuitionistic grammar logics into account, it is straightforward to verify the following theorem:
 
\begin{theorem}\label{thm:decid-simple-int-gram-logics}
The validity problem for simple intuitionistic grammar logics is decidable.
\end{theorem}

 Due to the ease with which this proof is adapted to our setting, we omit it from the main text and provide a sketch of the proof in the appendix for the interested reader.

%Conclude
\section{Concluding Remarks}\label{sec:conclusion}

 In this paper we provided an in-depth study of intuitionistic grammar logics and their associated proof theory. We supplied this class of logics with labeled calculi, obtained from each logic's semantics, and showed that certain hp-admissibility and hp-invertibility results obtained for these systems with a proof of syntactic cut-elimination. Subsequently, we showed how to apply the structural refinement methodology to derive deductively equivalent `refined' labeled calculi, from which nested calculi were extracted. Moreover, these derivative systems were shown to exhibit favorable admissibility and invertibility properties just as their parent labeled systems. We then concluded by employing our refined labeled and nested systems in the establishment of conservativity and (un)decidability results.
 
 A few interesting open problems still remain. For instance, certain inference rules in our provided systems are not (fully) invertible (e.g. the $\iimpl$ rule), giving rise to the question of if variants of these systems can be produced which admit the complete invertibility of every inference rule. In relation to intuitionistic grammar logics more specifically, it could be worthwhile to investigate if such logics possess the Craig interpolation property, by adapting proof-theoretic methods of interpolation to our setting~\cite{FitKuz15,LyoTiuGorClo20}. Finally, as decidability was only recognized to hold for a relatively small class of intuitionistic grammar logics, it would be of interest to determine decidability for larger classes of such logics.\\

%%%FUNDING
\noindent
\funding{This work was supported by the European Research Council Consolidator Grant [771779].} %\textit{A Grand Unified Theory of Decidability in Logic-Based Knowledge Representation} (DeciGUT).}



\bibliographystyle{elsarticle-num} 
\bibliography{bibliography}


%\begin{thebibliography}{00}


%\bibitem{}

%\end{thebibliography}

%%%Appendix
\appendix

\section{Decidability Proof}\label{app:decid} %: Simple Intuitionistic Grammar Logics}

 Recall that for a simple intuitionistic grammar logic $\ikma$ the set $\axs$ of axioms is defined as follows:
$$
\axs \subseteq\{(\ydia^{n} A \iimp \xdia A) \land (\xbox A \iimp \ybox^{n} A) \ \vert \ \charx, \chary \in \albet, 0 \leq n \leq 1\} \cup \{\D \ \vert \ \charx \in \albet\}.
$$
 For the remainder of the appendix, we assume that $\axs$ is defined as above. Furthermore, recall that for a simple intuitionistic grammar logic, the propagation rules from \fig~\ref{fig:propagation-rules-simple} are used in the refined labeled calculus $\rcalc$.

 We will repeat definitions and lemmata due to Simpson to give the reader intuition regarding the proof of decidability, but will omit various details as they can be found in Simpson's PhD thesis~\cite[\sect~7.3]{Sim94}.

 We recursively define the modal depth $\md{A}$ of a formula $A$ as follows: $\md{p} = \md{\bot} = 0$, $\md{B \odot C} = \max\{\md{B},\md{C}\}$ for $\odot \in \{\lor,\land,\iimp\}$, and $\md{\triangledown B} = \md{B} + 1$ for $\triangledown \in \{\xdia, \xbox \ \vert \ \charx \in \albet\}$. We also recursively define the set $\sufo{A}$ of subformulae of a formula $A$ as follows: $\sufo{p} = \{p\}$, $\sufo{\bot} = \bot$, $\sufo{B \odot C} = \{B \odot C\} \cup \sufo{B} \cup \sufo{C}$ for $\odot \in \{\lor,\land,\iimp\}$, and $\sufo{\triangledown B} = \{\triangledown B\} \cup \sufo{B}$ for $\triangledown \in \{\xdia, \xbox \ \vert \ \charx \in \albet\}$. For a multiset $\Gamma$ of labeled formulae, we define 
$$
\md{\Gamma} = \max\{\md{A} \ \vert \ w : A \in \Gamma \} \text{ and } \sufo{\Gamma} = \bigcup_{w : A \in \Gamma} \sufo{A}.
$$ 
 For a multiset $\rel$ of relational atoms, we define an \emph{$\rel$-extension} as follows: (1) $\rel$ is an $\rel$-extension, and (2) if $\rel'$ is an $\rel$-extension and $u \not \in \lab(\rel)$, then $\rel' \cup \{wRu\}$ is an $\rel$-extension. If $\rel'$ is an $\rel$-extension, then a label $u \in \lab(\rel')$ has depth $n \geq 0$ \iffi there exists a sequence $w_{0}R_{\charx_{1}}w_{1}, \ldots, w_{n-1}R_{\charx_{n}}w_{n} \in \rel'$ such that $w_{0} \in \lab(\rel)$, $w_{i} \in \lab(\rel') \setminus \lab(\rel)$, and $w_{n} = u$. The depth of a label is well-defined since every label has only one such accessing sequence by the definition of an $\rel$-extension. We define an $\rel$-extension $\rel'$ to be \emph{bounded} relative to a labeled sequent $\rel, \Gamma \sar w : A$ \iffi the depth of every label in $\rel'$ is less than or equal to $\md{\Gamma, w : A}$, and we define the \emph{bounded-restriction} of $\rel'$ relative to a labeled sequent $\rel, \Gamma \sar w : A$ to be the set of all relational atoms from $\rel'$ whose labels have a depth less than or equal to $\md{\Gamma, w : A}$. A labeled sequent $\rel', \Gamma' \sar w' : A'$ is defined to be \emph{bounded} relative to a labeled sequent $\rel, \Gamma \sar w : A$ \iffi (1) $\rel'$ is a bounded $\rel$-extension, and (2) if $u : B \in \Gamma', w' : A'$, then $B \in \sufo{\Gamma, w : A}$ and $\md{B} \leq \md{\Gamma, w : A} - n$ where $n$ is the depth of $u$ in $\rel'$. 
 
 We define a \emph{pseudo-derivation} of a labeled sequent $\rel, \Gamma \sar w : A$ in $\rcalc$ to be a derivation in $\rcalc$ from any finite collection of labeled sequents (which need not be instances of $\id$ or $\botl$), and we define its \emph{size} to be the number of inferences it contains. When a pseudo-derivation is derivable from labeled sequents that are instances of $\id$ or $\botl$, then it is a proof. Moreover, we stipulate that a pseudo-derivation $\prf$ of a labeled sequent $\rel, \Gamma \sar w : A$ is \emph{bounded} \iffi every labeled sequent in $\prf$ is bounded relative to $\rel, \Gamma \sar w : A$.

\begin{lemma} Let $\prf$ be a pseudo-derivation of $\rel, \Gamma \sar w : A$ from the labeled sequents in $\{\rel_{i}, \Gamma_{i} \sar w_{i} : A_{i} \ \vert \ 1 \leq i \leq n\}$ and let $d = \md{\Gamma, w : A}$. Then, the following holds:
\begin{enumerate}

\item For each $1 \leq i \leq n$, $\rel_{i}$ is an $\rel$-extension and if $u : B \in \Gamma_{i} \cup \{w_{i} : A_{i}\}$, then $u$ has a depth $n \leq d$, $B \in \sufo{\Gamma, w : A}$, and $\md{B} \leq d - n$.

\item There exists a bounded pseudo-derivation of $\rel, \Gamma \sar w : A$ from the labeled sequents in $\{\rel_{i}', \Gamma_{i}' \sar w_{i} : A_{i} \ | \ 1 \leq i \leq n\}$ where each $\rel_{i}'$ is the bounded-restriction of $\rel_{i}$.

\end{enumerate}
\end{lemma}

\begin{proof} The result is shown by induction on the size of the given pseudo-derivation $\prf$. We only show the $\prdiaii$ case of the inductive step as all cases are argued in an almost identical fashion as in~\cite[\lem~7.3.5]{Sim94}. 

 Suppose that $\prdiaii$ was applied at the top of $\prf$ as shown below.
\begin{center}
\AxiomC{$\rel', v R_{\conv{\chary}} u, \Gamma' \sar v : B$}
\RightLabel{$\prdiaii$}
\UnaryInfC{$\rel', v R_{\conv{\chary}} u, \Gamma' \sar u : \xdia B$}
\AxiomC{$\cdots$}
\noLine
\BinaryInfC{$\vdots$}
\noLine
\UnaryInfC{ }
\UnaryInfC{$\rel, \Gamma \sar w : A$}
\DisplayProof
\end{center}
 We know that $u$ occurs at a depth $n \leq d - 1$ because $\md{\xdia B} \leq d-1$ by the inductive hypothesis. This implies claim 1 since $v$ occurs at a depth $n-1 < n \leq d - 1$, $B \in \sufo{\Gamma, w : A}$, and $\md{B} \leq d - (n-1) \leq d - n$. For claim 2, we suppose that we have a bounded pseudo-derivation of $\rel'', \Gamma' \sar u : \xdia B$ as shown below, where $\rel''$ is the bounded-restriction of $\rel'$. Observe that a single application of $\prdiaii$ gives the desired bounded pseudo-derivation.
\begin{center}
%\AxiomC{$\rel'', \Gamma'' \sar u : B$}
%\RightLabel{$\refi$}
\AxiomC{$\rel'', v R_{\conv{\chary}} u, \Gamma' \sar u : \xdia B$}
\AxiomC{$\cdots$}
\noLine
\BinaryInfC{$\vdots$}
\noLine
\UnaryInfC{ }
\UnaryInfC{$\rel, \Gamma \sar w : A$}
\DisplayProof
\end{center}
\end{proof}

 We now define a pre-order on all labeled sequents $\rel', \Gamma' \sar w' : A'$ which are bounded relative to a given labeled sequent $\rel, \Gamma \sar w : A$. This will be used to define an equivalence relation over the set of such sequents, partitioning the set into a finite number of equivalence classes, and ultimately permitting us to restrict the number of proofs of $\rel, \Gamma \sar w : A$ considered during proof-search to a finite number. For the remainder of the appendix, we fix the labeled sequent $\lseq = \rel, \Gamma \sar w : A$ and only consider labeled sequents bounded relative to this one.
 
 Let $\lseq_{0} = \rel_{0}, \Gamma_{0} \sar w_{0} : A_{0}$ and $\lseq_{1} = \rel_{1}, \Gamma_{1} \sar w_{1} : A_{1}$ be bounded relative to $\rel, \Gamma \sar w : A$. We define a \emph{morphism} from $\lseq_{0}$ to $\lseq_{1}$ to be a function $f$ such that (1) for all $u \in \lab(\lseq)$, $f(u) = u$, (2) if $u : B \in \Gamma_{0}$, then $f(u) \in \Gamma_{1}$, (3) $f(w_{0}) = w_{1}$, and (4) if $uR_{\charx}v \in \rel_{0}$, then $f(u)R_{\charx}f(v) \in \rel_{1}$. We then define the pre-order $\preceq$ as follows: $\lseq_{0} \preceq \lseq_{1}$ \iffi $A_{0} = A_{1}$ and there exists a morphism $f$ from $\lseq_{0}$ to $\lseq_{1}$. The equivalence relation $\cong$ is then defined as: $\lseq_{0} \cong \lseq_{1}$ \iffi $\lseq_{0} \preceq \lseq_{1}$ and $\lseq_{1} \preceq \lseq_{0}$.

\begin{lemma}
 The equivalence relation $\cong$ partitions the set of labeled sequents bounded relative to $\lseq$ into a finite number of classes.
\end{lemma}

\begin{proof} The proof is similar to the proof of \cite[Proposition~7.3.6]{Sim94}, but with the exception that one must generalize the arguments of Simpson to account for multiple modalities.
\end{proof}

 We now define an \emph{irredundant} pseudo-derivation $\prf$ to be a bounded pseudo-derivation relative to a labeled sequent $\rel, \Gamma \sar w : A$ such that no two labeled sequents $\lseq_{0} = \rel_{0}, \Gamma_{0} \sar w_{0} : A_{0}$ and $\lseq_{1} = \rel_{1}, \Gamma_{1} \sar w_{1} : A_{1}$ occur in $\prf$, with the former above the latter, such that $\lseq_{0} \preceq \lseq_{1}$.

\begin{lemma} Let $\lseq = \rel, \Gamma \sar w : A$ with $\lseq_{i} = \rel_{i}, \Gamma_{i} \sar w_{i} : A_{i}$ bounded relative to $\lseq$ for $i \in \{1,2\}$. Then,
\begin{enumerate}

\item If $\lseq_{0} \preceq \lseq_{1}$ and $\lseq_{0}$ has a bounded proof of size $n$ relative to $\lseq$, then $\lseq_{1}$ has a bounded proof of size $n$ relative to $\lseq$.

\item If $\lseq$ is derivable in $\rcalc$, then $\lseq$ has an irredundant proof.

\end{enumerate}
\end{lemma}

\begin{proof} The proof of claim 1 is similar to \cite[\lem~7.3.7]{Sim94} and the proof of claim 2 is similar to \cite[Proposition~7.3.8]{Sim94}.
\end{proof}

\begin{theorem}\label{thm:proof-search}
 Let $\rcalc$ be a refined labeled calculus for a simple intuitionistic grammar logic $\ikma$. Then, it is decidable to check if $\models_{\axs}^{\albet} \rel, \Gamma \sar w : A$ for any arbitrary labeled sequent $\rel, \Gamma \sar w : A$.
\end{theorem}

\begin{proof} We decide the validity of a labeled sequent $\rel, \Gamma \sar w : A$ by searching for an irredundant proof of it in $\rcalc$. The proof-search algorithm is similar to the one given in~\cite[\sect~7.3.3]{Sim94}. First, we take the labeled sequent $\rel, \Gamma \sar w : A$ as input and check if it is an instance of $\id$ or $\botl$. If so, then we know that $\models_{\axs}^{\albet} \rel, \Gamma \sar w : A$, and if not, then we continue proof-search, searching for all irredundant proofs of size $1$, and then of size $2$, and so forth by applying all relevant rules of $\rcalc$ bottom-up. Only a finite number of bottom-up rule applications are possible at each stage modulo the choice of fresh labels in the $\xdial$, $\xboxr$, and $\ddr$ rules. However, this peculiarity can be overcome by fixing how fresh labels are chosen during proof-search. Also, note that it is computable to check irredundancy since the $\preceq$ relation is decidable on labeled sequents and it is decidable to check if a labeled sequent is bounded relative to the input $\rel, \Gamma \sar w : A$. Last, since the equivalence relation $\cong$ partitions the set of all labeled sequent bounded relative to $\rel, \Gamma \sar w : A$ into a finite number of equivalence classes, we know that eventually proof-search will terminate.
\end{proof}

\begin{customthm}{\ref{thm:decid-simple-int-gram-logics}}
The validity problem for simple intuitionistic grammar logics is decidable.
\end{customthm}

\begin{proof} Follows from \thm~\ref{thm:proof-search} above.
\end{proof}

\end{document}
\endinput
%%
%% End of file `elsarticle-template-num.tex'.